\documentclass{article}
\usepackage{ijcai20}

\usepackage{times}

\usepackage{soul}
\usepackage{url}
\usepackage[utf8]{inputenc}
\usepackage[small]{caption}
\usepackage{graphicx}
\usepackage{amsmath}
\usepackage{booktabs}
\usepackage{latexsym} 

\usepackage{amssymb}
\usepackage[inline]{enumitem}
\usepackage{thm-restate} 
\usepackage{microtype}
\usepackage{xspace}
\usepackage{amsthm}
\usepackage[hidelinks]{hyperref}
\newtheorem{theorem}{Theorem}
\newtheorem{definition}[theorem]{Definition}
\newtheorem{example}[theorem]{Example}
\newtheorem{lemma}[theorem]{Lemma}
\newtheorem{proposition}[theorem]{Proposition}
\newtheorem{corollary}[theorem]{Corollary}
\newtheorem{claim}[theorem]{Claim}
\newtheorem*{remark}{Remark}

\newcommand{\mypar}[1]{\medskip\noindent\textbf{#1.}}

\newcommand{\mn}[1]{\ensuremath{\mathsf{#1}}}

\newcommand{\p}[2]{\ensuremath{{\sf prov}_{#1}({#2})}\xspace}

\newcommand{\map}{\ensuremath{{\pi}}\xspace}
\newcommand{\pair}[2]{\ensuremath{({#1},{#2})}}
\newcommand{\polynomial}{\ensuremath{p}\xspace}

\newcommand{\semiringVariables}{\ensuremath{\NV}\xspace}

\renewcommand{\oplus}{\ensuremath{+}\xspace}
\renewcommand{\otimes}{\ensuremath{\times}\xspace}

\newcommand{\NF}{\mn{NF}}
\newcommand{\CR}{\mn{CR}}

\newcommand{\dlNames}[1]{\ensuremath{\mathsf{N_{#1}}}\xspace}

\newcommand{\NPr}{\dlNames{P}}
\newcommand{\DLLite}{\text{DL-Lite}\xspace}

\newcommand{\individuals}[1]{\mathsf{ind}(#1)}
\newcommand{\monomials}[1]{\mathsf{{\sf mon}}(#1)}
\newcommand{\auxs}[1]{\ensuremath{\mathsf{{\sf aux}}(#1)}\xspace} 
\newcommand{\roles}[1]{\mathsf{rol}(#1)}

\newcommand{\cupdot}{\mathbin{\mathaccent\cdot\cup}}
\newcommand{\variable}{\ensuremath{v}\xspace}
\newcommand{\monomial}{\ensuremath{m}\xspace}
\newcommand{\nonomial}{\ensuremath{n}\xspace}
\newcommand{\onomial}{\ensuremath{o}\xspace}
\newcommand{\query}{\ensuremath{q}\xspace}
\newcommand{\representative}[1]{\ensuremath{[{#1}]}\xspace}

\newcommand{\queryrewriting}{\ensuremath{q^\ast}\xspace}

\newcommand{\equivalentclass}[1]{\ensuremath{[[{#1}]]}\xspace}

\newcommand{\unraveling}{\ensuremath{{\Umc_\Omc}}\xspace}
\newcommand{\upath}{\ensuremath{{\sf p}}\xspace}
\newcommand{\dep}[1]{\ensuremath{{\sf dep}({#1})}\xspace}

\newcommand{\umapping}{\ensuremath{\delta}\xspace}
\newcommand{\aux}[3]{\ensuremath{d^{{#1}}_{{#2}{#3}}}\xspace}

\newcommand{\imonomial}{\ensuremath{\mu}\xspace}

\newcommand{\NC}{\ensuremath{{\sf N_C}}\xspace}
\newcommand{\NI}{\ensuremath{{\sf N_I}}\xspace}
\newcommand{\NM}{\ensuremath{{\sf N_M}}\xspace}
\newcommand{\NV}{\ensuremath{{\sf N_V}}\xspace}
\newcommand{\NMrep}{\ensuremath{{\sf N_{[M]}}}\xspace}

\newcommand{\NR}{\ensuremath{{\sf N_R}}\xspace}

\newcommand{\EL}{\ensuremath{{\cal E\!L}}\xspace}
\newcommand{\ELHr}{\ensuremath{{\cal E\!LH}^r}\xspace}
\newcommand{\ALC}{\ensuremath{{\cal ALC}}\xspace}

\newcommand{\NP}{\textsc{NP}\xspace}

\newcommand{\Amc}{\ensuremath{\mathcal{A}}\xspace}

\newcommand{\Imc}{\ensuremath{\mathcal{I}}\xspace}
\newcommand{\Jmc}{\ensuremath{\mathcal{J}}\xspace}

\newcommand{\Omc}{\ensuremath{\mathcal{O}}\xspace}

\newcommand{\Rmc}{\ensuremath{\mathcal{R}}\xspace}
\newcommand{\Smc}{\ensuremath{\mathcal{S}}\xspace}
\newcommand{\Tmc}{\ensuremath{\mathcal{T}}\xspace}
\newcommand{\Umc}{\ensuremath{\mathcal{U}}\xspace}

\usepackage[textsize=tiny,colorinlistoftodos,obeyFinal]{todonotes}

\usepackage[absolute,showboxes]{textpos}

\TPMargin{5pt}

\title{Provenance for the Description Logic \ELHr}

\author{
Camille Bourgaux$^1$\and
Ana Ozaki$^{2,4}$\and
Rafael Pe\~naloza$^3$\And
Livia Predoiu$^4$\\
\affiliations
$^1$DI ENS, ENS, CNRS, PSL University \& Inria, Paris, France\\
$^2$University of Bergen, Norway\\
$^3$University of Milano-Bicocca, Italy\\
$^4$Free University of Bozen-Bolzano, Italy\\
\emails
camille.bourgaux@ens.fr, ana.ozaki@uib.no, rafael.penaloza@unimib.it, livia.predoiu@unibz.it
}

\begin{document}

\maketitle

\begin{abstract}
We address the problem of handling provenance information in \ELHr ontologies. We consider a setting recently introduced for ontology-based data access, based on semirings and  extending classical data provenance, in which ontology axioms are annotated with provenance tokens. A consequence inherits the provenance of the axioms involved in deriving it, yielding a provenance polynomial as annotation. We analyse the semantics for the \ELHr case and show that the presence of conjunctions poses various difficulties for handling provenance, some of which are mitigated by assuming multiplicative idempotency of the semiring.  Under this assumption, we study three problems: ontology completion with provenance, computing the set of relevant axioms for a consequence, and query answering. 
\end{abstract}

\section{Introduction}

Description logics (DLs) are a well-known family of first-order logic fragments in which 
conceptual knowledge about a particular domain 
and facts about specific individuals are expressed in an ontology, using unary and binary predicates called \emph{concepts} and 
\emph{roles} \cite{dlhandbook}. 
Important reasoning tasks performed over 
DL ontologies are axiom entailment, i.e. deciding whether a given DL axiom follows from the ontology; and query answering. 
Since scalability is crucial when using large 
ontologies, 
DLs with favorable computational properties have been investigated. 
In particular, the \EL language and some of its extensions allow for axiom entailment in polynomial time, and conjunctive query entailment in \NP
 \cite{BBL-EL,BBL-EL08}. Many real-world 
 ontologies, including SNOMED CT, use languages from the \EL family, which underlies the OWL 2 EL profile of the Semantic Web standard ontology language. 

In many settings it is crucial to know \emph{how} a consequence---e.g.\ an axiom or a query---has been derived from the ontology. 
In the database community, provenance has been studied for nearly 30 years \cite{Bun2013} and gained traction when the connection to semirings, so called \emph{provenance semirings} \cite{Green07-provenance-seminal,GreenT17} was discovered. \emph{Provenance semirings} serve as an abstract algebraic tool to record and track provenance information; that is, to keep track of the specific database tuples used for deriving the query, and of the way they have been processed in the derivation. 
Besides explaining a query answer, provenance has many applications like: computing the probability or the degree of confidence of an answer, counting the different ways of producing an answer, 
handling authorship, 
data clearance, or 
user preferences 
\cite{Senellart17,Suciuetal2011,LMMPS2014,Ivesetal2008}. 
Semiring provenance has drawn interest beyond relational databases (e.g. \cite{BunemanK10,DBLP:journals/ws/ZimmermannLPS12,DeutchMRT14,RamusatMS18,DannertG19}), and in particular has recently been considered for ontology-based data access, a setting where a database is enriched with a 
$\DLLite_\Rmc$ ontology and mappings between them  \cite{CalvaneseLantiOzakiPenalozaXiao19}.  
In the latter, the ontology axioms are annotated with \emph{provenance variables}. 
Queries are then annotated with \emph{provenance polynomials} that express their provenance information. 

\begin{example}
Consider the facts $\mn{mayor}(\mn{Venice},\mn{Brugnaro})$ and $\mn{mayor}(\mn{Venice},\mn{Orsoni})$, stating that Venice has 
mayors Brugnaro and Orsoni, annotated respectively with provenance information $v_1$ and $v_2$, and the DL axiom ${\sf ran}(\mn{mayor})\sqsubseteq \mn{Mayor}$, expressing that the range of the role $\mn{mayor}$ is the concept $\mn{Mayor}$, annotated with $v_3$. 
The query $\exists x. \mn{Mayor}(x)$ asks if 
	there is someone who is a mayor. The answer is \emph{yes} and it can be derived using 
${\sf ran}(\mn{mayor})\sqsubseteq \mn{Mayor}$ together with any of the two facts, interpreting $x$ by $\mn{Brugnaro}$ or $\mn{Orsoni}$. 
This is expressed by the provenance polynomial $v_1\times v_3 + v_2\times v_3$. 
Intuitively, $\times$ expresses the joint use of axioms in a derivation path of the query, and $+$ the alternative derivations.
\end{example}

We adapt the provenance semantics of  \citeauthor{CalvaneseLantiOzakiPenalozaXiao19} 
for the \ELHr variant of \EL, 
extending it to those \ELHr axioms that do not occur in $\DLLite_\Rmc$. 
It turns out that handling the conjunction allowed in \ELHr axioms is not trivial. 
To obtain models from which we can derive meaningful provenance-annotated consequences, we adopt 
$\times$-idempotent semirings 
and a syntactic restriction on \ELHr (preserving the expressivity of full \ELHr when annotations are not considered). 
After introducing the basic definitions and the semantics for DL ontologies and queries annotated with provenance information, we present a completion algorithm and show that it solves annotated axiom entailment and instance queries in \ELHr in polynomial time in the size of the ontology and polynomial space in the size of the provenance polynomial. We then show that we can compute the set of relevant provenance variables for an entailment in polynomial time. Finally, we investigate conjunctive query answering. 
Note that the query answering methods developed by \citeauthor{CalvaneseLantiOzakiPenalozaXiao19} cannot be extended to \ELHr since they rely on the FO-rewritability of conjunctive queries in $\DLLite_\Rmc$, a property that does not hold for \ELHr 
\cite{BienetalFO2013}. Therefore, we adapt the combined approach for query answering in \EL~\cite{LTW:elcqrewriting09} to 
provenance-annotated \ELHr ontologies. 

Detailed proofs are available in the appendix.

 \section{Provenance for \ELHr}
\label{sec:prelim}

We first introduce our framework for provenance for \ELHr ontology and discuss our design choices.
\subsection{Basic Notions}

In the database setting, commutative semirings have proven to be convenient for representing various kinds of provenance information  \cite{Green07-provenance-seminal,GreenT17}.
In a commutative semiring $(K, \oplus ,\otimes,0,1)$, the product $\otimes$ and the addition $\oplus$ are commutative and associative binary operators over $K$, and~$\otimes$ distributes over~$\oplus$. 
Given a countably infinite set $\semiringVariables$ of 
\emph{variables} 
that are used to annotate the database tuples, a \emph{provenance semiring} 
is a semiring over a space of \emph{annotations}, or provenance expressions, 
with variables from~$\semiringVariables$. 
\citeauthor{GreenT17} present a hierarchy of expressiveness for provenance annotations \shortcite{GreenT17}. The most expressive form of annotations is provided by 
the \emph{provenance polynomials} semiring $\mathbb{N}[\semiringVariables] = (\mathbb{N}[\semiringVariables], \oplus ,\otimes,0,1)$ of polynomials with  
coefficients from $\mathbb{N}$ and variables from $\semiringVariables$, 
and the usual operations.
The semiring $\mathbb{N}[\semiringVariables]$ is universal, i.e., 
for any other commutative semiring $K = (K, \oplus ,\otimes,0,1)$, 
any function 
$\nu: \semiringVariables \rightarrow K$ can be extended to a semiring homomorphism $h: \mathbb{N}[\semiringVariables] \rightarrow K$, 
allowing the computations for $K$ to factor through the computations for $\mathbb{N}[\semiringVariables] $ 
\cite{GreenTheoCompSyst2011}. 
Hence, provenance polynomials  
provide 
the most informative provenance annotations and correspond to  
so-called \emph{how-provenance} \cite{DBLP:journals/ftdb/CheneyCT09}. 
Less general provenance semirings are obtained by restricting the operations $\oplus$ and $\otimes$ to idempotence and/or absorption 
\cite{GreenT17}. 
In this work, we 
focus on $\times$-idempotent semirings, i.e., 
for every $v\in K$, $v\times v=v$. 
This corresponds to 
the \emph{Trio semiring} $\mn{Trio}(\semiringVariables)$, defined in \cite{GreenTheoCompSyst2011} 
as 
the quotient semiring of $\mathbb{N}[\semiringVariables]$ by the equivalence kernel $\approx_\mn{trio}$ of the function $\mn{trio}:\mathbb{N}[\semiringVariables]\rightarrow \mathbb{N}[\semiringVariables]$ that ``drops exponents.'' 
An annotation is a polynomial $p$ that is understood to represent its equivalence class $p{\slash}{\approx}_\mn{trio}$. 
$\mn{Trio}(\semiringVariables)$ encompasses in the hierarchy  
the well-known \emph{why-provenance semiring} $\mn{Why}(\semiringVariables)$ 
obtained by restricting $\oplus$ to be idempotent as well, 
where an annotation corresponds to the set of sets of tuples used to derive the result \cite{DBLP:journals/ftdb/CheneyCT09}. 

We use the following notation. 
A \emph{monomial} is a finite product of variables in $\semiringVariables$.
Let \NM be the set of monomials, and \NPr the set of all finite sums of monomials,
i.e., \NPr contains polynomials of the form $\sum_{1\leq i \leq n}\prod_{1\leq j_i \leq m_i}v_{i,j_i}$,
with $v_{i,j_i}\in \semiringVariables$; $n,m_i>0$. 
By distributivity, every
polynomial can be 
written into 
this form. 
The \emph{representative} $\representative{\monomial}$ of 
a monomial $\monomial$ is the product of the variables occuring in $\monomial$, in
lexicographic order. 
Two monomials which are equivalent w.r.t.\ $\approx_\mn{trio}$ (i.e. are syntactically equal modulo commutativity, associativity and $\times$-idempotency)
have the same representative, e.g., $v\times u$ and $u\times v\times u$ have representative $u\times v$.  
$\NMrep$ denotes the set $\{\representative{\monomial}\mid\monomial\in\NM\}$. 
  
As ontology language we use a syntactic restriction of \ELHr. Consider three mutually disjoint countable sets of
\emph{concept-} \NC, \emph{role-} \NR, and \emph{individual names} \NI,
disjoint from \NV. 
\ELHr \emph{general concept inclusions} (GCIs) are expressions of the form $C\sqsubseteq D$,  
built according to 
the grammar rules
\[
  C::= A\mid \exists R.C \mid C\sqcap C\mid \top \quad\quad\quad D::= A\mid\exists R,
\]
where $R\in \NR$,~$A\in \NC$.
\emph{Role inclusions} (RIs) and 
\emph{range restrictions} (RRs) are expressions of the form 
$R\sqsubseteq S$ and ${\sf ran}(R)\sqsubseteq A$, respectively,
 with $R,S\in\NR$ and $A\in\NC$.
An \emph{ assertion}  
is an expression
 of the form $A(a)$
 or $R(a,b)$, with $A\in \NC$, $R\in \NR$, and $a,b\in \NI$. 
 An \emph{axiom} is a GCI, RI, RR, or assertion. 
An \ELHr ontology is a finite set of \ELHr axioms. 
\ELHr usually allows GCIs of the form $C\sqsubseteq C$, but
these can be translated into our format 
by exhaustively applying the 
rules: (i) replace $C\sqsubseteq C_1\sqcap C_2$ 
by $C\sqsubseteq C_1$ and $C\sqsubseteq C_2$, (ii)~replace $C_1\sqsubseteq \exists R.C_2$ 
by $C_1\sqsubseteq \exists S$, \mbox{$S\sqsubseteq R$} and ${\sf ran} (S)\sqsubseteq C_2$ 
where $S$ is a fresh role name.  
The reason for syntactically restricting \ELHr 
is that conjunctions or qualified restrictions of a role on the right-hand side 
of GCIs lead to counter-intuitive behavior when adding provenance annotations. 
We discuss this later in this section.

\subsection{Annotated Ontologies} 
Provenance information is stored as annotations. 
An \emph{annotated axiom} has the 
form 
$(\alpha,\monomial)$ with $\alpha$ an axiom and $\monomial\in\NM$.
An \emph{annotated \ELHr ontology} \Omc is a finite set of annotated \ELHr axioms of the form $(\alpha,\variable)$ with $\variable\in\NV\cup\{1\}$. We denote by  
 $\individuals{\Omc}$ the set of individual names occurring in~\Omc. 

The semantics of annotated 
ontologies extends the classical notion of
interpretations to track provenance.
An \emph{annotated interpretation} is a triple
$\Imc=(\Delta^\Imc,\Delta^\Imc_{\sf m},\cdot^\Imc)$
where $\Delta^\Imc,\Delta^\Imc_{\sf m}$ are non-empty disjoint sets (the
\emph{domain} 
and 
\emph{domain of monomials} of \Imc, respectively),
and $\cdot^\Imc$ maps
\begin{itemize}[noitemsep]
\item every $a\in\NI$ to $a^\Imc\in\Delta^\Imc$;
\item   every $A\in\NC$ to $A^\Imc\subseteq \Delta^\Imc\times \Delta^\Imc_{\sf m}$;
\item every   $R\in\NR$ to
$R^\Imc\subseteq \Delta^\Imc\times\Delta^\Imc\times \Delta^\Imc_{\sf m}$; and
\item  every $\monomial,\nonomial\in \NM$
to $\monomial^\Imc,\nonomial^\Imc\in\Delta^\Imc_{\sf m}$
s.t.\ $\monomial^\Imc=\nonomial^\Imc$ iff $m \approx_\mn{trio} n$.\footnote{or iff $\monomial=\nonomial$ if we consider $\mathbb{N}[\semiringVariables]$ instead of  $\mn{Trio}(\semiringVariables)$ 
}
\end{itemize}
We extend
$\cdot^\Imc$ to complex \ELHr expressions as usual: 
  \begin{align*} 
	(\top)^\Imc = {} & \Delta^\Imc\times \{1^\Imc\}; \\ 
	(\exists R)^\Imc = {} & \{(d,\monomial^\Imc)\mid \exists e\in\Delta^\Imc 
	 \text{ s.t. }(d,e,\monomial^\Imc)\in R^\Imc\};\\
    (C\sqcap D)^\Imc = {} & \{(d,(\monomial\otimes \nonomial)^\Imc)\mid (d,\monomial^\Imc)\in C^\Imc, (d,\nonomial^\Imc)\in D^\Imc\}; \\
    ({\sf ran} (R))^\Imc = {} & \{(e,\monomial^\Imc)\mid \exists d\in\Delta^\Imc
    \text{ s.t. }(d,e,\monomial^\Imc)\in R^\Imc\};   \\ 
    (\exists R.C)^\Imc = {} & \{(d,(\monomial\times \nonomial)^\Imc)\mid \exists e\in\Delta^\Imc
    \text{ s.t. }\\
	    & \ \ (d,e,\monomial^\Imc)\in R^\Imc, (e, \nonomial^\Imc)\in C^\Imc\}. 
  \end{align*} 
The annotated interpretation \Imc \emph{satisfies}: \\
\centerline{$
  \begin{array}{ll}
    \pair{R\sqsubseteq S}{\monomial}, &\text{if, for all } \nonomial\in\NM,
    (d,e,\nonomial^\Imc)\in R^\Imc \\
   	&\text{ implies  } (d,e,(\monomial\otimes \nonomial)^\Imc)\in S^\Imc\!; \\
    \pair{C\sqsubseteq D}{\monomial}, &\text{if, for all } \nonomial\in\NM,
    (d,\nonomial^\Imc)\in C^\Imc \\
   	&\text{ implies  } (d,(\monomial\otimes \nonomial)^\Imc)\in D^\Imc\!;
   	\\
\pair{A(a)}{\monomial}, & \text{if } (a^\Imc,\monomial^\Imc)\in A^\Imc\!; \quad \text{ and } \\ 
   \pair{R(a,b)}{\monomial}, & \text{if } (a^\Imc,b^\Imc,\monomial^\Imc)\in R^\Imc\!.  	
  \end{array}
$}
\smallskip

\noindent \Imc is a model of the annotated ontology $\Omc$, denoted $\Imc\models\Omc$, if
it satisfies all annotated axioms in \Omc.   
$\Omc$ \emph{entails} 
$(\alpha,\monomial)$, denoted $\Omc\models (\alpha,\monomial)$,
if 
$\Imc\models (\alpha,\monomial)$ for every model $\Imc$ of $\Omc$.  
\begin{remark}
While it may appear counter-intuitive at first sight that $C^\Imc$ differs from $(C\sqcap C)^\Imc$, 
this is 
 in line with the intuition behind the provenance of a conjunction. In the database setting, the Trio-provenance of tuple $(a)$ being an answer to query $\exists yz.R(x,y)\wedge R(x,z)$ over $\{R(a,b), R(a,c)\}$ is also different from that of $(a)$ being an answer to $\exists y.R(x,y)$. 
\end{remark}
Example~\ref{ex:semantic} illustrates the semantics  and some differences with the $\DLLite_\Rmc$ case from
\cite{CalvaneseLantiOzakiPenalozaXiao19}. 

\begin{example}\label{ex:semantic}
Consider the following annotated ontology.
\begin{align*}
\Omc=\{&
(\mn{mayor}(\mn{Venice},\mn{Orsoni}),v_1),\\&(\mn{predecessor}(\mn{Brugnaro},\mn{Orsoni}),v_2),\\&
(\exists\mn{predecessor}.\mn{Mayor}\sqsubseteq \mn{Mayor}, v_3),\\&({\sf ran}(\mn{mayor})\sqsubseteq \mn{Mayor},v_4)\}.
\end{align*}
Let $\Imc$ be s.t.\  $\Delta^\Imc=\{\mn{Brugnaro},\mn{Orsoni},\mn{Venice}\}$, $\Delta^\Imc_{\sf m}=\NMrep$, individual names are interpreted by themselves, monomials by their representatives and 
\begin{align*}
\mn{mayor}^\Imc=\{&(\mn{Venice},\mn{Orsoni},v_1)\},\\
\mn{predecessor}^\Imc=\{&(\mn{Brugnaro},\mn{Orsoni},v_2)\},\\
\mn{Mayor}^\Imc=\{&(\mn{Orsoni},v_1\times v_4),\\&(\mn{Brugnaro},v_1\times v_2\times v_3\times v_4)\}.
\end{align*}
$\Imc\models\Omc$ by the semantics of annotated \ELHr. Moreover, it can be verified that if $\Imc\models (\alpha,\monomial)$, then $\Omc\models (\alpha,\monomial)$. Note 
that  $\Omc$ entails $ (\mn{Mayor}(\mn{Brugnaro}),v_1\times v_2\times v_3\times v_4)$ whose provenance monomial contains $v_1$ and $v_2$, witnessing that the two assertions of $\Omc$ have been used to derive $\mn{Mayor}(\mn{Brugnaro})$. Combining two assertions to derive another one is not possible in $\DLLite_\Rmc$. 
The rewriting-based approach by \citeauthor{CalvaneseLantiOzakiPenalozaXiao19} 
 cannot be applied here as $\exists\mn{predecessor}.\mn{Mayor}\sqsubseteq \mn{Mayor}$  leads to infinitely many rewritings.
\end{example}

\subsection{Discussion on Framework Restrictions} 
Example \ref{ex:semantic} shows that conjunction and qualified role restriction 
lead to a behavior different from $\DLLite_\Rmc$. 
They are also the reason for some features of our setting. 
First, the next example illustrates the 
$\times$-idempotency impact 
for~the~\EL~family. 

\begin{example}\label{ex:idempotent}
Let $\Omc=\{(A\sqsubseteq B_1,v_1), (A\sqsubseteq B_2, v_2), (B_1\sqcap B_2\sqsubseteq C, v_3)\}$. 
If $\Imc$ is a model of $\Omc$ and 
$(e,\nonomial^\Imc)\in A^\Imc$, then $(e,(\nonomial\times v_1)^\Imc)\in B_1^\Imc$ and $(e,(\nonomial\times v_2)^\Imc)\in B_2^\Imc$ so $(e,(\nonomial\times v_1\times \nonomial\times v_2)^\Imc)\in (B_1\sqcap B_2)^\Imc$, i.e. $(e, (\nonomial\times v_1\times v_2)^\Imc)\in (B_1\sqcap B_2)^\Imc$ by $\times$-idempotency, which implies $(e, (\nonomial\times v_1\times v_2\times v_3)^\Imc)\in C^\Imc$. 
Thus $\Omc\models (A\sqsubseteq C, v_1\times v_2\times v_3)$. 
This intuitive entailment is lost if $\times$ is not idempotent. 
Indeed, 
assume that $\times$ is not idempotent and let 
$\Imc$ be the interpretation 
defined as follows (where $\Delta^\Imc=\{e\}$ and $\Delta_{\sf m}^\Imc$ contains all monomials with variables in lexicographic order).
\begin{align*}
A^{\Imc}=&\{(e,u)\} \quad B_1^{\Imc}=\{(e,u\times v_1)\} \quad B_2^{\Imc}=\{(e,u\times v_2)\}\\
C^{\Imc}=&\{(e,u\times u\times v_1\times v_2\times v_3)\}.
\end{align*}
$\Imc$ is a model of $\Omc$ such that $\Imc\not\models (A\sqsubseteq C, v_1\times v_2\times v_3)$.
\end{example}
A downside of 
$\times$-idempotency is a loss of the 
expressive power of provenance, 
neglecting the number of times an axiom is used in a derivation. 
Let $\Omc=\{(A\sqsubseteq B, v_1), (B\sqsubseteq A, v_2)\}$. 
With $\times$-idempotency, $\Omc\models (A\sqsubseteq B, v_1^k\times v_2^l)$ for $k\geq 1$ and $l\geq 0$ because for $k,l\geq 1$, $v_1^k\times v_2^l$ is  interpreted by $(v_1\times v_2)^\Imc$ in any interpretation $\Imc$. 
In contrast, if $\times$ is not idempotent, we only 
obtain $\Omc\models (A\sqsubseteq B, v_1^{k+1}\times v_2^k)$ for $k\geq 0$ (in particular $\Omc\not\models (A\sqsubseteq B, v_1\times v_2)$), which is a more informative result. 
Some useful 
semirings are not $\times$-idempotent; e.g.\ the Viterbi semiring $(\left[0,1\right],\mn{max},\times, 0,1)$, where $\times$ is the usual product over real numbers, which is applied
for representing
confidence scores. %
We limit ourselves to $\times$-idempotent semirings  because we are interested in computing provenance not only for 
assertions or 
queries, 
but also for 
GCIs. 
In particular, when a non-annotated ontology entails the GCI $C\sqsubseteq D$, we want the annotated version of the ontology to entail $(C\sqsubseteq D,\monomial)$ for some monomial $\monomial$. 
The non-idempotent case could be relevant when one is not concerned 
with 
 provenance for GCI entailment, and is left as future work. 
 
Many useful semirings are $\times$-idempotent. Examples of these are: the Boolean semiring, used  
for probabilistic query answering in databases; the security semiring, used to determine the minimal level of clearance required to get 
the consequence; and the fuzzy semiring which allows to determine the truth degree of the consequence (see e.g.~\cite{Senellart17} 
for details on these semirings and more examples).
\smallskip

Second, let us explain the restrictions on the form of the right-hand side of the GCIs. 
Example~\ref{ex:conj-right} illustrates the case of conjunctions. Qualified role restrictions lead to the same kind of behavior (they can be seen as implicit conjunctions). 
 \begin{example}\label{ex:conj-right}
Let $\Omc=\{(A\sqsubseteq B\sqcap C,v), (A(a),u)\}$. All the following interpretations 
which interpret  $a$ by itself and monomials by their representatives are models of $\Omc$:
\begin{equation*}
\begin{array}{l@{\ }l@{\ }l}
A^{\Imc_1}=\{(a,u)\}, & B^{\Imc_1}=\{(a,u\times v)\}, & C^{\Imc_1}=\{(a,u\times v)\}\\
A^{\Imc_2}=\{(a,u)\}, & B^{\Imc_2}=\{(a,u)\}, & C^{\Imc_2}=\{(a,v)\}
\\
A^{\Imc_3}=\{(a,u)\}, & B^{\Imc_3}=\{(a,1)\}, & C^{\Imc_3}=\{(a,u\times v)\}
\end{array}
\end{equation*}
Since the semantics does not provide a unique way to ``split'' the monomial $u\times v$ between the two elements of the conjunction, $\Omc\not\models (B(a), \monomial)$ for any $\monomial\in\NM$, and in particular, $\Omc\not\models (B(a), u\times v)$. 
It is arguably counter-intuitive since we intuitively know that $a$ is in $A$ with provenance $u$ and that $A$ is a subclass of the intersection of $B$ and $C$ with provenance $v$. 
\end{example}
Partially normalizing the ontology before annotating it, or more specifically, replacing e.g.\ annotated GCIs of the form 
$(C\sqsubseteq C_1\sqcap C_2,\variable)$ by $(C\sqsubseteq C_1,\variable)$ and $(C\sqsubseteq C_2,\variable)$, 
may be 
acceptable in most cases, even if the rewritten ontology leads to additional---arguably natural---consequences compared to 
the original one. For instance, even if  $\Omc\not\models(A\sqsubseteq B, v)$ 
in Example~ \ref{ex:conj-right}, 
in many cases a user would accept to change the GCI of $\Omc$ to $(A\sqsubseteq B, v)$ and $(A\sqsubseteq C, v)$ as it may reflect the original intention of the GCI since $\{A\sqsubseteq B,A\sqsubseteq C\}$ and $A\sqsubseteq B\sqcap C$ are semantically equivalent. 

One could 
argue that it would be better to 
define the semantics so that only $\Imc_1$ 
was a model of $\Omc$ in Example~\ref{ex:conj-right}, instead of restricting the language as we 
do. We explain next why this is not so simple.

One possibility is to change the definition of satisfaction of a GCI by an interpretation such that 
$\Imc\models (A\sqsubseteq B\sqcap C,\monomial)$ iff for every $(d,\nonomial^\Imc)\in A^\Imc$, then 
$(d,(\monomial\times\nonomial)^\Imc)\in B^\Imc$ and $(d,(\monomial\times\nonomial)^\Imc)\in C^\Imc$, and similarly for qualified role restrictions. This approach leads to a
counter-intuitive behavior. For instance if $\Omc=\{(A(a),u), (B(a),v)\}$, then $\Omc\not\models (A\sqcap B\sqsubseteq A\sqcap B, 1)$, since $(a, (u\times v)^\Imc)\in (A\sqcap B)^\Imc$ for every model $\Imc$ of $\Omc$, but there is a model $\Imc$ of $\Omc$ such that $(a, (u\times v)^\Imc)\notin A^\Imc$ (and $(a, (u\times v)^\Imc)\notin B^\Imc$). 
In contrast,
our definition of satisfaction ensures that for every interpretation $\Imc$ and concept $C$, $\Imc\models (C\sqsubseteq C, 1)$.

Another possibility is to modify the interpretation of conjunctions and qualified role restrictions such that
$(C\sqcap D)^\Imc=\{(d,\monomial^\Imc)\mid (d,\monomial^\Imc)\in C^\Imc, (d,\monomial^\Imc)\in D^\Imc\}$ and 
$(\exists R.C)^\Imc = \{(d,\monomial^\Imc)\mid \exists e\in\Delta^\Imc    \text{ s.t. }(d,e,\monomial^\Imc)\in R^\Imc, (e, \monomial^\Imc)\in C^\Imc\}$. 
In this case, we 
lose even basic entailments from annotated ABoxes; e.g.,
$\{(A(a),u), (B(a),v)\}\not\models ((A\sqcap B)(a), u\times v)$. 
We 
also lose the entailment of the 
GCI from Example~\ref{ex:idempotent}.

Hence, restricting the syntax to prevent conjunctions on the right and defining the semantics as usual in DLs 
seems to be the most natural way of handling provenance in DL languages with conjunction. 
Since 
\EL ontologies are often already 
expressed in normal form, 
the main restriction in our language is the avoidance of qualified existential restrictions
on the right-hand side.

\subsection{Annotated Queries}
Following 
\citeauthor{CalvaneseLantiOzakiPenalozaXiao19} \shortcite{CalvaneseLantiOzakiPenalozaXiao19},  
we extend 
DL conjunctive queries with 
binary and ternary predicates,
where the last term of the tuple is used for provenance information. 
Recall that by the semantics of annotated ontologies, tuples can only contain monomials. 
A \emph{Boolean conjunctive query~(BCQ)} $q$ is a sentence  
$\exists \vec{x}.\varphi(\vec{x},\vec{a})$,
where $\varphi$ is a conjunction of (unique) atoms of the form $A(t_1,t)$,
$R(t_1,t_2,t)$; $t_i$ is an individual name from $\vec{a}$,
or a variable from $\vec{x}$; and $t$ (the last term of the tuple) 
is 
a variable from $\vec{x}$ that does not occur anywhere else in~$q$  (\citeauthor{CalvaneseLantiOzakiPenalozaXiao19} call such a query \emph{standard}). 
We use $P(\vec{t},t)$ to refer to an atom which is either $A(t_1,t)$
or $R(t_1,t_2,t)$, and $P(\vec{t},t)\in q$ if 
$P(\vec{t},t)$ occurs in $q$.

A \emph{match} of the BCQ $q=\exists \vec{x}.\varphi(\vec{x},\vec{a})$ in the
annotated interpretation \Imc
is a function  $\map:\vec{x}\cup\vec{a}\to\Delta^\Imc\cup\Delta^\Imc_{\sf m}$,
such that $\map(b)=b^\Imc$ for all $b\in \vec{a}$,
and $\map(\vec{t},t)\in P^\Imc$ for every $P(\vec{t},t)\in q$, where $\map(\vec{t},t)$ is a shorthand for $(\pi(t_1),\pi(t))$ or $(\pi(t_1),\pi(t_2),\pi(t))$ depending on the arity of $P$.
\Imc satisfies the BCQ $q$, written $\Imc\models q$,
if there is a match of $q$ in \Imc.
A BCQ $q$ is \emph{entailed by} an annotated ontology $\Omc$, denoted $\Omc\models q$,
if every model of $\Omc$ satisfies~$q$. 
For a BCQ $q$ and an interpretation \Imc, $\nu_\Imc(q)$ denotes
the set of all matches   of $q$ in \Imc. 
The \emph{provenance} of $q$ on \Imc is the expression
\begin{align*}
\p{\Imc}{q} := \textstyle \sum_{\pi\in\nu_\Imc(q)}\representative{\prod_{P(\vec{t},t)\in q} \pi^-(t) }
\end{align*}
 where 
 $\pi(t)$
 is the last element of
 the tuple $\map(\vec{t},t)\in P^\Imc$; and
 $\pi^-(t)$ is the only $\monomial\in\NMrep$ s.t.\ $\monomial^\Imc=\pi(t)$.
 For $p\in\NPr$, we write $p\subseteq \p{\Imc}{q}$ if
 $p$ is a sum of monomials and for each occurrence of a monomial in
 $p$ we find an occurrence of 
its representative 
 in $\p{\Imc}{q}$.
\Imc \emph{satisfies} $q$ with provenance $p\in\NPr$, denoted 
$\Imc\models (q,p)$, if
    $\Imc\models q$ and $p\subseteq\p{\Imc}{q}$.
$\Omc\models (q,p)$,
if $\Omc\models q$ and $p\subseteq\p{\Imc}{q}$,
for all 
 $\Imc\models\Omc$. 
 We call $(q,p)$ an \emph{annotated query}. 
\begin{remark}
When $\Omc$ contains only assertions (no GCIs, RIs, and RRs), we can compare the provenance annotations we obtain to the database case. 
Similarly to Trio-provenance, 
the sums of monomials distinguish different ways the query atoms can be mapped to annotated interpretations. 
For example, 
given $\Omc=\{(R(a,b), v_1), (R(b,a), v_2)\}$ and query  
$q= \exists xytt'. R(x,y,t)\wedge R(y,x,t')$, it holds that $\Omc\models (q, v_1\times v_2 + v_1\times v_2)$. The provenance annotation $v_1\times v_2 + v_1\times v_2$ distinguishes among two derivations using the same axioms, contrary to the why-provenance $v_1\times v_2$. 
Note that given an axiom $\alpha$ and $\Omc$ that may contain GCIs, RIs, and RRs, the sum over all monomials $\monomial$ such that $\Omc\models (\alpha,\monomial)$ represents all possible derivations of $\alpha$, in the why-provenance spirit. 
\end{remark}
The \emph{size} $|X|$   of
  an annotated ontology, a polynomial or a BCQ $X$
is the length of the string representing $X$, where
elements of \NC, \NR, \NI and \NV in $X$ are of length one.
We often omit `annotated' and refer to  `ontologies,'
`queries,' `assertions,' etc.
 when it is clear from the context.

 \section{Reasoning with Annotated \ELHr Ontologies}\label{sec:completion}
 \begin{table*}[tb]
\resizebox{0.98\textwidth}{!}{
\begin{tabular}{@{}lll@{}}
\toprule
& if    &   then $\quad\quad\quad$ (if $\Phi\notin\Omc$) \\ 
\hline
\midrule
$\CR_0$ &  $X\in\NC\cup\NR\cup\{\top\}$ occurs in $\Omc$  
& 
					 add $\Phi=(X\sqsubseteq X,1)$ to $\Omc$ \\
$\CR_1$ &  $(R_1\sqsubseteq R_2,\monomial_1),(R_2\sqsubseteq R_3,\monomial_2)\in\Omc$ 
 & 
					 add $\Phi=(R_1\sqsubseteq R_3,\representative{\monomial_1\times \monomial_2})$ to $\Omc$ \\
$\CR_2$ &  $(R\sqsubseteq S,\monomial_1),({\sf ran}(S)\sqsubseteq A,\monomial_2)\in\Omc$ 
& 
					 add $\Phi=({\sf ran}(R)\sqsubseteq A,\representative{\monomial_1\times \monomial_2})$ to $\Omc$ \\

$\CR_3$ &  $(A\sqsubseteq \exists R,\monomial_1),(R\sqsubseteq S,\monomial_2)\in\Omc$  
&
					 add $\Phi=(A\sqsubseteq \exists S,\representative{\monomial_1\times \monomial_2})$ to $\Omc$  \\

$\CR_4$ &  $(A\sqsubseteq B,\monomial_1),(B\sqsubseteq C,\monomial_2)\in\Omc$ 
& 
					 add $\Phi=(A\sqsubseteq C,\representative{\monomial_1\times \monomial_2})$ to $\Omc$ \\
$\CR_5$ &  $(A\sqsubseteq B,\monomial_1),(B\sqsubseteq \exists R,\monomial_2)\in\Omc$ 
& 
					 add $\Phi=(A\sqsubseteq \exists R,\representative{\monomial_1\times \monomial_2})$ to $\Omc$ \\

$\CR_{6}$ &  $(A\sqsubseteq B_1,\monomial_1),(A\sqsubseteq B_2,\monomial_2),(B_1\sqcap B_2\sqsubseteq C,\monomial_3)\in\Omc$ 
& 
					 add $\Phi=(A\sqsubseteq C,\representative{\monomial_1\times \monomial_2\times \monomial_3})$ to $\Omc$  \\	
					
$\CR_{7}$ &  $({\sf ran}(R)\sqsubseteq B_1,\monomial_1),({\sf ran}(R)\sqsubseteq B_2,\monomial_2),(B_1\sqsubseteq C_1,\monomial_3),
(B_2\sqsubseteq C_2,\monomial_4),(C_1\sqcap C_2\sqsubseteq C,\monomial_5)\in\Omc$ 
& 
					 add $\Phi=({\sf ran}(R)\sqsubseteq C,\representative{\monomial_1\times \monomial_2\times \monomial_3\times \monomial_4\times \monomial_5})$ to $\Omc$  \\					
$\CR_{8}$ &  $(A\sqcap B\sqsubseteq C,\monomial_1), (\top\sqsubseteq B,\monomial_2)\in\Omc$ 
& 
					 add $\Phi=(A\sqsubseteq C,\representative{ \monomial_1\times\monomial_2})$ to $\Omc$  \\										
					
$\CR_{9}$ &  $(A\sqsubseteq \exists S,\monomial_1),({\sf ran}(S)\sqsubseteq B, \monomial_2), 
(B\sqsubseteq C,\monomial_3), (S\sqsubseteq R, \monomial_4),(\exists R.C\sqsubseteq D,\monomial_5)\in\Omc$  
&
					 add $\Phi=(A\sqsubseteq D,\representative{\monomial_1\times \monomial_2\times \monomial_3\times\monomial_4\times\monomial_5})$  to $\Omc$ \\
$\CR_{10}$ &  $(A\sqsubseteq \exists R,\monomial_1),(\top\sqsubseteq B,\monomial_2), 
(\exists R.B\sqsubseteq C,\monomial_3)\in\Omc$ 
& 
					 add $\Phi=(A\sqsubseteq C,\representative{\monomial_1\times \monomial_2\times \monomial_3})$ to $\Omc$ \\						
					
\midrule
$\CR_{11}$ &  $a\in\individuals{\Omc}$  
& 
					add $\Phi=(\top(a),1)$ to $\Omc$ \\
$\CR_{12}$ &  $(R(a,b),\monomial_1),(R\sqsubseteq S,\monomial_2)\in\Omc$  
&
					 add $\Phi=(S(a,b),\representative{\monomial_1\times \monomial_2})$ to $\Omc$  \\

$\CR_{13}$ &  $(A(a),\monomial_1),(A\sqsubseteq B,\monomial_2)\in\Omc$  
&
					 add $\Phi=(B(a),\representative{\monomial_1\times \monomial_2})$ to $\Omc$  \\
	$\CR_{14}$ &  $(A_1(a),\monomial_1),(A_2(a),\monomial_2),(A_1\sqcap A_2\sqsubseteq B,\monomial_3)\in\Omc$  
	&
					 add $\Phi=(B(a),\representative{\monomial_1\times \monomial_2\times \monomial_3})$ to $\Omc$  \\				
					
$\CR_{15}$ &  $(R(a,b),\monomial_1),(A(b),\monomial_2),(\exists R.A\sqsubseteq B,\monomial_3)\in\Omc$  
&
					 add $\Phi=(B(a),\representative{\monomial_1\times \monomial_2\times \monomial_3})$ to $\Omc$  \\
					
$\CR_{16}$ &  $(R(a,b),\monomial_1),({\sf ran}(R)\sqsubseteq A,\monomial_2)\in\Omc$  
&
					 add $\Phi=(A(b),\representative{\monomial_1\times \monomial_2})$ to $\Omc$  \\
\bottomrule
\end{tabular}
}
\caption{Completion rules. $A,\ldots,D\in\NC\cup\{\top\}$, $R,S,R_i\in\NR$, $\monomial,\monomial_i\in\NM$.}
\label{tab:completionRules}
\end{table*}

We present a completion algorithm for deriving basic entailments from an \ELHr ontology. As usual with completion
algorithms, we restrict to ontologies in normal form.
The annotated \ELHr ontology \Omc is in \emph{normal form} if for every GCI $(\alpha,v)\in \Omc$, $\alpha$ is of the form 
$A\sqsubseteq B$, $A\sqcap A'\sqsubseteq B$, 
$A\sqsubseteq \exists R$, or $\exists R.A\sqsubseteq B$, with $A,A'\in\NC\cup\{\top\}$, $B\in\NC$. Every annotated \ELHr ontology 
can be transformed, in polynomial time, into an ontology in normal form which entails the same axioms over the ontology signature, using 
the following rules where $\widehat{C},\widehat{D}\notin\NC\cup\{\top\}$ and $A$ is a fresh concept name:
\begin{equation*}
\begin{array}{l@{\ }r@{\ }l@{\ }l}
 \NF_1: &  (C\sqcap\widehat{D}\sqsubseteq E,\, v) &\longrightarrow&(\widehat{D}\sqsubseteq A,\, 1), (C\sqcap A\sqsubseteq E, \, v)\\
 \NF_2: &   (\exists R.\widehat{C}\sqsubseteq D,\, v) &\longrightarrow&(\widehat{C}\sqsubseteq A,\, 1), (\exists R. A\sqsubseteq D, \, v)\\
 \NF_3: &   (\widehat{C}\sqsubseteq\exists R,\, v) &\longrightarrow&(\widehat{C}\sqsubseteq A,\, 1), (A\sqsubseteq \exists R, \, v). 
\end{array}
\end{equation*}

\begin{restatable}{theorem}{Normalisation}
Let $\Omc$ be an annotated \ELHr  ontology, $\alpha$ an axiom, and $\monomial$ a monomial. 
Let $\Omc'$ be obtained by applying exhaustively  Rules $\NF_1$-$\NF_3$ 
to~$\Omc$.
\begin{itemize}
\item If $\Omc\models (\alpha,\monomial)$, then $\Omc'\models (\alpha,\monomial)$.
\item If $\Omc'\models (\alpha,\monomial)$ and every concept name occurring in $\alpha$ occurs in $\Omc$, then $\Omc\models (\alpha,\monomial)$.
\end{itemize}
\end{restatable}
Before describing the reasoning algorithm in detail, we present an important property of entailment; namely, that all entailment
problems can be polynomially reduced to each other. 
This allows us to focus on only one problem. In particular,
we focus on entailment of annotated assertions. 

\begin{theorem}
\label{prop:reductions}
Let \Omc be an annotated ontology, and $(\alpha,\monomial)$ an annotated GCI, RR, or RI. One can construct in polynomial
time an ontology $\Omc'$ and an annotated assertion $(\beta,\nonomial)$ such that $\Omc\models(\alpha,\monomial)$ iff
$\Omc'\models(\beta,\nonomial)$.
Conversely, if $(\alpha,\monomial)$ is an annotated concept (resp. role) assertion, one can construct in polynomial time an ontology $\Omc'$ and two annotated concept (resp. role) inclusions $(\beta,\nonomial)$, $(\gamma,\nonomial)$ such that $\Omc\models(\alpha,\monomial)$ iff
$\Omc'\models(\beta,\nonomial)$ or 
$\Omc'\models(\gamma,\nonomial)$.
\end{theorem}
We adapt the classical \EL completion rules to handle annotated \ELHr ontologies in normal form. 
The algorithm starts with the original ontology \Omc, and extends it through an iterative application of the rules from
Table~\ref{tab:completionRules} until \Omc becomes \emph{saturated}; i.e., no more rules are applicable. 
We cannot use the 
rules of \cite{BBL-EL08} which eliminate range restrictions by adding GCIs with qualified role restrictions on the right so we designed rules for \ELHr.

A rule application may add axioms annotated with monomials, and other assertions $(\top(a),1)$, which are not
foreseen in the definition of annotated ontologies. Still, $\times$\mbox{-}idempotency ensures that all monomials 
have at most $|\Omc|$ factors. To show that the completion algorithm is sound and complete for deciding assertion entailment, 
we prove a stronger result. 
The \emph{$k$-saturation} of \Omc is the saturated ontology $\Omc_k$ obtained from \Omc through the completion algorithm restricted 
to monomials of length at most $k$. We show that $\Omc_k$ suffices for deciding entailment of annotated assertions $(\alpha,\monomial)$ 
where $\monomial$ is a monomial of length at most $k$.

\begin{restatable}{theorem}{Correctness}
\label{prop:completionalgorithm}
If $\Omc_k$ is the $k$-saturation of \Omc, then
\begin{enumerate}
\item $\Omc_k$ is computable in polynomial time w.r.t.\ the size of $\Omc$, and in exponential time w.r.t.\ $k$,
\item for every assertion $\alpha$ and monomial $\monomial_k$ with at most $k$ variables,  
$\Omc\models (\alpha,\monomial_k)$ iff $(\alpha, \representative{\monomial_k})\in\Omc_k$. 
\end{enumerate}
\end{restatable}
This theorem states that to decide whether an assertion $(\alpha,\monomial)$ is entailed by \Omc, one just needs to find the 
$k$\mbox{-}saturation of \Omc, where $k$ is the number of variables in \monomial, and then check whether $(\alpha,\representative{\monomial})\in\Omc_k$.
Due to the first point of Theorem~\ref{prop:completionalgorithm} and Theorem~\ref{prop:reductions}, 
we obtain the following corollary.
\begin{corollary}
\label{cor:complexity:provmonomial}
For every axiom $\alpha$, $\Omc\models (\alpha,\monomial)$ is decidable in polynomial time in   $|\Omc|$ and in exponential time in  
$|\monomial|$.
\end{corollary}
In general there is no need to interrupt the completion algorithm;
the ontology saturated without restricting the monomial
length can be used to decide all relevant entailments regardless of the length of the monomial. 
Using $\Omc_k$ is merely an optimisation 
when one is only interested in a short monomial.

While the polynomial time upper bound w.r.t.\ the ontology size  is positive, and in line with the complexity of the \EL family,
the exponential time bound on the monomial size does not scale well 
for entailments with
larger monomials. Recall that these 
bounds are based on the number of annotated axioms 
generated by 
the completion rules. The following example illustrates the potential exponential blow-up.

\begin{example}\label{ex:exponential}
Consider 
$\Omc=\{(A\sqsubseteq A_i,v_i), (A_i\sqsubseteq B,u_i)\mid 0\leq i\leq n\}\cup\{(B\sqsubseteq A,u)\}$. If $\Omc'$ is the result of applying the completion algorithm to $\Omc$, then for every $S\subseteq \{1,\dots, n\}$, $(B\sqsubseteq A, \representative{u \times \Pi_{i\in S} u_i\times v_i})\in\Omc'$.
\end{example}
Following \citeauthor{HuPe-17} \shortcite{HuPe-17}, we can see the completion algorithm
as an automaton. More precisely, given $\Omc$ and $(\alpha,\monomial)$, we can construct a tree automaton \Amc, whose states 
correspond exactly to all the elements in $\Omc_k$, such that $(\alpha,\representative{\monomial})\in\Omc_k$ iff \Amc accepts 
at least one tree. Briefly, \Amc is constructed by reading the rule applications backwards, allowing transitions from the
consequence to the premises of the rule; see~\cite{HuPe-17} for details.
The number of states in \Amc is exactly the cardinality of $\Omc_k$ and hence potentially exponential
on $k$. However, the size of each state is bounded polynomially on $k$; the arity of the automaton is bounded by the maximum
number of premises in a rule, in this case 5; and one can bound polynomially on $k$ the number of different states that may 
appear in any successful run of \Amc. Thus, \Amc satisfies 
 the conditions 
 for a PSpace emptiness test
\cite{BaHP08}, which yields the following result.

\begin{proposition}
\label{prop:pspace}
For every axiom $\alpha$, $\Omc\models (\alpha,\monomial)$ is decidable in polynomial space in  $|\monomial|$.
\end{proposition}
Interestingly, these results allow us to bound the full complexity of answering \emph{instance queries} (IQ) of the
form $C(a)$ where $C$ is an \ELHr concept and $a\in\NI$.

\begin{restatable}{theorem}{Theorembasic}
\label{th:instancequeries}
Let \Omc be an ontology, $C(a)$ an IQ and $\monomial\in\NM$. $\Omc\models (C(a),\monomial)$ is decidable in polynomial time in $|\Omc|$ and $|C(a)|$, and polynomial 
space in $|\monomial|$.
\end{restatable}
 
\section{Computing Relevant Provenance Variables}\label{sec:relevance}
An interesting question is
whether a given variable
appears in the provenance of a query \query;
i.e., whether a given axiom occurs in some derivation of  \query. 
Formally, $\variable\in\semiringVariables$ is \emph{relevant} for $\query$ (w.r.t.\ 
ontology \Omc) iff 
$\exists\monomial\in\NM\text{ s.t. }\Omc\models(q,\variable\otimes \monomial)$. 
For IQs and \ELHr this problem can be solved in polynomial time, via an algorithm
computing all the relevant
variables for all queries of the form $A(a)$, with $a\in\NI$, $A\in\NC$.
We modify 
the completion algorithm (Section \ref{sec:completion})
to combine all monomials from a derivation, instead of storing them separately.

As in Section \ref{sec:completion}, 
the algorithm assumes normal form and keeps as data structure 
a set \Smc of annotated axioms $(\alpha,\monomial)$, 
where $\alpha$ uses the vocabulary of \Omc, and $\monomial\in \NM$.
\Smc is initialised as the original ontology where annotations of the same axiom are merged into a single monomial:
$$\Smc := \{ 
(\alpha, \representative{\Pi_{v\in V_\alpha}v})
 \mid (\alpha,u)\in\Omc, V_\alpha=\{v\mid (\alpha,v)\in \Omc\}
\},$$ 
and extended by exhaustively applying the rules in Table~\ref{tab:completionRules}, where rule applications 
change \Smc into
\[
\Smc\Cup(\alpha,\monomial) :=
	\begin{cases}
		\Smc\cup\{(\alpha,\monomial)\}  \text{ if there is no $(\alpha,\nonomial)\in\Smc$} \\
		\Smc\setminus\{(\alpha,\nonomial)\}\cup\{(\alpha,\representative{\monomial\times\nonomial})\}  \text{ if $(\alpha,\nonomial)\in\Smc$;}
	\end{cases}
\]
i.e., add the
axiom $\alpha$ with an associated monomial if it does not yet appear in \Smc, and modify the monomial
associated to $\alpha$ to include new variables otherwise. To ensure termination, a rule  is only applied if 
it modifies~
\Smc. The rules are applied until no new 
rule is applicable; i.e., \Smc is~\emph{saturated}.

\begin{example}
The relevance algorithm on the ontology of Example \ref{ex:exponential}, yields the saturated set
$\Smc=\{(A\sqsubseteq A,\monomial), (B\sqsubseteq B,\monomial),(A\sqsubseteq B,\monomial), (B\sqsubseteq A,\monomial)\}\cup\{(A_i\sqsubseteq B,\monomial), (B\sqsubseteq A_i,\monomial),(A_i\sqsubseteq A,\monomial), (A\sqsubseteq A_i,\monomial)\mid (1\leq i\leq n\}\cup\{(A_i\sqsubseteq A_j,\monomial)\mid 1\leq i,j\leq n\}$ with $\monomial=u\times \Pi_{i=1}^n u_i\times \Pi_{i=1}^n v_i$.
\end{example}
Each rule application either adds a new axiom, or adds
to the label of an existing axiom more variables. As the number of concept and role names, and variables appearing in \Smc is 
linear on \Omc, at most polynomially many rules are applied, each requiring polynomial time;
i.e, the algorithm is polynomial. 

\begin{lemma}
If \Smc is the saturated set obtained from \Omc,
$a\in\NI$, $A\in\NC$, and $\variable\in\semiringVariables$, then $\variable$ is relevant for $A(a)$ iff 
$\variable$ occurs in $\monomial$ for some
$(A(a),\monomial)\in\Smc$.
\end{lemma}

The algorithm 
decides relevance for assertion entailment in \ELHr,
yielding a polynomial-time upper bound for this problem. 
As in Section~\ref{sec:completion}, axioms and 
IQs can be handled in polynomial time as well. 
\begin{theorem}
Relevance for axiom and IQ entailment in \ELHr can be decided in polynomial time.
\end{theorem}
This result shows that if we only need to know which axioms are used to derive an axiom or an IQ, the complexity is the 
same as reasoning in \ELHr without provenance. 
This contrasts with \emph{axiom pinpointing}
: the task of finding the axioms 
responsible for a consequence to follow, in the sense of 
belonging to some \emph{minimal} subontology entailing it (a MinA). Deciding whether an axiom belongs to a MinA is 
\NP-hard for Horn-\EL\cite{PenalozaSertkaya10}. 
Relevance is easier in our context since provenance does not require minimality: if 
$\Omc=\{(A\sqsubseteq B, v_1), (B\sqsubseteq C, v_2), (C\sqsubseteq B, v_3)\}$, $v_2$ and $v_3$ are relevant for $A\sqsubseteq B$, but the only MinA is $\{A\sqsubseteq B\}$ so other axioms are not relevant for axiom pinpointing. 

Provenance relevance is related to \emph{lean kernels} (LKs) \cite{PenalozaMIM17}, 
which approximate the union of MinAs. 
The LK of a consequence $c$ is the set of axioms appearing in at least one proof of $c$
in a given inference method, generalizing the notion from propositional logic, 
where an LK is the set of 
clauses appearing in a resolution proof for unsatisfiability. 
The sets of variables computed by our algorithm 
are the sets of axioms 
used in the derivations by the completion algorithm, which is a consequence-based method for \ELHr. 
Thus they correspond to 
LKs for the associated axioms and  
our algorithm 
is an alternative way of computing LKs in \ELHr.

 \section{Query Answering with Provenance}
\label{sec:combined}

Even if \ELHr is expressive enough to reduce entailment of rooted tree-shaped BCQs to assertion entailment, 
the methods presented in  Section~\ref{sec:completion} do not apply to other kinds of BCQs. 
\begin{example} 
For
$\Omc{=}\{(R(a,a),u_1),(A(a),u_2), (A{\sqsubseteq} \exists R, v_1)$, $({\sf ran}(R){\sqsubseteq} A, v_2)\}$ and
$q= \exists xyztt't''. R(x,x,t)\land R(x,y,t')\land R(z,y,t'')$,
$\Omc\models (q,u_1)$ but $\Omc\not\models (q,u_2\times v_1\times v_2)$: 
\Omc has a model $\Imc$ with 
$R^\Imc=\nobreak\{(a,a,u_1), (a,b_1,u_2\times v_1),(a,c_1,u_1\times v_1\times v_2)\}\cup\{(b_i,b_{i+1}, u_2\times v_1\times v_2)\mid i\geq 1\}\cup\{(c_i,c_{i+1}, u_1\times v_1\times v_2)\mid i\geq 1\}$.
\end{example}
We adapt the combined approach by \citeauthor{LTW:elcqrewriting09} \shortcite{LTW:elcqrewriting09} to trace provenance. 
Assume that 
queries contain only individual names 
occurring in the ontology \Omc.
The combined approach builds a canonical model for \Omc
and shows that every query $\query$ can be rewritten into a query $\queryrewriting$
 that holds in this canonical model 
iff $\Omc\models\query$. 
We first define the canonical model $\Imc_\Omc$ of 
an ontology \Omc annotated with provenance information. 

Assume that \Omc is in normal form; 
$\monomials{\Omc}$ denotes the set of
monomial representatives built using variables of \NV occurring in \Omc, 
and $\roles{\Omc}$ is the set of role names occurring in \Omc. 
Also assume that $(\ast)$
  if 
  there is  $B\in\NC$, $R\in\NR$, and $n\in\NM$ such that $\Omc\models({\sf ran} (R)\sqsubseteq B,n)$, then $({\sf ran} (R)\sqsubseteq B,\representative{n})\in\Omc$.
This   
simplifies the presentation of the construction of the canonical model. 
 Let $\auxs{\Omc}:=\{\aux{m}{R}{} \mid   R\in\roles{\Omc}, m\in \monomials{\Omc}\}.$
Assume that $\individuals{\Omc}\cap \auxs{\Omc}=\emptyset$. 
We define the domain of $\Imc_\Omc$ and the domain of monomials of $\Imc_\Omc$ as follows:
\begin{align*}
\Delta^{\Imc_\Omc} := \individuals{\Omc}\cupdot \auxs{\Omc} \quad\quad\quad \Delta^{\Imc_\Omc}_{\sf m} := \NMrep
\end{align*}
We define the interpretation function of $\Imc_\Omc$ as the union of $\cdot^{\Imc^i_\Omc}$, $i\ge 0$.
The function $\cdot^{\Imc^0_\Omc}$ sets 
$a^{\Imc^0_\Omc}=a$ for all 
$a\in \individuals{\Omc}$ (for $a\in \NI\setminus\individuals{\Omc}$ 
the mapping $a^{\Imc_\Omc}$ is irrelevant), $\monomial^{\Imc^0_\Omc}=\representative{\monomial}$
for all $\monomial\in\NM$, and for all $A\in\NC$ and all $R\in\NR$, 
\begin{align*}
A^{\Imc^0_\Omc}&:=\{(a,\representative{m})\mid  \Omc\models (A(a),\monomial)\} \\
R^{\Imc^0_\Omc}&:=\{(a,b,[\monomial])\mid \Omc\models(R(a,b),\monomial)\}.
\end{align*}
If $\Imc^i_\Omc$ is defined, we define $\Imc^{i+1}_\Omc$ 
by choosing an annotated axiom $ \alpha\in \Omc$ and applying one of the following rules in a fair way (i.e., every applicable rule is eventually applied).
\begin{enumerate}[label=\textbf{R\arabic*},leftmargin=*,series=run]
\item\label{r1}   $\alpha=(C\sqsubseteq A,\monomial)$: if
there is  $d\in\Delta^{\Imc_\Omc}$ and $\nonomial\in\monomials{\Omc}$  s.t. 
$(d,\representative{\nonomial})\in C^{\Imc^{i}_\Omc} $,
then add $(d,\representative{m\times \nonomial})$ to $A^{\Imc^{i}_\Omc}$. 
\item \label{r2} $\alpha= (C\sqsubseteq \exists R,\monomial)$:
if there is  
$\nonomial\in\monomials{\Omc}$, $d\in\Delta^{\Imc_\Omc}$
s.t.\
$(d,[\nonomial])\in C^{\Imc^{i}_\Omc} $, 
then add $(d,\aux{[\monomial\times \nonomial]}{R}{},[\monomial\times \nonomial])$ to $R^{\Imc^{i}_\Omc}$.
\item \label{r3} $\alpha= (R\sqsubseteq S,\monomial)$: if there are  
$d,d'\in\Delta^{\Imc_\Omc}$, $\nonomial\in\monomials{\Omc}$ 
s.t.\
 $(d,d',\representative{\nonomial})\in R^{\Imc^{i}_\Omc}$, 
then add $(d,d',\representative{\monomial\times \nonomial})$ to $S^{\Imc^{i}_\Omc}$.
\end{enumerate}
\begin{example}\label{ex:comb}
For our running example, $\Imc_\Omc$ is as follows:
\begin{equation*}
\begin{array}{l@{}l}
A^{\Imc_\Omc}=\{&(a, u_2), (a, u_1\times v_2), (d_R^{u_2\times v_1}, u_2\times v_1\times v_2), \\&
(d_R^{u_1\times v_1\times v_2},u_1\times  v_1\times v_2), \\&
{(d_R^{u_2\times v_1\times v_2}, u_2\times v_1\times v_2)}\}\\
R^{\Imc_\Omc}=\{&(a,a,u_1), (a, d_R^{u_2\times v_1},u_2\times v_1), 
\\&
(a, d_R^{u_1\times v_1\times v_2},u_1\times  v_1\times v_2),\\&
(d_R^{u_2\times v_1},d_R^{u_2\times v_1\times v_2}, u_2\times v_1\times v_2),\\& 
(d_R^{u_1\times v_1\times v_2},d_R^{u_1\times v_1\times v_2},u_1\times  v_1\times v_2),\\&
{(d_R^{u_2\times v_1\times v_2},d_R^{u_2\times v_1\times v_2}, u_2\times v_1\times v_2)}\}.
\end{array}
\end{equation*}
\end{example}
Proposition~\ref{prop:model} formalises the fact that 
$\Imc_\Omc$ is a model of \Omc. 
\begin{restatable}{proposition}{PropositionModel}\label{prop:model}
$\Imc_\Omc$ is a model of \Omc.
\end{restatable}
We 
define the rewriting $\queryrewriting$ of 
a   query $\query$, closely following 
 \citeauthor{LTW:elcqrewriting09} \shortcite{LTW:elcqrewriting09}. 
It contains an additional  predicate ${\sf Aux}$, always interpreted 
as $(\Delta^{\Imc_\Omc}\setminus \individuals{\Omc})\times \{1^{\Imc_\Omc}\}$ 
in $\Imc_\Omc$. 
Let $\sim_{\query}$ be the smallest transitive relation over terms of~$\query$,
${\sf term}(\query)$, that includes 
identity relation, and satisfies the closure condition
\begin{itemize}[label=$(\dagger)$,leftmargin=*]
\item
$R_1(t_1,t_2,t),R_2(t'_1,t'_2,t')\in \query, 
t_2 \sim_\query t'_2 \implies t_1 \sim_\query t'_1$.
\end{itemize}
Clearly, the  relation $\sim_{\query}$ 
is computable in polynomial time in the size of $q$.
 Define for any equivalence class 
$\chi$ of $\sim_\query$, the set 
\[{\sf pre}(\chi)= \{t_1\mid \exists R\in \NR \text{ s.t. }R(t_1,t_2,t)\in q 
\text{ and } t_2\in \chi\}.\] 
We define the sets ${\sf Cyc}$ and ${\sf Fork}_=$ whose main purpose
in the translation  
 is to prevent spurious matches (e.g., with cycles)
of a query in the anonymous part of the canonical model.
\begin{itemize}
\item ${\sf Fork}_=$ is the set of pairs $ ({\sf pre}(\chi),\chi)$ with ${\sf pre}(\chi)$ of cardinality at least two. 
\item ${\sf Cyc}$ is the set of variables $x$ in ${\sf term}(q)$ such that 
there are $R_0(t^0_1,t^0_2,{t^0}),{\ldots}, R_m(t^m_1,t^m_2,{t^m}), {\ldots}, 
R_n(t^n_1,t^n_2,{t^n})$ in $q $ with $n,m\geq 0$,  $x\sim_\query t^j_1$ 
for some 
$j\leq n$, $t^i_2\sim_\query t^{i+1}_1$ for all $i < n$, 
and $t^n_2\sim_\query t^m_1$. 
\end{itemize}
${\sf Fork}_=$, 
and 
 ${\sf Cyc}$ 
can also be computed in polynomial time in the size of $q$.
For each equivalence class $\chi$ of $\sim_\query$, we choose a representative $t_\chi\in\chi$. 
For $q = \exists \vec{x}.\psi$, the rewritten query 
$\queryrewriting$ is defined as $\exists \vec{x}.(\psi\wedge 
\varphi_1\wedge \varphi_2)$,  
where
\begin{align*}
\varphi_1 :=& \bigwedge_{x\in 
{\sf Cyc}} \neg {\sf Aux}(x,1)\\
\varphi_2 :=& \bigwedge_{(\{t_1,\ldots, t_k\},\chi)\in {\sf Fork}_=}
({\sf Aux}(t_\chi,1)\rightarrow \bigwedge_{1\leq i < k} t_i = t_{i+1}). 
\end{align*}

\begin{example}\label{ex:combined}
The rewriting $\queryrewriting$ of $\query$ in Example~\ref{ex:comb} is  
$\exists xyztt't''.(R(x,x,t)\wedge R(x,y,t')\wedge R(z,y,t'')\wedge\neg {\sf Aux}(x,1)\wedge ({\sf Aux}(y,1)\rightarrow  x = z))$. 
  $\varphi_1$ prevents mapping $x$ to some $d^\monomial_R$, 
 avoiding  the $R$-loops in the anonymous part of $\Imc_\Omc$ to satisfy $R(x,x,t)$.
$\varphi_2$ enforces that if $y$ is mapped in the anonymous part, then $x$ and $z$ are mapped to the same object, 
which    avoids 
$R$-loops in the anonymous part of $\Imc_\Omc$. 
\end{example}
Our construction differs from the original rewriting of \citeauthor{LTW:elcqrewriting09} \shortcite{LTW:elcqrewriting09}. 
In particular, in their rewriting there is a 
formula $\varphi_3$, which is 
not necessary in our case. Intuitively, 
this is because we keep the information of the role name used to connect an element of 
 $\auxs{\Omc}$ to the rest of the model.   
Theorem~\ref{thm:combined} establishes that $\queryrewriting$ is as~required.

\begin{restatable}{theorem}{Theoremcombined}\label{thm:combined}
Let \Omc be an  
ontology in normal form and 
$(\query,\polynomial)$ be an annotated 
  query. 
Then, $\Omc\models (\query,\polynomial) \text{ iff } \Imc_\Omc\models (\queryrewriting,\polynomial). $
\end{restatable}
Although the domain of monomials is infinite, since only elements of $\monomials{\Omc}$ are relevant,
an exponential size structure representing $\Imc_\Omc$ is sufficient 
to check 
whether $\Imc_\Omc\models (\queryrewriting,\polynomial)$. 
The size of the resulting structure is exponential in $|\Omc|$ and can be constructed 
in exponential time using the completion algorithm (Theorem~\ref{prop:completionalgorithm}) 
to check entailment of 
assertions and~RRs. 

\begin{corollary}\label{cor:combinedcomplexity}
Let \Omc be an ontology, $q$ a BCQ and $p\in\NPr$.
$\Omc\models (q,p)$ is decidable in exponential time in $|\Omc|+|(q,p)|$. 
\end{corollary}
 
 \section{Discussion and Conclusions}
We study the problem of computing the provenance of an axiom or a BCQ entailment from \ELHr ontologies. 
In particular, 
entailment of annotated axioms or IQs for a fixed monomial size
is tractable, and 
the set of relevant provenance variables can be computed in polynomial time. 
For the more challenging problem of CQ answering, we adapt the 
combined approach.

\mypar{Related work}
Explaining inferences in DLs has been studied mostly focusing on explaining axiom entailment, in particular concept subsumption, through \emph{axiom pinpointing} \cite{ScCo-03,KalyanpurPHS07,BaPS07}. Few approaches address \emph{query answer explanation} for $\DLLite$ or existential rules \cite{BorgidaCR08,CroceL18,CeylanLMV19,BienvenuBG19}. 
However, current explanation services in DLs provide \emph{minimal} explanations, which is crucially different to provenance, since provenance takes into account \emph{all} derivations (cf. discussion 
in Section \ref{sec:relevance}).

Closest to our work is 
provenance for OBDA \cite{CalvaneseLantiOzakiPenalozaXiao19}. However, the 
 challenges in enriching the \EL family with provenance 
were not investigated. 
We also study 
additional problems such as axiom entailment and relevance. 
\citeauthor{DannertG19}~\shortcite{DannertG19} consider   provenance in 
the 
DL $\ALC$.
The setting is not the same as ours since they only consider annotated assertions (not annotated GCIs), 
do not study BCQs, and 
the semantics is different as well.  
There are several proposals for handling provenance in RDF(S), most notably an algebraic 
deductive system for annotated RDFS 
\cite{BunemanK10}. 
The approach by \citeauthor{bourgauxozaki} \shortcite{bourgauxozaki} for attributed $\DLLite$ 
 fundamentally differs by using GCIs and RIs to express constraints on provenance. 
 
\section*{Acknowledgments}
This work was supported by Camille Bourgaux's CNRS PEPS grant and contract ANR-18-CE23-0003 (CQFD), by the University of Bergen, by the Free University of Bozen-Bolzano (Unibz) projects PROVDL and FO2S, and by the Italian PRIN project HOPE.

\bibliographystyle{named}
\bibliography{ms}

\begin{thebibliography}{}

\bibitem[\protect\citeauthoryear{Baader \bgroup \em et al.\egroup
  }{2005}]{BBL-EL}
Franz Baader, Sebastian Brandt, and Carsten Lutz.
\newblock Pushing the $\mathcal{EL}$ envelope.
\newblock In {\em {IJCAI}}, 2005.

\bibitem[\protect\citeauthoryear{Baader \bgroup \em et al.\egroup
  }{2007a}]{dlhandbook}
Franz Baader, Diego Calvanese, Deborah McGuinness, Daniele Nardi, and Peter
  Patel-Schneider, editors.
\newblock {\em The Description Logic Handbook: Theory, Implementation, and
  Applications}.
\newblock Cambridge University Press, second edition, 2007.

\bibitem[\protect\citeauthoryear{Baader \bgroup \em et al.\egroup
  }{2007b}]{BaPS07}
Franz Baader, Rafael Pe{\~n}aloza, and Boontawee Suntisrivaraporn.
\newblock Pinpointing in the description logic $\mathcal{EL}^+$.
\newblock In {\em {KI}}, 2007.

\bibitem[\protect\citeauthoryear{Baader \bgroup \em et al.\egroup
  }{2008a}]{BBL-EL08}
Franz Baader, Sebastian Brandt, and Carsten Lutz.
\newblock Pushing the $\mathcal{EL}$ envelope further.
\newblock In {\em OWLED}, 2008.

\bibitem[\protect\citeauthoryear{Baader \bgroup \em et al.\egroup
  }{2008b}]{BaHP08}
Franz Baader, Jan Hladik, and Rafael Pe{\~{n}}aloza.
\newblock Automata can show {PSpace} results for description logics.
\newblock {\em Inf. Comput.}, 206(9-10), 2008.

\bibitem[\protect\citeauthoryear{Bienvenu \bgroup \em et al.\egroup
  }{2013}]{BienetalFO2013}
Meghyn Bienvenu, Carsten Lutz, and Frank Wolter.
\newblock {First-Order Rewritability of Atomic Queries in Horn Description
  Logics}.
\newblock In {\em {IJCAI}}, 2013.

\bibitem[\protect\citeauthoryear{Bienvenu \bgroup \em et al.\egroup
  }{2019}]{BienvenuBG19}
Meghyn Bienvenu, Camille Bourgaux, and Fran{\c{c}}ois Goasdou{\'{e}}.
\newblock Computing and explaining query answers over inconsistent {DL-Lite}
  knowledge bases.
\newblock {\em J. Artif. Intell. Res.}, 64, 2019.

\bibitem[\protect\citeauthoryear{Borgida \bgroup \em et al.\egroup
  }{2008}]{BorgidaCR08}
Alexander Borgida, Diego Calvanese, and Mariano Rodriguez{-}Muro.
\newblock Explanation in the {DL-Lite} family of description logics.
\newblock In {\em {OTM}}, 2008.

\bibitem[\protect\citeauthoryear{Bourgaux and Ozaki}{2019}]{bourgauxozaki}
Camille Bourgaux and Ana Ozaki.
\newblock Querying attributed {DL}-{L}ite ontologies using provenance
  semirings.
\newblock In {\em AAAI}, 2019.

\bibitem[\protect\citeauthoryear{Buneman and Kostylev}{2010}]{BunemanK10}
Peter Buneman and Egor~V. Kostylev.
\newblock Annotation algebras for {RDFS} data.
\newblock In {\em SWPM@ISWC}, 2010.

\bibitem[\protect\citeauthoryear{Buneman}{2013}]{Bun2013}
Peter Buneman.
\newblock The providence of provenance.
\newblock In {\em {BNCOD}}, 2013.

\bibitem[\protect\citeauthoryear{Calvanese \bgroup \em et al.\egroup
  }{2019}]{CalvaneseLantiOzakiPenalozaXiao19}
Diego Calvanese, Davide Lanti, Ana Ozaki, Rafael Pe{\~{n}}aloza, and Guohui
  Xiao.
\newblock Enriching ontology-based data access with provenance.
\newblock In {\em {IJCAI}}, 2019.

\bibitem[\protect\citeauthoryear{Ceylan \bgroup \em et al.\egroup
  }{2019}]{CeylanLMV19}
{\.I}smail~{\.I}lkan Ceylan, Thomas Lukasiewicz, Enrico Malizia, and Andrius
  Vaicenavicius.
\newblock Explanations for query answers under existential rules.
\newblock In {\em {IJCAI}}, 2019.

\bibitem[\protect\citeauthoryear{Cheney \bgroup \em et al.\egroup
  }{2009}]{DBLP:journals/ftdb/CheneyCT09}
James Cheney, Laura Chiticariu, and Wang~Chiew Tan.
\newblock Provenance in databases: Why, how, and where.
\newblock {\em Foundations and Trends in Databases}, 1(4), 2009.

\bibitem[\protect\citeauthoryear{Croce and Lenzerini}{2018}]{CroceL18}
Federico Croce and Maurizio Lenzerini.
\newblock A framework for explaining query answers in {DL-Lite}.
\newblock In {\em {EKAW}}, 2018.

\bibitem[\protect\citeauthoryear{Dannert and Gr{\"{a}}del}{2019}]{DannertG19}
Katrin~M. Dannert and Erich Gr{\"{a}}del.
\newblock Provenance analysis: {A} perspective for description logics?
\newblock In {\em Description Logic, Theory Combination, and All That}, 2019.

\bibitem[\protect\citeauthoryear{Deutch \bgroup \em et al.\egroup
  }{2014}]{DeutchMRT14}
Daniel Deutch, Tova Milo, Sudeepa Roy, and Val Tannen.
\newblock Circuits for datalog provenance.
\newblock In {\em {ICDT}}, 2014.

\bibitem[\protect\citeauthoryear{Green and Tannen}{2017}]{GreenT17}
Todd~J. Green and Val Tannen.
\newblock The semiring framework for database provenance.
\newblock In {\em PODS}, 2017.

\bibitem[\protect\citeauthoryear{Green \bgroup \em et al.\egroup
  }{2007}]{Green07-provenance-seminal}
Todd~J. Green, Gregory Karvounarakis, and Val Tannen.
\newblock Provenance semirings.
\newblock In {\em {PODS}}, 2007.

\bibitem[\protect\citeauthoryear{Green}{2011}]{GreenTheoCompSyst2011}
Todd~J. Green.
\newblock Containment of conjunctive queries on annotated relations.
\newblock {\em Theory of Computing Systems}, 49, 2011.

\bibitem[\protect\citeauthoryear{Hutschenreiter and
  Pe{\~{n}}aloza}{2017}]{HuPe-17}
Lisa Hutschenreiter and Rafael Pe{\~{n}}aloza.
\newblock An automata view to goal-directed methods.
\newblock In {\em {LATA}}, 2017.

\bibitem[\protect\citeauthoryear{Ives \bgroup \em et al.\egroup
  }{2008}]{Ivesetal2008}
Zachary~G. Ives, Todd~J. Green, Grigoris Karvounarakis, Nicholas~E. Taylor, Val
  Tannen, Partha~Pratim Talukdar, Marie Jacob, and Fernando C.~N. Pereira.
\newblock The {ORCHESTRA} collaborative data sharing system.
\newblock {\em Sigmod Record}, 37, 2008.

\bibitem[\protect\citeauthoryear{Kalyanpur \bgroup \em et al.\egroup
  }{2007}]{KalyanpurPHS07}
Aditya Kalyanpur, Bijan Parsia, Matthew Horridge, and Evren Sirin.
\newblock Finding all justifications of {OWL} {DL} entailments.
\newblock In {\em {ISWC}}, 2007.

\bibitem[\protect\citeauthoryear{Lukasiewicz \bgroup \em et al.\egroup
  }{2014}]{LMMPS2014}
Thomas Lukasiewicz, Maria~Vanina Mart\'{i}nez, Cristian Molinaro, Livia
  Predoiu, and Gerardo~I. Simari.
\newblock Answering ontological ranking queries based on subjective reports.
\newblock In {\em {SUM}}, 2014.

\bibitem[\protect\citeauthoryear{Lutz \bgroup \em et al.\egroup
  }{2009}]{LTW:elcqrewriting09}
Carsten Lutz, David Toman, and Frank Wolter.
\newblock Conjunctive query answering in the description logic $\mathcal{EL}$
  using a relational database system.
\newblock In {\em {IJCAI}}, 2009.

\bibitem[\protect\citeauthoryear{Pe{\~{n}}aloza and
  Sertkaya}{2010}]{PenalozaSertkaya10}
Rafael Pe{\~{n}}aloza and Baris Sertkaya.
\newblock On the complexity of axiom pinpointing in the {EL} family of
  description logics.
\newblock In {\em {KR}}, 2010.

\bibitem[\protect\citeauthoryear{Pe{\~{n}}aloza \bgroup \em et al.\egroup
  }{2017}]{PenalozaMIM17}
Rafael Pe{\~{n}}aloza, Carlos Menc{\'{\i}}a, Alexey Ignatiev, and Jo{\~{a}}o
  Marques{-}Silva.
\newblock Lean kernels in description logics.
\newblock In {\em {ESWC}}, 2017.

\bibitem[\protect\citeauthoryear{Ramusat \bgroup \em et al.\egroup
  }{2018}]{RamusatMS18}
Yann Ramusat, Silviu Maniu, and Pierre Senellart.
\newblock Semiring provenance over graph databases.
\newblock In {\em {TaPP}}, 2018.

\bibitem[\protect\citeauthoryear{Schlobach and Cornet}{2003}]{ScCo-03}
Stefan Schlobach and Ronald Cornet.
\newblock Non-standard reasoning services for the debugging of description
  logic terminologies.
\newblock In {\em {IJCAI}}, 2003.

\bibitem[\protect\citeauthoryear{Senellart}{2017}]{Senellart17}
Pierre Senellart.
\newblock Provenance and probabilities in relational databases.
\newblock {\em {SIGMOD} Record}, 46(4), 2017.

\bibitem[\protect\citeauthoryear{Suciu \bgroup \em et al.\egroup
  }{2011}]{Suciuetal2011}
Dan Suciu, Dan Olteanu, Christopher R\'{e}, and Christoph Koch.
\newblock {\em Probabilistic Databases}.
\newblock Synthesis Lectures on Data Management, Morgan and Claypool
  Publishers, 2011.

\bibitem[\protect\citeauthoryear{Zimmermann \bgroup \em et al.\egroup
  }{2012}]{DBLP:journals/ws/ZimmermannLPS12}
Antoine Zimmermann, Nuno Lopes, Axel Polleres, and Umberto Straccia.
\newblock A general framework for representing, reasoning and querying with
  annotated semantic web data.
\newblock {\em J. Web Sem.}, 11, 2012.

\end{thebibliography}

\section*{Proofs for Section \ref{sec:completion}}

\Normalisation*
\begin{proof}
For the first point, we show by induction on $i$ 
that for every annotated \ELHr ontology $\Omc$, 
if 
$\Omc'$ is obtained from $\Omc$ by applying $i$ normalization steps, then $\Omc'\models \Omc$. 
If $\Omc'$ 
is obtained from $\Omc$ by applying a single normalization rule $\NF$, then we have the three following possibilities.
\begin{description}
\item[$\NF=\NF_1$] let $\Imc$ be a model of $\Omc'=\Omc\setminus\{(C\sqcap\widehat{D}\sqsubseteq E,\, v)\}
\cup\{(\widehat{D}\sqsubseteq A,\, 1), \, (C\sqcap A\sqsubseteq E, \, v)\}$ and $(d,\nonomial^\Imc)\in(C\sqcap\widehat{D})^\Imc$. 
There exist $\nonomial_1,\nonomial_2 \in \NM$ such that $(d,\nonomial_1^\Imc)\in C^\Imc$, $(d,\nonomial_2^\Imc)\in \widehat{D}^\Imc$, 
and $(\nonomial_1\times \nonomial_2)^\Imc=\nonomial^\Imc$. 
Since $\Imc \models(\widehat D\sqsubseteq A,1)$, then $(d,\nonomial_2^\Imc)\in A^\Imc$.
Hence $(d,(\nonomial_1\times \nonomial_2)^\Imc)\in (C\sqcap A)^\Imc$. 
Since $\Imc\models (C\sqcap A\sqsubseteq E, \, v)$, it follows that 
$(d,(\nonomial_1\times \nonomial_2\times v)^\Imc)\in E^\Imc$, i.e., $(d,(\nonomial\times v)^\Imc)\in E^\Imc$. 
Thus $\Imc\models (C\sqcap\widehat{D}\sqsubseteq E,\, v)$, and $\Imc$ is a model of $\Omc$. 

\item[$\NF=\NF_2$] let $\Imc$ be a model of 
$\Omc'=\Omc\setminus\{(\exists R.\widehat{C}\sqsubseteq D,\, v)\}\cup\{(\widehat{C}\sqsubseteq A,\,
 1), \, (\exists R. A\sqsubseteq D, \, v)\}$ and $(d,\nonomial^\Imc)\in(\exists R.\widehat{C})^\Imc$.
  There exist $e \in \Delta^\Imc$ and $\nonomial_1,\nonomial_2 \in \NM$ such that $(d, e,\nonomial_1^\Imc)\in R^\Imc$, $(e,\nonomial_2^\Imc)\in \widehat{C}^\Imc$, and $(\nonomial_1\times \nonomial_2)^\Imc=\nonomial^\Imc$. Thus $(e,\nonomial_2^\Imc)\in A^\Imc$, and $(d,(\nonomial_1\times \nonomial_2)^\Imc)\in (\exists R. A)^\Imc$ so $(d,(\nonomial_1\times \nonomial_2\times v)^\Imc)\in D^\Imc$, i.e. $(d,(\nonomial\times v)^\Imc)\in D^\Imc$. It follows that $\Imc\models (\exists R.\widehat{C}\sqsubseteq D,\, v)$, 
  and hence $\Imc$ is a model of~$\Omc$. 

\item[$\NF=\NF_3$] let $\Imc$ be a model of 
$\Omc'=\Omc\setminus\{(\widehat{C}\sqsubseteq\exists R,\, v)\}\cup\{(\widehat{C}\sqsubseteq A,\, 1), 
\, (A\sqsubseteq \exists R, \, v)\}$ and $(d,\nonomial^\Imc)\in\widehat{C}^\Imc$. We have 
$(d,\nonomial^\Imc)\in A^\Imc$, so $(d,(\nonomial\times v)^\Imc)\in (\exists R)^\Imc$. It follows that 
$\Imc\models(\widehat{C}\sqsubseteq\exists R,\, v)$, and $\Imc$ is a model of $\Omc$. 
\end{description}
Assume that the property is true for $i$ 
and let $\Omc$ be an annotated \ELHr  ontology, $\Omc'$ be the result of applying $i+1$ normalization steps to $\Omc$ and $\Omc''$ that of applying the first $i$ steps to $\Omc$.  
By applying the induction hypothesis on $\Omc''$, we obtained that $\Omc''\models \Omc$. 
Then since $\Omc'$ results from applying one normalization step to $\Omc''$, we obtain that $\Omc'\models \Omc''$, and thus $\Omc'\models \Omc$. 

Conversely, we show by induction that for every $i$ and for every annotated \ELHr  ontology $\Omc$, 
 if $\Omc'$ can be obtained from $\Omc$ by applying $i$ normalization rules, then if $\Omc'\models (\alpha,\monomial)$ and every concept name occurring in $\alpha$ occurs in $\Omc$, 
then $\Omc\models (\alpha,\monomial)$. 
If $\Omc'$ is obtained from $\Omc$ by applying a single normalization rule $\NF$, we have three possible cases. 
\begin{description}
\item[$\NF=\NF_1$] $\Omc'=\Omc\setminus\{(C\sqcap\widehat{D}\sqsubseteq E,\, v)\}\cup\{(\widehat{D}\sqsubseteq A,\, 1), \, (C\sqcap A\sqsubseteq E, \, v)\}$. Let $\Imc$ be a model of $\Omc$ and $\Jmc$ the interpretation that extends $\Imc$ with $A^\Jmc=\widehat{D}^\Imc$. 
Clearly, $\Jmc\models (\widehat{D}\sqsubseteq A,\, 1)$. Let $(d,\nonomial^\Jmc)\in (C\sqcap A)^\Jmc$. There exist $\nonomial_1,\nonomial_2  \in \NM$ such that  
$(d,\nonomial_1^\Jmc)\in C^\Jmc$, $(d,\nonomial_2^\Jmc)\in A^\Jmc$, and $(\nonomial_1\times \nonomial_2)^\Jmc=\nonomial^\Jmc$.
Since $A^\Jmc=\widehat{D}^\Jmc$, then $(d,\nonomial_2^\Jmc)\in \widehat{D}^\Jmc$ so $(d,(\nonomial_1\times \nonomial_2)^\Jmc)\in (C\sqcap\widehat{D})^\Jmc$. 
Since $\Imc\models (C\sqcap\widehat{D}\sqsubseteq E,\, v)$ (and thus also $\Jmc$), it follows that $(d,(\nonomial_1\times \nonomial_2\times v)^\Jmc)\in E^\Jmc$, i.e. $(d,(\nonomial\times v)^\Jmc)\in E^\Jmc$. 
Hence $\Jmc\models (C\sqcap A\sqsubseteq E, \, v)$ and $\Jmc$ is a model of $\Omc'$. 
It follows that $\Jmc\models (\alpha,\monomial)$. Since $\alpha$ does not contain $A$, then $\Imc\models (\alpha,\monomial)$, so $\Omc\models (\alpha,\monomial)$.

\item[$\NF=\NF_2$] $\Omc'=\Omc\setminus\{(\exists R.\widehat{C}\sqsubseteq D,\, v)\}\cup\{(\widehat{C}\sqsubseteq A,\, 1), \, (\exists R. A\sqsubseteq D, \, v)\}$. Let $\Imc$ be a model of $\Omc$ and $\Jmc$ the interpretation that extends $\Imc$ with $A^\Jmc=\widehat{C}^\Imc$. 
Clearly, $\Jmc\models (\widehat{C}\sqsubseteq A,\, 1)$. Let $(d,\nonomial^\Jmc)\in (\exists R. A)^\Jmc$. There exist $e  \in \Delta^\Jmc$ and $\nonomial_1,\nonomial_2 \in \NM$ such that  $(d,e,\nonomial_1^\Jmc)\in R^\Jmc$, $(e,\nonomial_2^\Jmc)\in A^\Jmc$, and $(\nonomial_1\times \nonomial_2)^\Jmc=\nonomial^\Jmc$. Thus $(e,\nonomial_2^\Jmc)\in \widehat{C}^\Jmc$ and $(d,(\nonomial_1\times \nonomial_2)^\Jmc)\in (\exists R. \widehat{C})^\Jmc$. Since $\Imc\models (\exists R.\widehat{C}\sqsubseteq D,\, v)$ (and thus also $\Jmc$), it follows that $(d,(\nonomial_1\times \nonomial_2\times v)^\Jmc)\in D^\Jmc$, i.e. $(d,(\nonomial\times v)^\Jmc)\in D^\Jmc$. 
Hence $\Jmc\models (\exists R. A\sqsubseteq D, \, v)$ and $\Jmc$ is a model of $\Omc'$. 
It follows that $\Jmc\models (\alpha,\monomial)$. Since $\alpha$ does not contain $A$, then $\Imc\models (\alpha,\monomial)$, so $\Omc\models (\alpha,\monomial)$.

\item[$\NF=\NF_3$] $\Omc'=\Omc\setminus\{(\widehat{C}\sqsubseteq\exists R,\, v)\}\cup\{(\widehat{C}\sqsubseteq A,\, 1), \, (A\sqsubseteq \exists R, \, v)\}$. Let $\Imc$ be a model of $\Omc$ and $\Jmc$ be the interpretation that extends $\Imc$ with $A^\Jmc=\widehat{C}^\Imc$. 
Clearly, $\Jmc\models (\widehat{C}\sqsubseteq A,\, 1)$. Let $(d,\nonomial^\Jmc)\in A^\Jmc=\widehat{C}^\Jmc$. Since $\Imc\models (\widehat{C}\sqsubseteq\exists R,\, v)$ (and thus also $\Jmc$), it follows that $(d,(\nonomial\times v)^\Jmc)\in (\exists R)^\Jmc$. 
Hence $\Jmc\models (A\sqsubseteq \exists R, \, v)$ and $\Jmc$ is a model of $\Omc'$. 
It follows that $\Jmc\models (\alpha,\monomial)$. Since $\alpha$ does not contain $A$, then $\Imc\models (\alpha,\monomial)$, so $\Omc\models (\alpha,\monomial)$.
\end{description}
Assume that the property is true for some $i$ and let $\Omc$ be an annotated \ELHr  ontology. 
Let $\Omc_{i+1}$ 
be obtained by applying $i+1$ normalization rules to $\Omc$ and $\Omc_{i}$ 
be obtained by applying the first $i$ normalization rules to $\Omc$  (so that $\Omc_{i+1}$ is obtained by applying a normalization rule $\NF$ to $\Omc_i$). 
Assume that $\Omc_{i+1}\models (\alpha,\monomial)$ and every concept name occurring in $\alpha$ occurs in $\Omc$.  
Since the normalization rules can only introduce new concept names, the concept names occurring in $\Omc$ are a sunset of those occurring in $\Omc_i$, so every concept name occurring in $\alpha$ occurs in $\Omc_i$. 
Since we thus have that $\Omc_{i+1}\models (\alpha,\monomial)$, 
that $\Omc_{i+1}$ results from the application of a single normalization rule to $\Omc_i$, and that all concept names in $\alpha$ occur in $\Omc_i$,  
the base case applies and we obtain that $\Omc_i\models (\alpha,\monomial)$. Hence by the induction hypothesis, $\Omc\models (\alpha,\monomial)$. 
\end{proof}

We now proceed to prove Theorem~\ref{prop:reductions} by showing that annotated GCI, RI and RR entailments and annotated 
assertion entailment can be reduced to each other in polynomial time. To increase readability, we divide the proof in several
cases, and use the following lemmas.

\begin{lemma}\label{lem:conj-right}
Let $\Omc$ be an \ELHr ontology, $C, C_1, C_2, D_1, D_2, D, D'$ be \ELHr concepts, $S,R\in\NR$, and $\monomial,\monomial_1,\monomial_2\in\NM$. 
\begin{enumerate}
\item If $\Omc\models (C\sqsubseteq C_1\sqcap C_2, \monomial)$, $\Omc\models (C_1 \sqsubseteq D_1, \monomial_1)$ and $\Omc\models (C_2 \sqsubseteq D_2, \monomial_2)$, then $\Omc\models (C\sqsubseteq D_1\sqcap D_2, \monomial\times\monomial_1\times\monomial_2)$.
\item If $\Omc\models (C\sqsubseteq \exists S. D', \monomial)$ and $\Omc\models (D' \sqsubseteq D, \monomial_1)$, then $\Omc\models (C\sqsubseteq \exists S.D, \monomial\times\monomial_1)$.
\item If $\Omc\models (C\sqsubseteq \exists R. D, \monomial)$ and $\Omc\models (R\sqsubseteq S, \monomial_1)$, then $\Omc\models (C\sqsubseteq \exists S.D, \monomial\times\monomial_1)$.
\end{enumerate}
\end{lemma}
\begin{proof}
(1) Let $\Imc$ be a model of $\Omc$ and $(e,\nonomial^\Imc)\in C^\Imc$. It holds that $(e,(\nonomial\times\monomial)^\Imc)\in (C_1\sqcap C_2)^\Imc$ so there exists $\onomial_1,\onomial_2\in\NM$ such that  $(e,\onomial_1^\Imc)\in C_1^\Imc$, $(e,\onomial_2^\Imc)\in C_2^\Imc$, and $(\onomial_1\times\onomial_2)^\Imc=(\nonomial\times\monomial)^\Imc$. 
Hence $(e,(\onomial_1\times\monomial_1)^\Imc)\in D_1^\Imc$ and $(e,(\onomial_2\times\monomial_2)^\Imc)\in D_2^\Imc$. 
It follows that $(e,(\onomial_1\times\monomial_1\times\onomial_2\times\monomial_2)^\Imc)\in (D_1\sqcap D_2)^\Imc$, i.e. $(e,(\nonomial\times\monomial\times\monomial_1\times\monomial_2)^\Imc)\in(D_1\sqcap D_2)^\Imc$. 
Thus $\Omc\models(C\sqsubseteq D_1\sqcap D_2, \monomial\times\monomial_1\times\monomial_2)$.

\noindent (2) Let $\Imc$ be a model of $\Omc$ and $(e,\nonomial^\Imc)\in C^\Imc$. It holds that $(e,(\nonomial\times\monomial)^\Imc)\in (\exists S.D')^\Imc$ so there exists $f\in\Delta^\Imc$ and $\onomial_1,\onomial_2\in\NM$ such that  $(e,f,\onomial_1^\Imc)\in S^\Imc$, $(f,\onomial_2^\Imc)\in {D'}^\Imc$, and $(\onomial_1\times\onomial_2)^\Imc=(\nonomial\times\monomial)^\Imc$. 
Hence $(f,(\onomial_2\times\monomial_1)^\Imc)\in {D}^\Imc$. 
It follows that $(e,(\onomial_1\times\onomial_2\times\monomial_1)^\Imc)\in (\exists S.D)^\Imc$, i.e. $(e,(\nonomial\times\monomial\times\monomial_1)^\Imc)\in (\exists S.D)^\Imc$. 
Thus $\Omc\models (C\sqsubseteq \exists S.D, \monomial\times\monomial_1)$.

\noindent(3) Let $\Imc$ be a model of $\Omc$ and $(e,\nonomial^\Imc)\in C^\Imc$. 
It holds that $(e,(\nonomial\times\monomial)^\Imc)\in (\exists R.D)^\Imc$ so there exists $f\in\Delta^\Imc$ and $\onomial_1,\onomial_2\in\NM$ such that  $(e,f,\onomial_1^\Imc)\in R^\Imc$, $(f,\onomial_2^\Imc)\in {D}^\Imc$, and $(\onomial_1\times\onomial_2)^\Imc=(\nonomial\times\monomial)^\Imc$. 
Hence $(e,f,(\onomial_1\times \monomial_1)^\Imc)\in S^\Imc$. 
It follows that $(e,(\onomial_1\times \monomial_1\times\onomial_2)^\Imc)\in (\exists S.D)^\Imc$, i.e. $(e,( \nonomial\times\monomial\times \monomial_1)^\Imc)\in (\exists S.D)^\Imc$. Thus $\Omc\models (C\sqsubseteq \exists S.D, \monomial\times\monomial_1)$.
\end{proof}

Given an ontology $\Omc$ in normal form and an interpretation $\Imc_i$ that interprets the monomials by their representatives, define the three following rules to build an interpretation $\Imc_{i+1}$ that extends~$\Imc_i$. Let $A\in\NC\cup\{\top\}$, $B\in\NC$, $R,S\in\NR$ and $G$ be an \ELHr concept or ${\sf ran}(R)$ for some $R\in\NR$.
\begin{description}
\item[$R1$.] If $ (G\sqsubseteq B, v)\in\Omc$, $(d,\monomial)\in G^{\Imc_i}$, and $(d, \representative{v\times\monomial})\notin B^{\Imc_i}$, then $B^{\Imc_{i+1}}=B^{\Imc_i}\cup\{(d, \representative{v\times\monomial})\}$. 

\item[$R2$.] If $(R\sqsubseteq S, v)\in\Omc$, $(d,f,\monomial)\in R^{\Imc_i}$, and $(d,f, \representative{v\times\monomial})\notin S^{\Imc_i}$, then $S^{\Imc_{i+1}}=S^{\Imc_i}\cup\{(d,f, \representative{v\times\monomial})\}$. 

\item[$R3$.] If $(A\sqsubseteq \exists S, v)\in\Omc$,  $(d,\monomial)\in A^{\Imc_i}$, and there is no $f\in\Delta^{\Imc_i}$ such that $(d,f, \representative{v\times\monomial})\in S^{\Imc_i}$, then $\Delta^{\Imc_{i+1}}=\Delta^{\Imc_i}\cup\{f\}$ where $f\notin\Delta^{\Imc_i}$ and $S^{\Imc_{i+1}}=S^{\Imc_i}\cup\{(d,f, \representative{v\times\monomial})\}$. 
\end{description}

\begin{lemma}\label{lem:complinter}
Let $C$ be an \ELHr concept. 
If $\Imc_{i+1}$ is obtained by applying a rule in $\{R1,R2,R3\}$ to $\Imc_i$, and $(d, \monomial)\in C^{\Imc_{i+1}}$ while $(d, \monomial)\notin C^{\Imc_{i}}$, then there exists an \ELHr concept $D$, a provenance variable $v\in\NV$ and a monomial $\nonomial$ such that $\Omc\models (D\sqsubseteq C,v)$, $(d, \nonomial)\in D^{\Imc_{i}}$, and $\representative{v\times\nonomial}=\monomial$.
\end{lemma}
\begin{proof}
We prove a stronger version of the property, requiring that $v$ is the provenance of the GCI, RI or RR of the rule applied to obtain $\Imc_{i+1}$. 
The proof is by structural induction. 
For the base case, $C\in\NC$, the rule applied is necessarily $R1$. 
Thus by the conditions of applicability of $R1$, since $R1$ adds $(d, \monomial)$ in $C^{\Imc_{i+1}}$, there exist $(D\sqsubseteq C,v)\in \Omc$ and $(d,\nonomial)\in D^{\Imc_i}$ such that $\representative{v\times\nonomial}=\monomial$. 

If $C=C_1\sqcap C_2$, since $(d, \monomial)\in C^{\Imc_{i+1}}$, then 
there exist $(d, \monomial_1)\in C_1^{\Imc_{i+1}}$ and $(d, \monomial_2)\in C_2^{\Imc_{i+1}}$ such that $\representative{\monomial_1\times\monomial_2}=\monomial$. 
Since $(d, \monomial)\notin C^{\Imc_{i}}$, then either $(d, \monomial_1)\notin C_1^{\Imc_{i}}$ or $(d, \monomial_2)\notin C_2^{\Imc_{i}}$. 
Let $v$ be the provenance of the GCI, RI or RR of the rule applied to obtain $\Imc_{i+1}$. 
By induction, we obtain the following. 
\begin{itemize}
\item If $(d, \monomial_1)\notin C_1^{\Imc_{i}}$, there exists an \ELHr concept $D_1$ and a monomial $\nonomial_1$ such that $\Omc\models (D_1\sqsubseteq C_1,v)$, $(d, \nonomial_1)\in D_1^{\Imc_{i}}$, and $\representative{v\times\nonomial_1}=\monomial_1$.
\item If $(d, \monomial_2)\notin C_2^{\Imc_{i}}$, there exists an \ELHr concept $D_2$ and a monomial $\nonomial_2$ such that $\Omc\models (D_2\sqsubseteq C_2,v)$, $(d, \nonomial_2)\in D_2^{\Imc_{i}}$, and $\representative{v\times\nonomial_2}=\monomial_2$.
\end{itemize}
We thus have three possible cases.
\begin{itemize}
\item In the case $(d, \monomial_1)\notin C_1^{\Imc_{i}}$ and $(d, \monomial_2)\in C_2^{\Imc_{i}}$, it holds that $(d, \representative{\nonomial_1\times\monomial_2})\in (D_1\sqcap C_2)^{\Imc_{i}}$. Moreover, since $\Omc\models (D_1\sqsubseteq C_1,v)$ and $\Omc\models (D_1\sqcap C_2\sqsubseteq D_1\sqcap C_2,1)$, it follows by Lemma \ref{lem:conj-right} that $\Omc\models (D_1\sqcap C_2\sqsubseteq C_1\sqcap C_2,v)$. 
Finally, $\representative{v\times\nonomial_1\times\monomial_2}=\representative{\monomial_1\times\monomial_2}=\monomial$.

\item The case $(d, \monomial_1)\in C_1^{\Imc_{i}}$ and $(d, \monomial_2)\notin C_2^{\Imc_{i}}$ is analogous. 

\item In the case $(d, \monomial_1)\notin C_1^{\Imc_{i}}$ and $(d, \monomial_2)\notin C_2^{\Imc_{i}}$, it holds that $(d, \representative{\nonomial_1\times\nonomial_2})\in (D_1\sqcap D_2)^{\Imc_{i}}$. Moreover, since $\Omc\models (D_1\sqsubseteq C_1,v)$, $\Omc\models (D_2\sqsubseteq C_2,v)$ and $\Omc\models (D_1\sqcap D_2\sqsubseteq D_1\sqcap D_2,1)$, it follows by Lemma \ref{lem:conj-right} that $\Omc\models (D_1\sqcap D_2\sqsubseteq C_1\sqcap C_2,v)$. 
Finally, $\representative{v\times\nonomial_1\times\nonomial_2}=\representative{\monomial_1\times\monomial_2}=\monomial$.

\end{itemize}

If $C=\exists S.C'$, since $(d, \monomial)\in C^{\Imc_{i+1}}$, then 
there exist $(d,f, \monomial_1)\in S^{\Imc_{i+1}}$ and $(f, \monomial_2)\in {C'}^{\Imc_{i+1}}$ such that $\representative{\monomial_1\times\monomial_2}=\monomial$. 
Since $(d, \monomial)\notin C^{\Imc_{i}}$, then either $(d,f, \monomial_1)\notin S^{\Imc_{i}}$ or $(f, \monomial_2)\notin {C'}^{\Imc_{i}}$. 
Let $v$ be the provenance of the GCI, RI or RR of the rule applied to obtain $\Imc_{i+1}$. 
We obtain the following. 
\begin{itemize}
\item If $(d, f,\monomial_1)\notin S^{\Imc_{i}}$, we have two cases, depending on the rule that has been applied. 
\begin{itemize}[leftmargin=5.5mm]
\item[$R2$] There exist $(R\sqsubseteq S, v)\in\Omc$ and $(d, f,\nonomial_1)\in R^{\Imc_{i}}$ such that $\representative{v\times\nonomial_1}=\monomial_1$. 
In this case, $(d, \representative{\nonomial_1\times\monomial_2})\in (\exists R.C')^{\Imc_{i}}$. 
Moreover, since $\Omc\models (R\sqsubseteq S, v)$ and $\Omc\models (\exists R.C'\sqsubseteq \exists R.C', 1)$, it follows by Lemma \ref{lem:conj-right} that $\Omc\models (\exists R.C'\sqsubseteq \exists S.C', v)$. 
Finally,  $\representative{v\times\nonomial_1\times\monomial_2}=\representative{\monomial_1\times\monomial_2}=\monomial$.

\item[$R3$]  There exist $(A\sqsubseteq \exists S, v)\in\Omc$ and $(d, \nonomial_1)\in A^{\Imc_{i}}$ such that $\representative{v\times\nonomial_1}=\monomial_1$. Moreover, in this case $f\notin\Delta^{\Imc_i}$ and occurs only in $(d,f, \monomial_1)\in S^{\Imc_{i+1}}$ in $\Imc_{i+1}$, which implies that $C'=\top$ and $\monomial_2=1$, so that $C=\exists S.\top$ and $\representative{\monomial_1}=\monomial$, i.e.  $\representative{v\times\nonomial_1}=\monomial$.
\end{itemize}
\item If $(f, \monomial_2)\notin {C'}^{\Imc_{i}}$, by induction there exists an \ELHr concept $D$ and a monomial $\nonomial_2$ such that $\Omc\models (D\sqsubseteq C',v)$, $(f, \nonomial_2)\in D^{\Imc_{i}}$, and $\representative{v\times\nonomial_2}=\monomial_2$. 
It follows that $(d, \representative{\monomial_1\times\nonomial_2})\in {\exists S.D}^{\Imc_{i}}$. 
Moreover, since $\Omc\models (D\sqsubseteq C',v)$ and $\Omc\models (\exists S. D\sqsubseteq \exists S.D,1)$, by Lemma \ref{lem:conj-right}, $\Omc\models (\exists S. D\sqsubseteq \exists S.C',v)$. 
Finally, $\representative{v\times\monomial_1\times\nonomial_2}=\representative{\monomial_1\times\monomial_2}=\monomial$.
\end{itemize}
This concludes the proof of the lemma.
\end{proof}

Let $\Omc$ be an annotated \ELHr ontology in normal form, 
 $C\sqsubseteq D$ be an \ELHr GCI, 
and $\monomial_0\in\NM$. 
Let $$\Tmc_D=\{(D'\sqsubseteq E,1)\}$$ where $E$ is a concept name that does not occur in~$\Omc$ and $D'=D$ if $D\in\NC$, 
$D'=\exists R.\top$ if $D=\exists R$. 
Let $$\Amc_C=f(C,a_0, 0)$$ where $a_0$ is an individual name that does not occur in $\Omc$ and $f$ is the function inductively defined as 
follows, where all constants introduced are fresh and provenance variables of the form $v^r_{A(a)}$ or $v^r_{R(a,b)}$ do not occur in $\Omc$: 
\begin{itemize}
\item $f(\top,a,i)=\emptyset$,
\item $f(A,a,i)=\{(A(a),v^i_{A(a)})\}$ if $A\in\NC$, 
\item $f(\exists R.B,a,i)=\{(R(a,b),v^i_{R(a,b)})\}\cup f(B,b,i+1)$, 
\item $f(B\sqcap B',a,i)=f(B,a,i)\cup f(B',a, i+1)$. 
\end{itemize}
Both $\Tmc_D$ and $\Amc_C$ can be constructed in polynomial time.

\begin{proposition}
$\Omc\models (C\sqsubseteq D,\monomial_0)$ iff $\Omc\cup\Tmc_D\cup\Amc_C\models (E(a_0), \monomial_0\times \Pi_{(\alpha, v^r_\alpha)\in \Amc_C} v^r_\alpha)$. 
\end{proposition}
\begin{proof}
$(\Rightarrow)$ Assume that $\Omc\models (C\sqsubseteq D,\monomial_0)$ and let $\Imc$ be a model of $\Omc\cup\Tmc_D\cup\Amc_C$. Since $\Imc\models \Omc$, then $\Imc\models (C\sqsubseteq D,\monomial_0)$.
 Thus, since $\Imc\models (D'\sqsubseteq E,1)$, then $\Imc\models (C\sqsubseteq E,\monomial_0)$: 
 it is clear in the case $D'=D$ and otherwise it follows from the fact that $D^\Imc=(\exists R)^\Imc=\{(d,\nonomial^\Imc)\mid (d,e,\nonomial^\Imc)\in R^\Imc\}=(\exists R.\top)^\Imc = D'^\Imc$ by definition of the interpretation of $\top$. 
 Moreover, $(a_0^\Imc,(\Pi_{(\alpha, v^r_\alpha)\in \Amc_C} v^r_\alpha)^\Imc)\in C^\Imc$ by construction of $\Amc_C$, so $\Imc\models (E(a_0),\monomial_0\times \Pi_{(\alpha, v^r_\alpha)\in \Amc_C} v^r_\alpha)$. Hence $\Omc\cup\Tmc_D\cup\Amc_C\models (E(a_0),\monomial_0\times \Pi_{(\alpha, v^r_\alpha)\in \Amc_C} v^r_\alpha)$.\\
  
\noindent$(\Leftarrow)$ We show the other direction by contrapositive: we assume that $\Omc\not\models (C\sqsubseteq D,\monomial_0)$ and show that $\Omc\cup\Tmc_D\cup\Amc_C\not\models (E(a_0),\monomial_0\times \Pi_{(\alpha, v^r_\alpha)\in \Amc_C} v^r_\alpha)$. 
We next define a sequence $\Imc_0,\Imc_1,\dots$ of interpretations and the annotated interpretation
 $\Imc$ with $\Delta_{\sf m}^\Imc=\NMrep$, $\Delta^\Imc=\bigcup_{i\geq 0} \Delta^{\Imc_i}$ and $\cdot^\Imc$ is the union of the $\cdot^{\Imc_i}$. 
We start with $\Delta^{\Imc_0}=\individuals{\Amc_C}$, $a^{\Imc_0}=a$ for every $a\in\individuals{\Amc_C}$, and $A^{\Imc_0}=\{(a,v^r_\alpha)\mid (A(a),v^r_\alpha)\in\Amc_C\}$ for every $A\in\NC$ and $R^{\Imc_0}=\{(a,b,v^r_\alpha)\mid (R(a,b),v^r_\alpha)\in\Amc_C\}$ for every $R\in\NR$. Then for every $i\geq 0$,  $\Imc_{i+1}$ results from applying one of the rules $R1,R2,R3$ to $\Imc_i$. 

By construction, $\Imc$ is a model of $\Amc_C$ and of the GCIs, RIs and RRs of $\Omc$. 
We show by induction on $i$ that for all $\monomial\in\NM$ such that $\monomial$ does not contain any variable of the form $v^r_\alpha$  and  general \ELHr concept $G\neq \top$, 
if $(a_0, \representative{\monomial\times  \Pi_{(\alpha, v^r_\alpha)\in \Amc_C} v^r_\alpha})\in G^{\Imc_i}$ then $\Omc\models (C\sqsubseteq G, \monomial)$.

For the base case $i=0$, if $(a_0, \representative{\monomial\times \Pi_{(\alpha, v^r_\alpha)\in \Amc_C} v^r_\alpha})\in G^{\Imc_0}$, by construction of $\Imc_0$, $\representative{\monomial\times \Pi_{(\alpha, v^r_\alpha)\in \Amc_C} v^r_\alpha}=\representative{ \Pi_{(\alpha, v^r_\alpha)\in \Amc_C} v^r_\alpha}$. 
We use the following claim to show that the property holds in this case. 
\begin{claim}\label{claim:incltoassertions}
For every $a\in\NI$, for every \ELHr concept $C$, and for every \ELHr concept $G\neq \top$, if there is $k$ such that $(a, \representative{\Pi_{(\alpha, v^r_\alpha)\in f(C,a,k)} v^r_\alpha})\in G^{\Imc_0}$, then $G$ and $C$ are equal modulo some repetitions in $G$, i.e. if we replace some subconcepts of the form $G'\sqcap G'$, $G'\sqcap \top$ or $\top\sqcap G'$ by $G'$ in $G$ or some subconcepts of the form $\exists R.G'\sqcap \exists R.\top$ by $\exists R.G'$, we obtain exactly~$C$. 
\end{claim}
\noindent\emph{Proof of the claim.} The proof of the claim is by structural induction. If $C=A\in\NC$, $C=\top$ or $C=\exists R$, necessarily $k=0$ and $f(C,a,0)$ is a singleton or the emptyset. It can be checked that in the three cases, $G$ follows the grammar rule $G:=C \mid G\sqcap \top\mid\top\sqcap G$, so that $G$ is equal to $C$ modulo repetitions in $G$. 
If $C=C_1\sqcap C_2$, $f(C,a,k)=f(C_1,a,k)\cup f(C_2,a,k+1)$. 
Since $(a, \representative{ \Pi_{(\alpha, v^r_\alpha)\in f(C_1,a,k)} v^r_\alpha \times \Pi_{(\alpha, v^r_\alpha)\in f(C_2,a,k+1)} v^r_\alpha})\in G^{\Imc_0}$, it follows that $G$ follows the grammar rule $G:=G_1\sqcap G_2 \mid G_2\sqcap G_1 \mid G\sqcap G_1\mid G\sqcap G2\mid G\sqcap \top\mid \top\sqcap G$ 
with $(a, \representative{\Pi_{(\alpha, v^r_\alpha)\in f(C_1,a,k)} v^r_\alpha})\in G_1^{\Imc_0}$ and $(a, \representative{\Pi_{(\alpha, v^r_\alpha)\in f(C_2,a,k+1)} v^r_\alpha})\in G_2^{\Imc_0}$. 
By induction, $G_1$ is equal to $C_1$ and $G_2$ is equal to $C_2$ modulo repetitions in $G_1$ and $G_2$. 
Hence, $C$ and $G$ are equal modulo repetitions in $G$. 
Finally, if $C=\exists R.C'$, $f(C,a,k)=\{(R(a,b), v^k_{R(a,b)})\}\cup f(C',b,k+1)$. 
Since $(a, \representative{v^k_{R(a,b)}\times \Pi_{(\alpha, v^r_\alpha)\in f(C',b,k+1)} v^r_\alpha})\in G^{\Imc_0}$, it then follows that $G$ follows the grammar rules 
$G:= \exists R.G''\mid G\sqcap G\mid G\sqcap \top\mid \top \sqcap G\mid G\sqcap \exists R. T\mid  \exists R. T \sqcap G$, $G'':= G'\mid G''\sqcap G'\mid G''\sqcap \top\mid \top\sqcap G''$ and $T:=\top\mid\top\sqcap T$
with $(b, \representative{\Pi_{(\alpha, v^r_\alpha)\in f(C',b,k+1)} v^r_\alpha})\in G'^{\Imc_0}$. 
By induction, 
$G'$ is equal to $C'$ modulo repetitions in $G'$. 
Hence, $G$ is equal to $C$ modulo repetitions in $G$. 
This finishes the proof of the claim.\\

By Claim \ref{claim:incltoassertions}, since $(a_0, \representative{\Pi_{(\alpha, v^r_\alpha)\in \Amc_C} v^r_\alpha})\in G^{\Imc_0}$, it follows that 
$G$ and $C$ are equal modulo some repetitions in $G$. 
It follows  that $\Omc\models (C\sqsubseteq G, 1)$ (this is clear in the case $G=C$ and can be shown by structural induction using Lemma \ref{lem:conj-right} in the case where $G$ has repetitions). 
Moreover, since $\monomial$ does not contain any $v^r_\alpha$, $\monomial=1$. 
This shows that the property holds for $i=0$.

For the inductive step, assume that for all $\monomial\in\NM$ that does not contain any variable of the form $v^r_\alpha$ and general \ELHr concept $G\neq \top$, if $(a_0, \representative{\monomial\times \Pi_{(\alpha, v^r_\alpha)\in \Amc_C} v^r_\alpha})\in G^{\Imc_i}$ then $\Omc\models (C\sqsubseteq G, \monomial)$. Let $\monomial\in\NM$ without $v^r_\alpha$ and $G\neq \top$  and assume that $(a_0, \representative{\monomial\times \Pi_{(\alpha, v^r_\alpha)\in \Amc_C} v^r_\alpha})\in G^{\Imc_{i+1}}$. 
If $(a_0, \representative{\monomial\times \Pi_{(\alpha, v^r_\alpha)\in \Amc_C} v^r_\alpha})\in G^{\Imc_{i}}$, then $\Omc\models (C\sqsubseteq G, \monomial)$ by induction hypothesis. 
Otherwise, by Lemma \ref{lem:complinter}, since $\Imc_{i+1}$ is obtained by applying a rule in $\{R1,R2,R3\}$ to $\Imc_i$, and $(a_0, \representative{\monomial\times \Pi_{(\alpha, v^r_\alpha)\in \Amc_C} v^r_\alpha})\in G^{\Imc_{i+1}}$ while $(a_0, \representative{\monomial\times \Pi_{(\alpha, v^r_\alpha)\in \Amc_C} v^r_\alpha})\notin G^{\Imc_{i+1}}$, there exists an \ELHr concept $F$, a provenance variable $v\in\NV$ and a monomial $\nonomial$ such that $\Omc\models (F\sqsubseteq G,v)$, $(a_0, \nonomial)\in F^{\Imc_{i}}$, and $\representative{v\times\nonomial}=\representative{\monomial\times \Pi_{(\alpha, v^r_\alpha)\in \Amc_C} v^r_\alpha}$. 
Since the provenance variables of form $v^r_\alpha$ do not occur in $\Omc$, $v$ is not of this form. There thus exists $\nonomial'\in\NM$ that does not contain any $v^r_\alpha$ and is such that $\representative{\nonomial}=\representative{\nonomial'\times\Pi_{(\alpha, v^r_\alpha)\in \Amc_C} v^r_\alpha}$ and $\representative{v\times\nonomial'}=\monomial$. 
Since $(a_0, \nonomial)\in F^{\Imc_{i}}$, it follows by induction that $\Omc\models (C\sqsubseteq F, \nonomial')$. 
Hence, since $\Omc\models (F \sqsubseteq G,v)$, we conclude that $\Omc\models (C\sqsubseteq G,v\times\nonomial')$, i.e. that $\Omc\models (C\sqsubseteq G,\monomial)$.

It follows that for every \ELHr concept $G\neq \top$  and $\monomial\in\NM$ that does not contain any $v^r_\alpha$, if $(a_0, \representative{\monomial\times \Pi_{(\alpha, v^r_\alpha)\in \Amc_C} v^r_\alpha})\in G^{\Imc}$ then $\Omc\models (C\sqsubseteq G, \monomial)$.
Since $D$ is either a concept name or is of the form $\exists R$ and $\Omc\not\models (C\sqsubseteq D,\monomial_0)$, it follows that $ (a_0, \representative{\monomial_0\times \Pi_{(\alpha, v^r_\alpha)\in \Amc_C} v^r_\alpha})\notin D^\Imc$.\\
 
 Let $\Jmc_0$ be a model of $\Omc$ such that $\Delta^{\Jmc_0}\cap\Delta^\Imc=\emptyset$, $\Delta^{\Jmc_0}_{\sf m}=\Delta^\Imc_{\sf m}=\NMrep$ and $\monomial^{\Jmc_0}=\monomial^\Imc=\representative{\monomial}$ for every $\monomial\in\NM$. Let $\Jmc$ be the interpretation defined as follows:
\begin{itemize}
\item $\Delta^\Jmc=\Delta^{\Jmc_0}\cup\Delta^\Imc$, $\Delta^{\Jmc}_{\sf m}=\Delta^\Imc_{\sf m}=\NMrep$,
\item $a^{\Jmc}=a^\Imc$ for every $a\in\individuals{\Amc_C}$,
\item $a^{\Jmc}=a^{\Jmc_0}$ for every $a\in\individuals{\Omc}$,
\item $R^\Jmc=R^{\Imc}\cup R^{\Jmc_0}$ for every $R\in\NR$,
\item $A^\Jmc=A^{\Imc}\cup A^{\Jmc_0}$ for every $A\in\NC$, $A\neq E$, and
\item $E^\Jmc=D^\Jmc$.
\end{itemize}
Since $\individuals{\Amc_C}\cap\individuals{\Omc}=\emptyset$, $\Jmc$ is well defined. 
Moreover, since $\Delta^{\Jmc_0}\cap\Delta^\Imc=\emptyset$ and $(a_0, \representative{\monomial_0\times \Pi_{(\alpha, v^r_\alpha)\in \Amc_C} v^r_\alpha})\notin D^\Imc$, then $(a_0, \representative{\monomial_0\times \Pi_{(\alpha, v^r_\alpha)\in \Amc_C} v^r_\alpha})\notin D^\Jmc$, so 
$(a_0, \representative{\monomial_0\times \Pi_{(\alpha, v^r_\alpha)\in \Amc_C} v^r_\alpha})\notin E^\Jmc$. 
Finally, $\Jmc$ is a model of $\Omc\cup\Tmc_D\cup\Amc_C$. Indeed, $\Jmc$ satisfies all assertions of $\Omc$ and $\Amc_C$ because $\Jmc_0$ is a model of $\Omc$ and $\Imc$ a model of $\Amc_C$. It also satisfies $\Tmc_D$ by construction of $E^\Jmc$, and all GCIs, RIs and RRs of $\Omc$ because both $\Jmc_0$ and $\Imc$ satisfy them and $\Delta^{\Jmc_0}\cap\Delta^\Imc=\emptyset$. 
Hence, $\Omc\cup\Tmc_D\cup\Amc_C\not\models E(a_0, \monomial_0\times \Pi_{(\alpha, v^r_\alpha)\in \Amc_C} v^r_\alpha)$.
\end{proof}

Regarding role inclusions and range restriction, the following propositions can be proven. 
\begin{proposition}
$\Omc\models (S\sqsubseteq R,\monomial)$ iff $\Omc\cup\{(S(a_0,b_0), 1)\}\models(R(a_0,b_0),\monomial)$ where $a_0,b_0$ are fresh individual names. 
\end{proposition}

\begin{proposition}
$\Omc\models ({\sf ran}(R)\sqsubseteq A,\monomial)$ iff $\Omc\cup\{(R(a_0,b_0), 1)\}\models(A(b_0),\monomial)$ where $a_0,b_0$ are fresh individual names. 
\end{proposition}

Let $\Omc$ be an annotated \ELHr ontology in normal form, and let $(B(a_0),\monomial_0)$ be an annotated concept assertion. 
For all $a,b\in\individuals{\Omc}$ and $R\in\NR$, assume that the concept or role names $C_a, C_{{\sf ran}(R)}, R_{ab}$ do not occur in $\Omc$ and
let \begin{align*}
\Tmc_{C_a}=&\{(C_a\sqsubseteq A, v)\mid (A(a),v)\in\Omc\}\cup\\
&\{(C_a\sqsubseteq \exists R_ {ab},1), (R_{ab}\sqsubseteq R, v), ({\sf ran}(R_{ab})\sqsubseteq C_b, 1) \\&\phantom{\{}
\mid (R(a,b),v)\in\Omc\}\cup\\
&\{(C_a\sqsubseteq C_{{\sf ran}(R)},v)\mid (R(b,a),v)\in\Omc\}\cup\\&
\{(C_{{\sf ran}(R)}\sqsubseteq C_{{\sf ran}(S)},v)\mid (R\sqsubseteq S,v)\in\Omc\}\cup\\&
\{(C_{{\sf ran}(R)}\sqsubseteq A,v)\mid ({\sf ran}(R)\sqsubseteq A,v)\in\Omc\}\\
\Tmc=&\bigcup_{a\in\individuals{\Omc}}\Tmc_{C_a}\cup(\Omc\setminus\{(\alpha,v)\mid \alpha\text{ is an assertion} \}).
\end{align*}
$\Tmc$ can be computed in polynomial time.

We will use the following lemma to show that $\Omc\models(B(a_0),\monomial_0)$ iff $\Tmc\models (C_{a_0}\sqsubseteq B,\monomial_0)$ or $\Tmc\models (\top\sqsubseteq B,\monomial_0)$.

\begin{lemma}\label{lem:reductionGCI-assertion3}
There exist a model $\Imc$ of $\Tmc$ and $\epsilon\in\Delta^\Imc$ such that 
\begin{enumerate}
\item for every $a\in\NI$, $a^\Imc=\epsilon$ iff there is no directed path of roles $(R_0(a_0,a_1),v_0)$\dots $(R_n(a_n,a),v_n)$ from $a_0$ to $a$ in~$\Omc$, 
\item for every $a\in\individuals{\Omc}$, $C_a^\Imc=\{(a^\Imc,1^\Imc)\}$ if $a^\Imc\neq \epsilon$, and $C_a^\Imc=\emptyset$ otherwise,
\item for $(R(a,b),v)\in\Omc$, $R_{ab}^\Imc=\{(a^\Imc,b^\Imc,1^\Imc)\}$ if $a^\Imc\neq \epsilon$ and $b^\Imc\neq \epsilon$, and $R_{ab}^\Imc=\emptyset$ otherwise, and
\item for every $a\in\individuals{\Omc}$ and $A\in\NC$, if $(a^\Imc,\monomial^\Imc)\in A^\Imc$ then $\Tmc\models (C_a\sqsubseteq A,\monomial)$ or $\Tmc\models (\top\sqsubseteq A,\monomial)$.
\end{enumerate}
\end{lemma}
\begin{proof}
Let $\epsilon\notin\individuals{\Omc}$. We inductively build an interpretation $\Imc=\bigcup_{i\geq 0}\Imc_i$ with $\Delta^\Imc_{\sf m}=\NMrep$ such that 
\begin{itemize}[leftmargin=5.5mm]
\item[(a)] for every $i$, for every $a\in\individuals{\Omc}$, $C_{a}^{\Imc_i}=\{(a^{\Imc_i}, 1)\}$ if $a^{\Imc_i}\neq \epsilon$, and $C_{a}^{\Imc_i}=\emptyset$ otherwise,
\item[(b)] for every $i$, for every $(R(a,b),v)\in\Omc$, $R_{ab}^{\Imc_i}=\{(a^{\Imc_i},b^{\Imc_i}, 1)\}$ if $a^{\Imc_i}\neq \epsilon$ and $b^{\Imc_i}\neq\epsilon$, and $R_{ab}^{\Imc_i}=\emptyset$ otherwise, 
\item[(c)] for every $i$, for every $a\in\individuals{\Omc}$ such that $a^{\Imc_i}\neq \epsilon$, $A\in \NC$ and $S\in\NR$, 
\begin{itemize}
\item if $(a^{\Imc_i},\monomial)\in A^{\Imc_i}$, then $\Tmc\models (C_a\sqsubseteq A,\monomial)$ or $\Tmc\models (\top\sqsubseteq A,\monomial)$, 
\item if $(a^{\Imc_i},\monomial)\in (\exists S. A)^{\Imc_i}$, then $\Tmc\models (C_a \sqsubseteq \exists S',\nonomial_1)$ or $\Tmc\models (\top \sqsubseteq \exists S',\nonomial_1)$, and $\Tmc\models (S' \sqsubseteq S, \nonomial_2)$,  
and (i) if $A=\top$, then  $\monomial=\representative{\nonomial_1\times\nonomial_2}$
 (ii)  if $A\neq\top$ then 
$\Tmc\models ({\sf ran}(S') \sqsubseteq A, \nonomial_3)$ with $\monomial=\representative{\nonomial_1\times\nonomial_2\times\nonomial_3}$, and 
\item if $(a^{\Imc_i},\monomial)\in {\sf ran}(S)^{\Imc_i}$ and $S$ is not of the form $R_{ab}$, then $\Tmc\models (C_a\sqsubseteq C_{{\sf ran}(S)}, \monomial)$.
\end{itemize}

\end{itemize}
We start with $\Delta^{\Imc_0}=\individuals{\Omc}\cup\{\epsilon\}$, $a_0^{\Imc_0}=a_0$, $a^{\Imc_0}=\epsilon$ for every $a\neq a_0$, $C_{a_0}^{\Imc_0}=\{(a_0, 1)\}$ and $A^{\Imc_0}=\emptyset$, $R^{\Imc_0}=\emptyset$ for all $A\in\NC\setminus\{C_{a_0}\}$ and $R\in \NR$. Then we apply the following rules. 
\begin{enumerate}
\item If $(R\sqsubseteq S, v)\in\Tmc$ and $(d,f,\monomial)\in R^{\Imc_i}$, then $S^{\Imc_{i+1}}=S^{\Imc_i}\cup\{(d,f, \representative{v\times\monomial})\}$. \\
We check that if $\Imc_i$ fulfills our requirements, then it is also the case of $\Imc_{i+1}$.  
For point (b), this follows from the fact that $S$ cannot be of the form $R_{ab}$ by construction of $\Tmc$. 
Regarding point (c), if $d=a^{\Imc_i}$ for some $a\in\individuals{\Omc}$ and  $(d,\nonomial)\in (\exists S.A)^{\Imc_{i+1}}$ while $(d,\nonomial)\notin (\exists S.A)^{\Imc_{i}}$, there must be some $(f,\onomial)\in A^{\Imc_{i}}$ such that $\nonomial=\representative{v\times\monomial\times\onomial}$. 
It follows that $(d,\representative{\monomial\times \onomial})\in (\exists R.A)^{\Imc_{i}}$ so since $\Imc_i$ fulfills the requirements, $\Tmc\models (C_a \sqsubseteq \exists S',\nonomial_1)$ or $\Tmc\models (\top \sqsubseteq \exists S',\nonomial_1)$, and $\Tmc\models (S' \sqsubseteq R, \nonomial_2)$ and (i) either $A=\top$ and $\representative{\monomial\times \onomial}=\representative{\nonomial_1\times\nonomial_2}$ or (ii) $A\neq\top$ and $\Tmc\models ({\sf ran}(S') \sqsubseteq A, \nonomial_3)$ with $\representative{\monomial\times \onomial}=\representative{\nonomial_1\times\nonomial_2\times\nonomial_3}$. 
Since $(R\sqsubseteq S, v)\in\Tmc$ , it thus follows that $\Tmc\models (S' \sqsubseteq S, v\times\nonomial_2)$, so that the second point of (c) holds (since if $A=\top$, $\nonomial=\representative{v\times\monomial\times\onomial}=\representative{\nonomial_1\times v\times \nonomial_2}$, and if $A\neq\top$, $\nonomial=\representative{v\times\monomial\times\onomial}=\representative{\nonomial_1\times v\times \nonomial_2\times\nonomial_3}$). 
In a similar way, if $f=a^{\Imc_i}$ for some $a\in\individuals{\Omc}$, since $(f,\monomial)\in {\sf ran}(R)^{\Imc_{i}}$ and $\Imc_i$ fulfills the requirements, $\Tmc\models (C_a\sqsubseteq C_{{\sf ran}(S)}, \monomial)$.
By construction of $\Tmc$, since $(R\sqsubseteq S, v)\in\Tmc$, $(C_{{\sf ran}(R)}\sqsubseteq C_{{\sf ran}(S)}, v)\in\Tmc$. 
Thus $\Tmc\models (C_a\sqsubseteq C_{{\sf ran}(S)}, \representative{v\times \monomial})$, so that the third point of (c) holds. 

\item If $ (A\sqsubseteq B, v)\in\Tmc$ and $(d,\monomial)\in A^{\Imc_i}$, then $B^{\Imc_{i+1}}=B^{\Imc_i}\cup\{(d, \representative{v\times\monomial})\}$. \\
It is easy to check that if $\Imc_i$ fulfills our requirements, then it is also the case of $\Imc_{i+1}$. 
In particular, note that by construction of $\Tmc$, $B$ is not of the form~$C_a$.

\item If $(A_1\sqcap A_2\sqsubseteq B, v)\in\Tmc$ and $(d,\monomial_1)\in A_1^{\Imc_i}$, $(d,\monomial_2)\in A_1^{\Imc_i}$ then $B^{\Imc_{i+1}}=B^{\Imc_i}\cup\{(d, \representative{v\times\monomial_1\times\monomial_2})\}$. \\
It is easy to check that if $\Imc_i$ fulfills our requirements, then it is also the case of $\Imc_{i+1}$. 
In particular, for (a), note that by construction of $\Tmc$, $B$ is not of the form $C_a$. 
We detail the proof that the first point of (c) holds (the second point is similar). 
If $d=a^{\Imc_i}$ for some $a\in\individuals{\Omc}$, since  $(d,\monomial_1)\in A_1^{\Imc_i}$, $(d,\monomial_2)\in A_1^{\Imc_i}$ and $\Imc_i$ fulfills the requirements, then (1) $\Tmc\models (C_a\sqsubseteq A_1,\monomial_1)$ or (2) $\Tmc\models (\top\sqsubseteq A_1,\monomial_1)$, and (i) $\Tmc\models (C_a\sqsubseteq A_1,\monomial_2)$ or (ii) $\Tmc\models (\top\sqsubseteq A_1,\monomial_2)$. 
In case (1)-(i), since $(A_1\sqcap A_2\sqsubseteq B, v)\in\Tmc$, it is easy to check that $\Tmc\models (C_a\sqsubseteq B,\representative{v\times\monomial_1\times\monomial_2})$. 
In case (2)-(ii), in the same way, $\Tmc\models (\top\sqsubseteq B,\representative{v\times\monomial_1\times\monomial_2})$. 
Finally in case (1)-(ii) (and (2)-(i) is similar), let $\Jmc$ be a model of $\Tmc$ and $(e,\nonomial^\Jmc)\in C_a^\Jmc$. 
Since $\Tmc\models (C_a\sqsubseteq A_1,\monomial_1)$,  
$(e,(\nonomial\times \monomial_1)^\Jmc)\in A_1^\Jmc$, and since $\Tmc\models (\top\sqsubseteq A_1,\monomial_2)$, $(e, \monomial_2^\Jmc)\in A_2^\Jmc$. 
Hence $(e,(\nonomial\times \monomial_1\times\monomial_2)^\Jmc)\in (A_1\sqcap A_2)^\Jmc$, so $(e,(\nonomial\times \monomial_1\times\monomial_2\times v)^\Jmc)\in B^\Jmc$. It follows that $\Tmc\models (C_a\sqsubseteq B,\monomial_1\times\monomial_2\times v )$.

\item If $(\exists S.A\sqsubseteq B, v)\in\Tmc$ and $(d,f,\monomial_1)\in S^{\Imc_i}$, $(f,\monomial_2)\in A^{\Imc_i}$, then $B^{\Imc_{i+1}}=B^{\Imc_i}\cup\{(d, \representative{v\times\monomial_1\times\monomial_2})\}$. \\
We check that if $\Imc_i$ fulfills our requirements, then it is also the case of $\Imc_{i+1}$.  
Note that by construction of $\Tmc$, $B$ is not of the form $C_a$, so that (a) holds. 
We detail the proof that the first point of (c) holds.  
If $d=a^{\Imc_i}$ for some $a\in\individuals{\Omc}$, since $(d,\representative{\monomial_1\times\monomial_2})\in (\exists S.A)^{\Imc_i}$ and $\Imc_i$ fulfills the requirements, then either (i) $\Tmc\models (C_a \sqsubseteq \exists S',\nonomial_1)$ or (ii) $\Tmc\models (\top\sqsubseteq \exists S',\nonomial_1)$, and $\Tmc\models (S' \sqsubseteq S, \nonomial_2)$ and $\Tmc\models ({\sf ran}(S') \sqsubseteq A, \nonomial_3)$ with $\representative{\monomial_1\times\monomial_2}=\representative{\nonomial_1\times\nonomial_2\times\nonomial_3}$. 
It is easy to check that it follows that in  case (i) $\Tmc\models (C_a\sqsubseteq B, \nonomial_1\times\nonomial_2\times\nonomial_3\times v)$, i.e.  $\Tmc\models (C_a\sqsubseteq B, \representative{v\times\monomial_1\times\monomial_2})$ and in case (ii) $\Tmc\models (\top\sqsubseteq B, \nonomial_1\times\nonomial_2\times\nonomial_3\times v)$, i.e.  $\Tmc\models (\top\sqsubseteq B, \representative{v\times\monomial_1\times\monomial_2})$.

\item If $({\sf ran}(S)\sqsubseteq B, v)\in\Tmc$ and $(f,d,\monomial)\in S^{\Imc_i}$, then $B^{\Imc_{i+1}}=B^{\Imc_i}\cup\{(d, \representative{v\times\monomial})\}$. \\
We check that if $\Imc_i$ fulfills our requirements, then it is also the case of $\Imc_{i+1}$.  
Note that by construction of $\Tmc$, $B$ is not of the form $C_a$, so that (a) holds. 
We detail the proof that the first point of (c) holds. 
If $d=a^{\Imc_i}$ for some $a\in\individuals{\Omc}$, since $\Imc_i$ fulfills the requirements, it follows that $\Tmc\models (C_a\sqsubseteq C_{{\sf ran}(S)}, \monomial)$. 
Moreover, by construction of $\Tmc$, $({\sf ran}(S)\sqsubseteq B, v)\in\Tmc$ implies that $(C_{{\sf ran}(S)}\sqsubseteq B, v)\in\Tmc$. 
It is thus easy to check that $\Tmc\models (C_a\sqsubseteq B,  \representative{v\times\monomial})$. 

\item If $(A\sqsubseteq \exists S, v)\in\Tmc$ and  $(d,\monomial)\in A^{\Imc_i}$, then:
\begin{itemize}
\item  Case 1: $A$ is not of the form $C_a$. 
Set $\Delta^{\Imc_{i+1}}=\Delta^{\Imc_{i}}\cup\{f\}$ where $f\notin\Delta^{\Imc_{i}}$ and $S^{\Imc_{i+1}}=S^{\Imc_{i}}\cup\{(d,f, \representative{v\times \monomial})\}$. It is easy to check that if $\Imc_i$ fulfills the requirements, then so does $\Imc_{i+1}$. Indeed, $S$ cannot be of the form $R_{ab}$ by construction and there is no $a\in\individuals{\Omc}$ such that $a^{\Imc_i}=f$ so the only point to check is the second point of (c). 
If $d=a^{\Imc_i}$ for some $a\in\individuals{\Omc}$,  and  $(d,\nonomial)\in (\exists S'.C)^{\Imc_{i+1}}$ while $(d,\nonomial)\notin (\exists S'.C)^{\Imc_{i}}$, the only possibility is that $S'=S$, $C=\top$ and $\nonomial=\representative{v\times \monomial}$. 
Since $(a^{\Imc_i},\monomial)\in A^{\Imc_i}$ and $\Imc_i$ fulfills the requirements, $\Tmc\models (C_a \sqsubseteq A,\monomial)$, so since $(A \sqsubseteq \exists S,v)\in\Tmc$, it is easy to check that $\Tmc\models (C_a \sqsubseteq \exists S, v\times\monomial)$. Moreover $\Tmc\models (S \sqsubseteq S, 1)$ trivially. Hence (c) holds. 

\item Case 2: $A=C_a$. By construction of $\Tmc$, there is $(R(a,b), u)\in\Omc$ such that $\exists S=\exists R_{ab}$, $v=1$ and $\Tmc$ also contains $({\sf ran}(R_{ab})\sqsubseteq C_b, 1)$.
 Moreover, since $C_a^{\Imc_i}\neq\emptyset$, then $C_a^{\Imc_i}=\{(d,1)\}=\{(a^{\Imc_i},1)\}=\{(a,1)\}$.  
If $b^{\Imc_i}\neq\epsilon$, it was already the case that $R_{ab}^{\Imc_{i}}=\{(a^{\Imc_i},b^{\Imc_i}, 1)\}$ so the rule applies to the case $b^{\Imc_i}=\epsilon$. Set $b^{\Imc_{i+1}}=b$ and $R_{ab}^{\Imc_{i+1}}=\{(a,b, 1)\}$, $C_b^{\Imc_{i+1}}=\{(b, 1)\}$ and ${R'}_{ab}^{\Imc_{i+1}}=\{(a,b, 1)\}$ for every $R'_{ab}$ occurring in $\Tmc$. 
  
We check that if $\Imc_i$ fulfills our requirements, then it is also the case of $\Imc_{i+1}$.  
Requirements (a) and (b) hold by construction so we focus on (c). 
For the first point, $b$ only occurs in the concept interpretation $C_b^{\Imc_{i+1}}$ with provenance 1, and $\Tmc\models (C_b\sqsubseteq C_b, 1)$ trivially. 
For the second point, if $(a,\nonomial)\in (\exists S.C)^{\Imc_{i+1}}$ while $(a,\nonomial)\notin (\exists S.C)^{\Imc_{i}}$, the only possibility is that $S={R'}_{ab}$, $C=\top$ and $\nonomial=1$. 
By construction of $\Tmc$, $\Tmc\models (C_a \sqsubseteq \exists {R'}_{ab},1)$, and trivially $\Tmc\models ({R'}_{ab} \sqsubseteq {R'}_{ab}, 1)$. This shows the second point of (c). 
 Finally,  if $(b,\nonomial)\in {\sf ran}(S)^{\Imc_{i+1}}$, $S$ is of the form ${R'}_{ab}$ so the third point trivially holds.
\end{itemize}
\end{enumerate}

Since $\Imc$ has been obtained by applying rules to satisfy all axioms of $\Tmc$ (as $\Tmc$ is in normal form), $\Imc$ is a model of $\Tmc$. 
We show that for every directed path of role assertions $(R_1(a_0,a_1),v_1)$\dots $(R_n(a_{n-1},a_n),v_n)$ in $\Omc$, $a_i^\Imc\neq\epsilon$ for all $a_i$. 
It is clear for $a_0$. Assume that the property is true for every path of lenght $n-1$ and consider $(R_1(a_0,a_1),v_1)\dots (R_n(a_{n-1},a_n),v_n)$ in $\Omc$.  
We have $C_{a_{n-1}}^\Imc=\{(a_{n-1}^\Imc, 1^\Imc)\}$. 
Since $(R_n(a_{n-1},a_n),v_n)\in\Omc$, then $(C_{a_{n-1}}\sqsubseteq \exists R_{a_{n-1}a_n},1)\in\Tmc$ so by the construction of $\Imc$ (rule 6), $a_{n}^\Imc=a_n\neq \epsilon$ and $R_{a_{n-1}a_n}^\Imc=\{(a_{n-1}^\Imc, a_{n}^\Imc, 1^\Imc)\}$. 
In the other direction, for every $a\in\NI$, if $a^\Imc\neq\epsilon$, it follows from the construction that $a^\Imc=a$ and either $a=a_0$ or $a^\Imc=a$ has been defined in some application of rule 6. We can show by induction on the number of applications of rule 6 before the one that defined $a^\Imc=a$ that there is a role path between $a_0$ and $a$.
\end{proof}

\begin{proposition}
$\Omc\models(B(a_0),\monomial_0)$ iff $\Tmc\models (C_{a_0}\sqsubseteq B,\monomial_0)$ or $\Tmc\models (\top\sqsubseteq B,\monomial_0)$.
\end{proposition}
\begin{proof}
Assume that $\Omc\not\models(B(a_0),\monomial_0)$. 
Assume for a contradiction that $\Tmc\models (\top\sqsubseteq B,\monomial_0)$. Then $\Omc\models (\top\sqsubseteq B,\monomial_0)$ by the form of the GCIs in $\Tmc\setminus\Omc$. 
It follows that $\Omc\models(B(a),\monomial_0)$ for every $a\in\individuals{\Omc}$, and in particular for $a_0$. Hence $\Tmc\not\models (\top\sqsubseteq B,\monomial_0)$. 
We next show that $\Tmc\not\models (C_{a_0}\sqsubseteq B,\monomial_0)$.

Let $\Imc$ be a model of $\Omc$ such that $\Imc\not\models (B(a_0),\monomial_0)$, i.e. $(a_0^\Imc, \monomial_0^\Imc)\notin B^\Imc$. 
Let $\Jmc$ be the interpretation that extends $\Imc$ with $C_{a}^\Jmc=\{(a^\Imc, 1^\Imc)\}$ for every $a\in\individuals{\Omc}$, $R_{ab}^\Jmc=\{(a^\Imc, b^\Imc, 1^\Imc)\}$ for all $a,b\in\individuals{\Omc}$ and $C_{{\sf ran}(R)}^\Jmc=\{(e,\monomial^\Imc)\mid (d,e,\monomial^\Imc)\in R^\Imc\}$ for every $R\in\NR$. 
Since $(a_0^\Imc, 1^\Imc)\in C_{a_0}^\Jmc$ and $(a_0^\Imc, \monomial_0^\Imc)\notin B^\Jmc$, then $\Jmc\not\models (C_{a_0}\sqsubseteq B,\monomial_0)$. 
We show that $\Jmc$ is a model of $\Tmc$, so that $\Tmc\not\models (C_{a_0}\sqsubseteq B,\monomial_0)$. 

It is clear that $\Jmc$ is a model of every GCI, RI or RR in $\Omc$ since interpretations of concepts and roles that occur in $\Omc$ are not modified.  
Regarding RIs of $\Tmc\setminus\Omc$, let $(R_{ab}\sqsubseteq R, v)\in\Tmc$. By construction of $\Jmc$, $R_{ab}^\Jmc=\{(a^\Imc, b^\Imc, 1^\Imc)\}$  and since $(R(a,b),v)\in\Omc$ and $\Imc$ is a model of $\Omc$, then $(a^\Imc, b^\Imc, v^\Imc)\in R^\Jmc$. Thus $\Jmc\models (R_{ab}\sqsubseteq R, v)$. 
For RRs of $\Tmc\setminus\Omc$, let $({\sf ran}(R_{ab})\sqsubseteq C_b, 1)\in\Tmc$. By construction $R_{ab}^\Jmc=\{(a^\Imc, b^\Imc, 1^\Imc)\}$ and $C_b^\Jmc=\{(b^\Imc, 1^\Imc)\}$ so $\Jmc\models ({\sf ran}(R_{ab})\sqsubseteq C_b, 1)$. 
We now consider the different kinds of GCIs in $\Tmc\setminus\Omc$. 
\begin{itemize}
\item Let $(C_a\sqsubseteq A,v)\in\Tmc\setminus\Omc$ with $A\in\NC$. 
Since $(A(a),v)\in\Omc$ and $\Imc\models \Omc$, then $(a^\Imc, v^\Imc)\in A^\Jmc$. 
Thus, since $C_{a}^\Jmc=\{(a^\Imc, 1^\Imc)\}$, it follows that $\Jmc\models (C_a\sqsubseteq A,v)$.
\item Let $(C_a\sqsubseteq \exists R_{ab},1)\in\Tmc\setminus\Omc$. 
By construction $R_{ab}^\Jmc=\{(a^\Imc, b^\Imc, 1^\Imc)\}$ so $(\exists R_{ab})^\Jmc=\{(a^\Imc, 1^\Imc)\}$ and $\Jmc\models (C_a\sqsubseteq \exists R_{ab},1)$. 
\item Let $(C_a\sqsubseteq C_{{\sf ran}(R)},v)\in\Tmc\setminus\Omc$. 
Since $\Imc\models \Omc$ and there is $(R(b,a), v)\in\Omc$, $(b^\Imc,a^\Imc,v^\Imc)\in R^\Jmc$ so  $(a^\Imc,v^\Imc)\in C_{{\sf ran}(R)}^\Jmc$. Since $C_{a}^\Jmc=\{(a^\Imc, 1^\Imc)\}$, it follows that $\Jmc\models (C_a\sqsubseteq C_{{\sf ran}(R)},v)$.
\item Let $(C_{{\sf ran}(R)}\sqsubseteq C_{{\sf ran}(S)},v)\in\Tmc\setminus\Omc$. 
Let $(e,\monomial^\Imc)\in C_{{\sf ran}(R)}^\Jmc$. There exists $(d,e,\monomial^\Imc)\in R^\Imc$.
Moreover, $(R\sqsubseteq S,v)\in\Omc$ so since $\Imc$ is a model of $\Omc$, $(d,e,(\monomial\times v)^\Imc)\in S^\Imc$. Thus $(e,(\monomial\times v)^\Imc)\in C_{{\sf ran}(S)}^\Imc$. It follows that $\Jmc\models (C_{{\sf ran}(R)}\sqsubseteq C_{{\sf ran}(S)},v)$.
\item Let $(C_{{\sf ran}(R)}\sqsubseteq A,v)\in\Tmc\setminus\Omc$ with $A\in\NC$. 
 Let $(e,\monomial^\Imc)\in C_{{\sf ran}(R)}^\Jmc$. There exists $(d,e,\monomial^\Imc)\in R^\Imc$.
Moreover, $({\sf ran}(R)\sqsubseteq A,v)\in\Omc$ so since $\Imc$ is a model of $\Omc$, $(e,(\monomial\times v)^\Imc)\in A^\Imc=A^\Jmc$. 
It follows that $\Jmc\models (C_{{\sf ran}(R)}\sqsubseteq A,v)$.
\end{itemize}
We conclude that $\Jmc\models\Tmc$, so $\Tmc\not\models (C_{a_0}\sqsubseteq B, \monomial_0)$.\\

In the other direction, assume that $\Tmc\not\models (C_{a_0}\sqsubseteq B,\monomial_0)$ and $\Tmc\not\models (\top\sqsubseteq B,\monomial_0)$, and let $\Imc$ be a model of $\Tmc$ that fulfills the conditions of Lemma \ref{lem:reductionGCI-assertion3}. 
Let $\Jmc_0$ be a model of $\Omc$ such that $\Delta^{\Jmc_0}\cap\Delta^\Imc=\emptyset$, $\Delta^{\Jmc_0}_{\sf m}=\Delta^\Imc_{\sf m}$ and $\monomial^{\Jmc_0}=\monomial^\Imc$ for every $\monomial\in\NM$. Let $\Jmc$ be the interpretation defined as follows:
\begin{itemize}
\item $\Delta^\Jmc=\Delta^{\Jmc_0}\cup\Delta^\Imc$, $\Delta^{\Jmc}_{\sf m}=\Delta^\Imc_{\sf m}$,
\item $a^{\Jmc}=a^\Imc$ for every $a$ such that there is a directed role path from $a_0$ to $a$ in $\Omc$ (including $a_0$),
\item $a^{\Jmc}=a^{\Jmc_0}$ for every $a$ such that there is no directed role path from $a_0$ to $a$ in $\Omc$,
\item $A^\Jmc=A^{\Imc}\cup A^{\Jmc_0}$ for every $A\in\NC$,
\item $R^\Jmc=R^{\Imc}\cup R^{\Jmc_0}\cup\{(b^{\Jmc_0},a^\Imc,\monomial^\Imc)\mid  \Omc\models (R(b,a),\monomial) $ and there is a directed role path from $a_0$ to $a$ but not to $b\}$ for every $R\in\NR$.
\end{itemize}
We show that $\Jmc$ is a model of $\Omc$. Since $(a_0^\Jmc,\monomial_0^\Jmc)=(a_0^\Imc,\monomial_0^\Imc)\notin B^\Imc$ by the properties of $\Imc$ (point 4 of Lemma \ref{lem:reductionGCI-assertion3}, since $\Tmc\not\models (C_{a_0}\sqsubseteq B,\monomial_0)$ and $\Tmc\not\models (\top\sqsubseteq B,\monomial_0)$) and $\Delta^{\Jmc_0}\cap\Delta^\Imc=\emptyset$, it follows that $(a_0^\Jmc,\monomial_0^\Jmc)\notin B^\Jmc$. This will thus show that $\Omc\not\models (B(a_0),\monomial_0)$. 

We start with $\Jmc$ being a model of the assertions of $\Omc$. 
Let $(A(a),v)\in\Omc$. If there is no directed path from $a_0$ to $a$ in $\Omc$, $(a^{\Jmc},v^\Jmc)=(a^{\Jmc_0}, v^{\Jmc_0})\in A^{\Jmc_0}\subseteq A^\Jmc$. 
If there is a directed path from $a_0$ to $a$, $(a^{\Jmc},1^\Jmc)=(a^\Imc, 1^{\Imc})\in C_a^\Imc$ so since $\Imc\models \Tmc$ and $(C_a\sqsubseteq A, v)\in\Tmc$ because $(A(a),v)\in\Omc$, then $(a^{\Jmc},v^\Jmc)=(a^\Imc, v^{\Imc})\in A^{\Imc}\subseteq A^\Jmc$. 
Regarding role assertions, let $(R(a,b),v)\in\Omc$. If there is no directed path from $a_0$ to $a$ nor to $b$, $(a^{\Jmc},b^{\Jmc}, v^\Jmc)=(a^{\Jmc_0},b^{\Jmc_0}, v^{\Jmc_0})\in R^{\Jmc_0}\subseteq R^\Jmc$. 
If there is a directed path from $a_0$ to $a$ (then also to $b$),
$(a^{\Jmc},b^{\Jmc}, v^\Jmc)=(a^\Imc,b^\Imc, v^\Imc)\in R^{\Imc}\subseteq R^\Jmc$ because $R_{ab}^\Imc=\{(a^\Imc,b^\Imc, 1^\Imc)\}$, $(R_{ab}\sqsubseteq R,v)\in\Tmc$ and $\Imc\models\Tmc$. 
Finally, if there is a directed path from $a_0$ to $b$ but not to $a$, $(a^{\Jmc},b^{\Jmc}, v^\Jmc)=(a^{\Jmc_0}, b^\Imc, v^\Imc)\in R^\Jmc$ by construction of $R^\Jmc$. 
We now turn to GCIs, RIs and RRs. 
\begin{itemize}
\item Let $(R\sqsubseteq S,v)\in \Omc$ and $(d, e,\monomial^\Imc)\in R^\Jmc$. The only non trivial case is $(d, e,\monomial^\Imc)=(b^{\Jmc_0}, a^\Imc,\monomial^\Imc)\in R^\Jmc\setminus(R^\Imc\cup R^{\Jmc_0})$. By construction, in this case $\Omc\models (R(b,a),\monomial)$, so it is easy to see that $\Omc\models (S(b,a),\monomial\times v)$. Thus by construction of $S^\Jmc$, $(b^{\Jmc_0}, a^\Imc,(\monomial\times v)^\Imc)\in S^\Jmc$. 
Hence $\Jmc\models (R\sqsubseteq S,v)$. 

\item Let $({\sf ran}(S)\sqsubseteq A,v)\in \Omc$ and $(d, e,\monomial^\Imc)\in S^\Jmc$. Here again, we only need to check the case $(d, e,\monomial^\Imc)=(b^{\Jmc_0}, a^\Imc,\monomial^\Imc)\in S^\Jmc\setminus(S^\Imc\cup S^{\Jmc_0})$. 
Since $\Omc\models (S(b,a),\monomial)$, there must be some $(R(b,a),u)\in\Omc$ such that $\Omc\models (R\sqsubseteq S, \nonomial)$ with $\monomial=u\times\nonomial$. 
Since for every $R_1, R_2$, if $(R_1\sqsubseteq R_2,v_1)\in\Omc$ then $(C_{{\sf ran}(R_1)}\sqsubseteq C_{{\sf ran}(R_2)},v_1)\in\Tmc$, we can show that $\Tmc\models (C_{{\sf ran}(R)}\sqsubseteq C_{{\sf ran}(S)},\nonomial)$. 
Since $C_a^\Imc=\{(a^\Imc, 1^\Imc)\}$ and
$(C_a\sqsubseteq C_{{\sf ran}(R)}, u)\in\Tmc$, it follows that $(a^\Imc, (u\times \nonomial)^\Imc)\in C_{{\sf ran}(S)}^\Imc$. 
Then since $({\sf ran}(S)\sqsubseteq A,v)\in \Omc$, $(C_{{\sf ran}(S)}\sqsubseteq A,v)\in \Tmc$ so $(a^\Imc, (v\times u\times \nonomial)^\Imc)\in A^\Imc$, i.e. $(a^\Imc, (v\times \monomial)^\Imc)\in A^\Imc\subseteq A^\Jmc$. 
It follows that $\Jmc\models ({\sf ran}(S)\sqsubseteq A,v)$.
\item For GCIs of the form $(A\sqsubseteq B, v)$, $(A_1\sqcap A_2\sqsubseteq B, v)$, or $(A\sqsubseteq \exists R, v)$, the result is straightforward.
\item Let $(\exists R.A\sqsubseteq B,v)\in\Omc$ and  $(d,\monomial^\Imc)\in (\exists R.A)^\Jmc$. 
There exists $(d,e,\nonomial^\Imc)\in R^\Jmc$ such that $(e,\onomial^\Imc)\in A^\Jmc$ and $\monomial=\nonomial\times\onomial$. The non trivial case is $(d, e,\nonomial^\Imc)=(b^{\Jmc_0}, a^\Imc,\nonomial^\Imc)\in R^\Jmc\setminus(R^\Imc\cup R^{\Jmc_0})$. 
Since $(a^\Imc,\onomial^\Imc)\in A^\Imc$, it follows from the definition of $\Imc$ that either (i) $\Tmc\models (C_a\sqsubseteq A, \onomial)$ or (ii) $\Tmc\models (\top\sqsubseteq A, \onomial)$. In case (i), by construction of $\Tmc$, there is $(A'(a),u)\in\Omc$ such that $\Omc\models (A'\sqsubseteq A, \onomial')$ and $\onomial=u\times\onomial'$, so that $\Omc\models (A(a), u\times o')$ i.e. $\Omc\models (A(a),o)$. 
In case (ii) it is clear that $\Omc\models (A(a),o)$. 
By construction of $\Jmc$, $\Omc\models (R(b,a),\nonomial)$. Thus $\Omc\models (\exists R.A(b),\nonomial\times o)$, i.e. $\Omc\models (\exists R.A(b),\monomial)$ so $\Omc\models (B(b),v\times\monomial)$. It follows that $(b^{\Jmc_0}, (v\times\monomial)^\Imc)\in B^{\Jmc_0}\subseteq B^\Jmc$.
\end{itemize} 
We conclude that $\Jmc\models\Omc$ and $\Omc\not\models (B(a_0),\monomial_0)$.
\end{proof}

Regarding role assertions, we have the following proposition.

\begin{proposition}
$\Omc\models(R(a_0,b_0),\monomial_0)$ iff $  \Tmc_{R}\models (S\sqsubseteq R,\monomial_0)$ where $S$ is a fresh role name and $\Tmc_{R}=\{(S\sqsubseteq P,\variable)\mid (P(a_0,b_0),\variable)\in\Omc\}\cup\{(P_1\sqsubseteq P_2,\variable)\mid (P_1\sqsubseteq P_2,\variable)\in\Omc\}$.
\end{proposition}
\begin{proof}
Assume that $\Omc\not\models(R(a_0,b_0),\monomial_0)$ and let $\Imc$ be a model of $\Omc$ such that $\Imc\not\models(R(a_0,b_0),\monomial_0)$. Let $\Jmc$ the interpretation that extends $\Imc$ with $S^\Jmc=\{(a_0^\Imc,b_0^\Imc,1^\Imc)\}$. Since $\Imc$ is a model of $\Omc$, it is easy to see that $\Jmc$ is a model of $\Tmc_R$ (in particular, for every $(S\sqsubseteq P,\variable)\in\Tmc_R$, $(P(a_0,b_0),\variable)\in\Omc$ so $(a_0^\Imc,b_0^\Imc,\variable^\Imc)\in P^\Jmc$). 
Moreover, $\Jmc\not\models (S\sqsubseteq R,\monomial_0)$ so $\Tmc_R\not\models (S\sqsubseteq R,\monomial_0)$. 

In the other direction, assume that $\Omc\models(R(a_0,b_0),\monomial_0)$. 
If $(R(a_0,b_0),\monomial_0)\in\Omc$, then $(S\sqsubseteq R,\monomial_0)\in\Tmc_R$. 
Otherwise, we show that there must be $(P_1(a_0,b_0),\variable)\in\Omc$ and $(P_1\sqsubseteq P_2,\variable_1)\in\Omc$, \dots, $(P_{k}\sqsubseteq R,\variable_k)\in\Omc$ with $\representative{\variable\times\variable_1\times\dots\times\variable_k}=\representative{\monomial_0}$, which implies that $(S\sqsubseteq P_1,\variable)$, $(P_1\sqsubseteq P_2,\variable_1)$, \dots, $(P_{k}\sqsubseteq R,\variable_k)$ are in $\Tmc_R$, and that $\Tmc_R\models (S\sqsubseteq R, \monomial_0)$. 
Assume to the contrary that this is not the case 
and let $\Imc$ be a model of $\Omc$. Let 
$\Jmc$ be the interpretation with domain $\Delta^\Jmc:=\Delta^\Imc\cup\{e\}$,
where $e\notin\Delta^\Imc$, and the function $\cdot^\Jmc$ defined for concept/individual names and monomials  as follows: 
\begin{itemize}
\item $a^\Jmc=a^\Imc$ for all $a\in\NI$, $\monomial^\Jmc=\monomial^\Imc$ for all $\monomial\in\NM$, 
\item $A^\Jmc=A^\Imc\cup\{(e,\monomial^\Jmc)\mid (b_0^\Imc,\monomial^\Imc)\in A^\Imc\}$ for all $A\in\NC$.  
\end{itemize}
To define $P^\Jmc$, with $P\in\NR$, we first define $P_{-}$  as the set
of tuples $(a_0^\Imc,b_0^\Imc,\monomial^\Imc)$ such that 
$(a_0^\Imc,b_0^\Imc,\monomial^\Imc)\in P^\Imc$ 
and  there is no $(P_1(a_0,b_0),\variable),$ $(P_1\sqsubseteq P_2,\variable_1), $
$\dots$, $(P_{k}\sqsubseteq P,\variable_k)\in\Omc \text{ such that }\representative{\variable\times\variable_1\times\dots\times\variable_k}
=\representative{\monomial}$. 
We   define 
 $P^\Jmc$ as 
 $$(P^\Imc\cup\{(e,f,\monomial^\Jmc)\mid (b_0^\Imc,f,\monomial^\Imc)\in P^\Imc\}\
 \cup$$$$\{(f,e,\monomial^\Jmc)\mid (f,b_0^\Imc,\monomial^\Imc)\in P^\Imc\})\setminus P_{-}.$$

By construction, $\Jmc\not\models (R(a_0,b_0),\monomial_0)$. 
We show that $\Jmc$ is a model of $\Omc$, contradicting $\Omc\models(R(a_0,b_0),\monomial_0)$. 
Since $\Imc\models\Omc$ and $(P(a_0,b_0),\variable)\in\Omc$ implies that $(a_0^\Imc,b_0^\Imc,\variable^\Imc)\notin P_{-}$, 
then $\Jmc$ satisfies all assertions of $\Omc$. 
It is also easy to check that $\Jmc$ satisfies all GCIs and RRs of $\Omc$ (note that in the case of GCIs of the form $(A\sqsubseteq \exists P, u)$, if $\Imc$ uses an element of the form $(a_0^\Imc, b_0^\Imc,\monomial^\Imc)$ to satisfy such GCI, then $\Jmc$ can use $(a_0^\Jmc, e,\monomial^\Jmc)$). 
Let $(P\sqsubseteq P',u)\in\Omc$ and $(d,d',\monomial^\Jmc)\in P^\Jmc$. 
If $d\neq a_0^\Imc$ or $d'\neq b_0^\Imc$, it is clear that $(d,d',(u\times\monomial)^\Jmc)\in P'^\Jmc$. Otherwise, 
since $(a_0^\Jmc,b_0^\Jmc,\monomial^\Jmc)\in P^\Jmc$, then there exist $P_1(a_0,b_0),\variable), (P_1\sqsubseteq P_2,\variable_1), \dots, (P_{k}\sqsubseteq P,\variable_k)\in\Omc$  such that $\representative{\variable\times\variable_1\times\dots\times\variable_k}=\representative{\monomial}$. Since $(P\sqsubseteq P',u)\in\Omc$, it follows that $(a_0^\Jmc,b_0^\Jmc,(u\times \monomial)^\Jmc)\notin P'_{-}$. Hence $(a_0^\Jmc,b_0^\Jmc,(u\times \monomial)^\Jmc)\in P'^\Jmc$. 
Thus $\Jmc$ is a model of $\Omc$ and $\Omc\not\models(R(a_0,b_0),\monomial_0)$, which closes the argument. 
\end{proof}

\Correctness*
\begin{proof}
(1) Each application of a completion rule adds an annotated axiom of the form $(\alpha,\monomial)$ where $\alpha$ is a GCI, an RI, an RR or an assertion built from the concept, role, and individual names that occur in $\Omc$  and such that $\alpha$ contains at most three such names, and $\monomial$ is a product of at most $k$ variables that occur in $\Omc$. The number of such annotated axioms is polynomial in the size of $\Omc$ and exponential in $k$, so $\Omc_k$ can be computed in polynomial time w.r.t.\ the size of $\Omc$, and in exponential time w.r.t.\ $k$.

\noindent (2) 
($\Leftarrow$) For the ``if" direction, we show a stronger property: for every $k$, every assertion, GCI, RI, or RR $\alpha$ and every momonial $\monomial$, if $(\alpha, \representative{\monomial})\in\Omc_k$, then $\Omc\models (\alpha,\monomial)$. 
The proof is by induction on the number $i$ of completion steps that have been applied before adding $(\alpha, \representative{\monomial})$ in $\Omc_k$. 
In the case $i=0$, $(\alpha, \representative{\monomial})\in\Omc$ so clearly $\Omc\models (\alpha,\monomial)$. 
Assume that the property is true for all $i\leq N$ and consider $(\alpha, \representative{\monomial})\in\Omc_k$ such that $(\alpha, \representative{\monomial})$ has been added in $\Omc_k$ at step $N+1$. We have 17 possible cases depending on which completion rule has been applied in step $N+1$.
\begin{description}
\item[$\CR_0$] $(\alpha, \representative{\monomial})=(X\sqsubseteq X, 1)$ for some $X\in\NC\cup\NR$ or  $(\alpha, \representative{\monomial})=(\top\sqsubseteq \top, 1)$ so $\Omc\models (\alpha, \monomial)$ trivially.
\item[$\CR_1$] $(\alpha, \representative{\monomial})=(R_1\sqsubseteq R_3,\representative{\monomial_1\times \monomial_2})$ and the induction hypothesis applies to $(R_1\sqsubseteq R_2,\monomial_1)$ and  $(R_2\sqsubseteq R_3,\monomial_2)$ so that $\Omc\models (R_1\sqsubseteq R_2,\monomial_1)$ and $\Omc\models (R_2\sqsubseteq R_3,\monomial_2)$. 
It follows from the definition of annotated interpretations that $\Omc\models(R_1\sqsubseteq R_3,\monomial_1\times \monomial_2)$.

\item[$\CR_2$] $(\alpha, \representative{\monomial})=({\sf ran}(R)\sqsubseteq A,\representative{\monomial_1\times \monomial_2})$ and the IH applies to $(R\sqsubseteq S,\monomial_1)$ and $({\sf ran}(S)\sqsubseteq A,\monomial_2)$. Let $\Imc$ be a model of $\Omc$ and $(e,d,\nonomial^\Imc)\in R^\Imc$. Since $\Omc\models (R\sqsubseteq S,\monomial_1)$, then $(e,d,(\nonomial\times\monomial_1)^\Imc)\in S^\Imc$, and since $\Omc\models ({\sf ran}(S)\sqsubseteq A,\monomial_2)$, it follows that $(d,(\nonomial\times\monomial_1\times\monomial_2)^\Imc)\in A^\Imc$. Thus $\Omc\models ({\sf ran}(R)\sqsubseteq A,\monomial_1\times \monomial_2)$.

\item[$\CR_3$, $\CR_4$ and $\CR_5$] are analogous to $\CR_1$. 

\item[$\CR_{6}$] $(\alpha, \representative{\monomial})=(A\sqsubseteq C,\representative{\monomial_1\times \monomial_2\times \monomial_3})$ and the IH applies to $(A\sqsubseteq B_1,\monomial_1)$, $(A\sqsubseteq B_2,\monomial_2)$, and $(B_1\sqcap B_2\sqsubseteq C,\monomial_3)$. It follows from the definition of annotated interpretations and the fact that the semiring is idempotent that $\Omc\models(A\sqsubseteq C,\monomial_1\times \monomial_2\times\monomial_3)$. Indeed, if $\Imc$ is a model of $\Omc$ and $(e,\nonomial^\Imc)\in A^\Imc$, then $(e,(\nonomial\times\monomial_1)^\Imc)\in B_1^\Imc$ and $(e,(\nonomial\times\monomial_2)^\Imc)\in B_2^\Imc$ so $(e,(\nonomial\times\monomial_1\times \nonomial\times\monomial_2)^\Imc)\in (B_1\sqcap B_2)^\Imc$, i.e. $(e, (\nonomial\times\monomial_1\times\monomial_2)^\Imc)\in (B_1\sqcap B_2)^\Imc$. Thus $(e, (\nonomial\times\monomial_1\times\monomial_2\times\monomial_3)^\Imc)\in C^\Imc$.

\item[$\CR_{7}$] $(\alpha, \representative{\monomial})=({\sf ran}(R)\sqsubseteq C,\representative{\monomial_1\times \monomial_2\times \monomial_3\times \monomial_4\times \monomial_5})$ and the IH applies to $({\sf ran}(R)\sqsubseteq B_1,\monomial_1)$, $({\sf ran}(R)\sqsubseteq B_2,\monomial_2)$, $(B_1\sqsubseteq C_1,\monomial_3)$, $(B_2\sqsubseteq C_2,\monomial_4)$, and $(C_1\sqcap C_2\sqsubseteq C,\monomial_5)$. It follows from the definition of annotated interpretations and the fact that the semiring is idempotent that $\Omc\models({\sf ran}(R)\sqsubseteq C,\monomial_1\times \monomial_2\times\monomial_3\times \monomial_4\times \monomial_5)$. Indeed, if $\Imc$ is a model of $\Omc$ and $(e,\nonomial^\Imc)\in {\sf ran}(R)^\Imc$, then $(e,(\nonomial\times\monomial_1)^\Imc)\in B_1^\Imc$ and $(e,(\nonomial\times\monomial_2)^\Imc)\in B_2^\Imc$ so 
$(e,(\nonomial\times\monomial_1\times\monomial_3)^\Imc)\in C_1^\Imc$ and $(e,(\nonomial\times\monomial_2\times\monomial_4)^\Imc)\in C_2^\Imc$. 
Hence $(e,(\nonomial\times\monomial_1\times\monomial_2\times\monomial_3\times\monomial_4)^\Imc)\in (C_1\sqcap C_2)^\Imc$. Thus $(e, (\nonomial\times\monomial_1\times\monomial_2\times\monomial_3\times \monomial_4\times \monomial_5)^\Imc)\in C^\Imc$.

\item[$\CR_{8}$] $(\alpha, \representative{\monomial})=(A\sqsubseteq C,\representative{\monomial_1\times\monomial_2})$ and the induction hypothesis applies to $(A\sqcap B\sqsubseteq C,\monomial_1), (\top\sqsubseteq B, \monomial_2)\in\Omc$. 
Let $\Imc$ be a model of $\Omc$ and $(e,\nonomial^\Imc)\in A^\Imc$. By definition of the interpretation of $\top$, $(e, 1^\Imc)\in\top^\Imc$ so $(e, \monomial_2^\Imc)\in B^\Imc$. It follows that  $(e, (\nonomial\times \monomial_2)^\Imc)\in(A\sqcap B)^\Imc$. Hence $(e, (\nonomial\times \monomial_2\times \monomial_1)^\Imc)\in C^\Imc$. 

\item[$\CR_{9}$] $(\alpha, \representative{\monomial})=(A\sqsubseteq D,\representative{\monomial_1\times \monomial_2\times \monomial_3\times\monomial_4\times\monomial_5})$ and the induction hypothesis applies to $(A\sqsubseteq \exists S,\monomial_1),({\sf ran}(S)\sqsubseteq B,\monomial_2),(B\sqsubseteq C,\monomial_3),(S\sqsubseteq R,\monomial_4),(\exists R.C\sqsubseteq D,\monomial_5)$. 
Let $\Imc$ be a model of $\Omc$ and $(e,\nonomial^\Imc)\in A^\Imc$. 
Since $\Omc\models (A\sqsubseteq \exists S,\monomial_1)$, there exists  $(e,d,(\nonomial\times\monomial_1)^\Imc)\in S^\Imc$. Since 
$\Omc\models ({\sf ran}(S)\sqsubseteq B,\monomial_2)$, it follows that $(d,(\nonomial\times\monomial_1\times\monomial_2)^\Imc)\in B^\Imc$. Since $\Omc\models (B\sqsubseteq C,\monomial_3)$, then $(d,(\nonomial\times\monomial_1\times\monomial_2\times \monomial_3)^\Imc)\in C^\Imc$. 
Since $(e,d,(\nonomial\times\monomial_1)^\Imc)\in S^\Imc$ and $\Omc\models (S\sqsubseteq R,\monomial_4)$, then $(e,d,(\nonomial\times\monomial_1\times\monomial_4)^\Imc)\in R^\Imc$. 
Finally, since $\Omc\models (\exists R.C\sqsubseteq D,\monomial_5)$ and $(e,(\nonomial\times\monomial_1\times\monomial_2\times \monomial_3\times \monomial_4)^\Imc) \in (\exists R.C)^\Imc$ (since the semiring is idempotent), then $(e,(\nonomial\times\monomial_1\times\monomial_2\times \monomial_3\times \monomial_4\times \monomial_5)^\Imc) \in D^\Imc$. Thus $\Omc\models (A\sqsubseteq D,\monomial_1\times \monomial_2\times \monomial_3\times\monomial_4\times \monomial_5)$.

\item[$\CR_{10}$] $(\alpha, \representative{\monomial})=(A\sqsubseteq C,\representative{\monomial_1\times \monomial_2\times\monomial_3})$ and the induction hypothesis applies to $(A\sqsubseteq \exists R,\monomial_1)$, $(\top\sqsubseteq B,\monomial_2)$ and $(\exists R.B\sqsubseteq C,\monomial_3)$. 
Let $\Imc$ be a model of $\Omc$ and $(e,\nonomial^\Imc)\in A^\Imc$. 
Since $\Omc\models (A\sqsubseteq \exists R,\monomial_1)$, there exists $(e,d,(\nonomial\times\monomial_1)^\Imc)\in R^\Imc$. Since 
 $(d,1^\Imc)\in \top^\Imc$ by definition of the interpretation of $\top$,  then $(d,\monomial_2^\Imc)\in B^\Imc$. 
 Hence $(e,(\nonomial\times\monomial_1\times\monomial_2)) \in (\exists R.B)^\Imc$. It follows that $(e,(\nonomial\times\monomial_1\times\monomial_2\times\monomial_3)) \in C^\Imc$. Thus $\Omc\models (A\sqsubseteq C,\monomial_1\times \monomial_2\times\monomial_3)$.

\item[$\CR_{11}$] $(\alpha, \representative{\monomial})=(\top(a),1)$ so by definition of the interpretation of $\top$, $\Omc\models (\alpha, \monomial)$.

\item[$\CR_{12}$] $(\alpha, \representative{\monomial})=(S(a,b),\representative{\monomial_1\times \monomial_2})$  and the  induction hypothesis applies to $(R(a,b),\monomial_1),(R\sqsubseteq S,\monomial_2)$. It is easy to see that it follows from the definition of annotated interpretations that $\Omc\models (S(a,b),\representative{\monomial_1\times \monomial_2})$.

\item[$\CR_{13}$] is similar to $\CR_{12}$.

\item[$\CR_{14}$] $(\alpha, \representative{\monomial})=(B(a),\representative{\monomial_1\times \monomial_2\times \monomial})$ and the  induction hypothesis applies to $(A_1(a),\monomial_1)$, $(A_2(a),\monomial_2)$, and $(A_1\sqcap A_2\sqsubseteq B,\monomial)$. 
It follows from the definition of annotated interpretations that $\Omc\models(B(a),\representative{\monomial_1\times \monomial_2\times \monomial})$.

\item[$\CR_{15}$] $(\alpha, \representative{\monomial})=(B(a),\representative{\monomial_1\times \monomial_2\times \monomial_3})$ and the  induction hypothesis applies to $(R(a,b),\monomial_1)$, $(A(b),\monomial_2)$, and $(\exists R.A\sqsubseteq B,\monomial_3)$. The definition of annotated interpretations implies that $\Omc\models(B(a),\representative{\monomial_1\times \monomial_2\times \monomial_3})$.

\item[$\CR_{16}$] $(\alpha, \representative{\monomial})=(A(b),\representative{\monomial_1\times \monomial_2})$ and the  induction hypothesis applies to $(R(a,b),\monomial_1)$ and $({\sf ran}(R)\sqsubseteq A,\monomial_2)$. For every model $\Imc$ of $\Omc$, $(a^\Imc,b^\Imc,\monomial_1^\Imc)\in R^\Imc$, so $(b^\Imc,(\monomial_1\times \monomial_2)^\Imc)\in A^\Imc$, so $\Omc\models (A(b),\representative{\monomial_1\times \monomial_2}) $.
\end{description}
($\Rightarrow$) We show the ``only if" direction by contrapositive: 
assuming that $(\alpha, \representative{\monomial_k})\notin \Omc_k$, we construct a model $\Imc$ of $\Omc$ such that $\Imc\not\models(\alpha,\monomial_k)$. 
Let $\Imc$ be defined as the union of interpretations $\Imc_0,\Imc_1,\dots$ defined as follows: 
start with $\Delta^{\Imc_0}=\individuals{\Omc}$ and $A^{\Imc_0}=\{(a,\nonomial)\mid (A(a),\nonomial)\in\Omc_k\}$ for every $A\in\NC$ and $R^{\Imc_0}=\{(a,b,\nonomial)\mid (R(a,b),\nonomial)\in\Omc_k\}$ for every $R\in\NR$. Then we complete the interpretation using the following rules, so that $\Imc_{i+1}$ results from applying one of the rules to $\Imc_i$. 
Note that in the following rules, $A$ and $A'$ are concept names (we treat the corresponding GCIs with $\top$ separately to keep the proof simple by limiting the number of cases treated by each rule).
\begin{enumerate}

\item If $(R\sqsubseteq S, \nonomial) \in \Omc_k$, $(d,e, \onomial)\in R^{\Imc_i}$, and $(d,e, \representative{\onomial\times \nonomial})\notin S^{\Imc_i}$, then $S^{\Imc_{i+1}}=S^{\Imc_i}\cup\{(d,e, \representative{\onomial\times \nonomial})\}$.

\item If $({\sf ran}(R)\sqsubseteq B, \nonomial)\in \Omc_k$,  $(d,e, \onomial)\in R^{\Imc_i}$ and  $(e, \representative{\onomial\times \nonomial})\notin B^{\Imc_i}$, then $B^{\Imc_{i+1}}=B^{\Imc_i}\cup\{(e, \representative{\onomial\times \nonomial})\}$.

\item If $(A\sqsubseteq B, \nonomial) \in \Omc_k$, $(d, \onomial)\in A^{\Imc_i}$, and $(d, \representative{\onomial\times \nonomial})\notin B^{\Imc_i}$, then $B^{\Imc_{i+1}}=B^{\Imc_i}\cup\{(d, \representative{\onomial\times \nonomial})\}$. 

\item If $(\top\sqsubseteq B,\nonomial)\in\Omc_k$ or $(\top\sqcap\top \sqsubseteq B,\nonomial)\in\Omc_k$, $e\in\Delta^{\Imc_i}$ and $(e, \nonomial)\notin B^{\Imc_i}$ then $B^{\Imc_{i+1}}=B^{\Imc_i}\cup\{(e, \representative{\nonomial})\}$.

\item If $(A\sqcap A'\sqsubseteq B, \nonomial) \in \Omc_k$, $(d, \onomial)\in A^{\Imc_i}$, $(d, \onomial')\in A'^{\Imc_i}$ and $(d, \representative{\onomial\times \onomial'\times \nonomial})\notin B^{\Imc_i}$, then $B^{\Imc_{i+1}}=B^{\Imc_i}\cup\{(d, \representative{\onomial\times \onomial'\times \nonomial})\}$.

\item If $(A\sqcap \top\sqsubseteq B, \nonomial) \in \Omc_k$ or $(\top\sqcap A\sqsubseteq B, \nonomial) \in \Omc_k$, $(d, \onomial)\in A^{\Imc_i}$, and $(d, \representative{\onomial\times \nonomial})\notin B^{\Imc_i}$, then $B^{\Imc_{i+1}}=B^{\Imc_i}\cup\{(d, \representative{\onomial\times \nonomial})\}$.

\item If $(\exists R.\top\sqsubseteq B, \nonomial)\in \Omc_k$, $(d,e, \onomial)\in R^{\Imc_i}$, and $(d, \representative{\onomial\times \nonomial})\notin B^{\Imc_i}$, then $B^{\Imc_{i+1}}=B^{\Imc_i}\cup\{(d, \representative{\onomial\times \nonomial})\}$.

\item If $(\exists R.A\sqsubseteq B, \nonomial)\in \Omc_k$, $(d,e, \onomial)\in R^{\Imc_i}$, $(e, \onomial')\in A^{\Imc_i}$ and $(d, \representative{\onomial\times \onomial'\times \nonomial})\notin B^{\Imc_i}$, then $B^{\Imc_{i+1}}=B^{\Imc_i}\cup\{(d, \representative{\onomial\times \onomial'\times \nonomial})\}$.

\item If 
$(A\sqsubseteq \exists R, \nonomial)\in \Omc_k$, $(d, \onomial)\in A^{\Imc_i}$, and there is no $e$  such that $(d,e, \representative{\onomial\times \nonomial})\in R^{\Imc_i}$ then $\Delta^{\Imc_{i+1}}=\Delta^{\Imc_i}\cup\{x\}$ and $R^{\Imc_{i+1}}=R^{\Imc_i}\cup\{(d, x, \representative{\onomial\times \nonomial})\}$ where $x\notin\Delta^{\Imc_i}$ is a fresh domain element.

\item If 
$(\top\sqsubseteq \exists R, \nonomial)\in \Omc_k$ $d\in\Delta^{\Imc_i}$and there is no $e$  such that $(d,e, \representative{\nonomial})\in R^{\Imc_i}$ then $\Delta^{\Imc_{i+1}}=\Delta^{\Imc_i}\cup\{x\}$ and $R^{\Imc_{i+1}}=R^{\Imc_i}\cup\{(d, x, \representative{\nonomial})\}$ where $x\notin\Delta^{\Imc_i}$ is a fresh domain element.
\end{enumerate}

Since $\Imc_0$ satisfies all assertions of $\Omc_k$ and since all GCIs, RIs and RRs of $\Omc_k$ are of one of the previous forms 
($\Omc$ and all the axioms introduced by the completion algorithm are in normal form; except for $(\top\sqsubseteq\top, 1)$ which is trivially satisfied anyway), by construction $\Imc$ is a model of $\Omc_k$ and thus also of $\Omc\subseteq \Omc_k$.

We show by induction that for every $i$, for every annotated assertion $(\beta,\monomial)$ built from constants that occur in $\Omc$ and such that $\monomial$ contains at most $k$ variables, if $(\beta, \representative{\monomial})\notin \Omc_k$, then $\Imc_i\not\models (\beta,\monomial)$. 
For $i=0$, $(\beta, \representative{\monomial})\notin \Omc_k$ implies $\Imc_0\not\models (\beta,\monomial)$ by construction of $\Imc_0$ (note that all assertions of $\Omc_k$ are annotated with representatives since the assertions of $\Omc$ are annotated with variables which are their own representatives). 

Assume that the property is true for some $i\geq 0$ and let $(\beta,\monomial)$ be such that $\monomial$ contains at most $k$ variables and $(\beta, \representative{\monomial})\notin \Omc_k$. Assume for a contradiction that $\Imc_{i+1}\models (\beta,\monomial)$. 
Since $\beta$ contains only individual names that occur in $\Omc$, it follows that $\beta$ is of the form $S(a,b)$ or $B(a)$ with $S\in\NR$, $B\in\NC$ and $a,b\in\individuals{\Omc}$. 
Thus, since $\Imc_{i}\not\models (\beta,\monomial)$ by the induction hypothesis, it follows  that $\Imc_{i+1}$ has been obtained from $\Imc_i$ by applying a rule from cases 1 to 8 (since the tuples added by cases 9 and 10 involve at least one domain element $x\in\Delta^\Imc\setminus\individuals{\Omc}$). 
We next show that in every case, $(\beta,\representative{\monomial})\in\Omc_k$. 
\begin{itemize}
\item In case 1, $(\beta,\representative{\monomial})=(S(a,b),\representative{\onomial\times\nonomial})$ and it holds that $\Imc_i\models (R(a,b), \onomial)$. 
Since $\representative{\onomial\times \nonomial}=\representative{\monomial}$, it follows that $ \onomial$ has at most $k$ variables so by induction hypothesis, $(R(a,b), \representative{ \onomial})\in\Omc_k$. Hence, since $(R\sqsubseteq S,  \nonomial)\in\Omc_k$ and $ \representative{\onomial\times  \nonomial}=\representative{\monomial}$ has at most $k$ variables, it follows from the construction of $\Omc_k$ (by $\CR_{12}$) that $(\beta, \representative{\monomial})\in\Omc_k$. 

\item In case 2, $(\beta,\representative{\monomial})=(B(a),\representative{\onomial\times\nonomial})$ and it holds that $\Imc_i\models (R(b,a), \onomial)$ for some $b\in\individuals{\Omc}$ ($b$ cannot be an element from $\Delta^{\Imc_i}\setminus\individuals{\Omc}$ by the form of the rules).  
Since $\representative{\onomial\times \nonomial}=\representative{\monomial}$, it follows that $ \onomial$ has at most $k$ variables and by induction hypothesis, $(R(b,a),\representative{ \onomial})\in\Omc_k$. Hence, since $({\sf ran }(R)\sqsubseteq B, \nonomial)\in\Omc_k$, by $\CR_{16}$, $(\beta, \representative{\monomial})\in\Omc_k$.

\item Case 3 is similar to case 1, based on $\CR_{13}$. 
\item In case 4, $(\beta,\representative{\monomial})=(B(a),\representative{\nonomial})$. Since $a\in\individuals{\Omc}$, by $\CR_{11}$, $(\top(a),1)\in\Omc_k$. Since $(\top\sqsubseteq B,\nonomial)\in\Omc_k$ (in the case $(\top\sqcap\top \sqsubseteq B,\nonomial)\in\Omc_k$, it is also the case that $(\top\sqsubseteq B,\nonomial)\in\Omc_k$ by $\CR_0$ and $\CR_{8}$) and $\representative{ \nonomial}=\representative{\monomial}$ has at most $k$ variables, by $\CR_{13}$, $(\beta,\representative{\monomial})\in\Omc_k$.

\item In case 5,  $(\beta,\representative{\monomial})=(B(a),\representative{\onomial\times\onomial'\times\nonomial})$ and it holds that $\Imc_i\models (A(a),  \onomial)$ and $\Imc_i\models (A'(a),  \onomial')$. Since  $\representative{ \onomial\times  \onomial'\times  \nonomial}=\representative{\monomial}$,  
it follows that $ \onomial$ and $\onomial'$ have at most $k$ variables so by induction hypothesis, $(A(a), \representative{ \onomial})\in\Omc_k$ and $(A'(a), \representative{ \onomial'})\in\Omc_k$. Hence since $(A\sqcap A'\sqsubseteq B,  \nonomial)\in\Omc_k$ and $\representative{ \onomial\times  \onomial'\times  \nonomial}=\representative{\monomial}$ has at most $k$ variables, it follows from the construction of $\Omc_k$ (by $\CR_{14}$) that $(\beta, \representative{\monomial})\in\Omc_k$. 

\item In case 6,  $(\beta,\representative{\monomial})=(B(a),\representative{\onomial\times\nonomial})$ and it holds that $\Imc_i\models (A(a),  \onomial)$. Since  $\representative{ \onomial\times \nonomial}=\representative{\monomial}$,  
it follows that $ \onomial$ has at most $k$ variables so by induction hypothesis, $(A(a), \representative{ \onomial})\in\Omc_k$. Since $(A\sqcap \top\sqsubseteq B,  \nonomial)\in\Omc_k$ and $(\top\sqsubseteq \top,1)\in\Omc_k$ by $\CR_0$, it follows by $\CR_{8}$ that $(A\sqsubseteq B,  \nonomial)\in\Omc_k$. 
Hence, since $\representative{ \onomial\times  \nonomial}=\representative{\monomial}$ has at most $k$ variables, it follows from the construction of $\Omc_k$ (by $\CR_{14}$) that $(\beta, \representative{\monomial})\in\Omc_k$.

\item In case 7, $(\beta,\representative{\monomial})=(B(a),\representative{\onomial\times\nonomial})$ and it holds that there is some $x\in\Delta^{\Imc_i}$ such that $\Imc_i\models (R(a, x),  \onomial)$. 

If $x\in\individuals{\Omc}$, we obtain that $(R(a, x),  \representative{\onomial})\in\Omc_k$ by induction and that $(\top(x),1)\in\Omc_k$ by $\CR_{11}$. Since $(\exists R.\top\sqsubseteq B,\nonomial)\in\Omc_k$ and $\representative{\onomial\times\nonomial}=\representative{\monomial}$ has at most $k$ variables, it follows by $\CR_{15}$ that $(\beta,\representative{\monomial})\in\Omc_k$.

Otherwise,  $x$ is a fresh element that has been introduced during the construction of $\Imc_i$ to satisfy an inclusion of the form $(C\sqsubseteq \exists S, r_0)\in\Omc_k$ (cases 9 or 10). 
It follows that $\Imc_i\models (C(a), r')$ and there exist $(S\sqsubseteq S_1, r_1)$,\dots,$(S_{p-1}\sqsubseteq R, r_p)$ in $\Omc_k$ such that $\representative{r_0\times\dots\times r_p\times r'}= \representative{\onomial}$. 
Note that $\representative{r_0\times\dots\times r_p\times r'}$ is thus a submonomial of $\monomial$, so that it has at most $k$ variables, as well as all its submonomials. 
By $\CR_3$, it follows that $(C\sqsubseteq \exists R, \representative{r_0\times \dots \times r_p})\in\Omc_k$. 
Moreover, we have that $(\exists R.\top\sqsubseteq B,\nonomial)\in\Omc_k$ and (by $\CR_0$) that $(\top\sqsubseteq \top,1)\in\Omc_k$. 
Then by $\CR_{10}$ we have $(C\sqsubseteq B, \representative{r_0\times \dots \times r_p\times\nonomial})\in\Omc_k$. Finally, since $\Imc_i\models (C(a), r')$, by induction, $(C(a), \representative{r'})\in\Omc_k$. 
Note that in the case where $C=\top$, $r'=1$ and $(\top(a),1)\in\Omc_k$ by $\CR_{11}$. 
Hence, by $\CR_{13}$, $(\beta,\representative{\monomial})\in\Omc_k$.

\item
 In case 8, $(\beta,\representative{\monomial})=(B(a),\representative{\onomial\times\onomial'\times\nonomial})$ and it holds that there is some $x\in\Delta^{\Imc_i}$ such that $\Imc_i\models (R(a, x),  \onomial)$ and $\Imc_i\models (A(x),  \onomial')$. 

If $x\in\individuals{\Omc}$, we obtain that $(R(a, x),  \representative{\onomial})\in\Omc_k$ and $(A(x),  \representative{\onomial'})\in\Omc_k$ by induction. Since $(\exists R.A\sqsubseteq B,\nonomial)\in\Omc_k$ and $\representative{\onomial\times\onomial'\times\nonomial}=\representative{\monomial}$ has at most $k$ variables, it follows by $\CR_{15}$ that $(\beta,\representative{\monomial})\in\Omc_k$.

Otherwise, $x$ is a fresh element that has been introduced during the construction of $\Imc_i$, let us say between $\Imc_{j-1}$ and $\Imc_j$ to satisfy an inclusion of the form $(C\sqsubseteq \exists S, r_0)\in\Omc_k$ (cases 9 or 10). 
By construction of $\Imc$, it must thus be the case that $\Imc_{j-1}\models (C(a), r')$ for some submonomial $r'$ of $\onomial$, so that $r'$ has at most $k$ variables. Thus $(C(a), \representative{r'})\in\Omc_k$ by induction. 
Since $(\exists R.A\sqsubseteq B,\nonomial)\in\Omc_k$, by Claim~\ref{claimcompl} (whose proof is deferred at the end), it follows that there exists $(C\sqsubseteq B, r)\in \Omc_k$ such that $\representative{\onomial\times\onomial'\times\nonomial}=\representative{r'\times r}$. 
Since $(C(a), \representative{r'})\in\Omc_k$ and $\representative{\monomial}= \representative{r'\times r}$, it follows that $(\beta, \representative{\monomial})\in\Omc_k$ by $\CR_{13}$.
\end{itemize}

We have thus shown that $(\beta, \representative{\monomial})\in\Omc_k$ regardless the form of the rule applied between $\Imc_i$ and $\Imc_{i+1}$, which contradicts our original assumption.  
Hence $\Imc_{i+1}\not\models (\beta,\monomial)$, and we conclude by induction that $\Imc_i\not\models (\beta,\monomial)$ for every $i\geq 0$.

We conclude that for every annotated assertion $(\beta,\monomial)$ built from constants that occur in $\Omc$ and such that $\monomial$ contains at most $k$ variables, $\Imc\not\models (\beta,\monomial)$. In particular, $\Imc\not\models (\alpha,\monomial_k)$ so $\Omc\not\models (\alpha,\monomial_k)$.

\begin{claim}\label{claimcompl}
For all $x,y\in\Delta^\Imc$, if $\Imc_i\models(R(x,y),\onomial)$, $\Imc_i\models(A(y),\onomial')$, $(\exists R.A\sqsubseteq B,\nonomial)\in\Omc_k$, $\representative{\onomial\times\onomial'\times\nonomial}$ has at most $k$ variables and $y\notin\individuals{\Omc}$ has been introduced between $\Imc_{j-1}$ and $\Imc_j$ ($j\leq i$) to satisfy an inclusion of the form $(C\sqsubseteq \exists S, r_0)$ by applying rules 9 or 10 using $(x,r')\in C^{\Imc_{j-1}}$, then there exists $(C\sqsubseteq B,r)\in\Omc_k$ such that $\representative{\onomial\times\onomial'\times\nonomial}=\representative{r\times r'}$.
\end{claim}
\noindent\emph{Proof of the claim.} 
We show the following property by induction on $l=i-j$: For all $x,y\in\Delta^\Imc$, $R\in \NR$ and $A\in\NC\cup\{\top\}$, if $\Imc_i\models(R(x,y),\onomial)$, $\Imc_i\models(A(y),\onomial')$, and $y\notin\individuals{\Omc}$ has been introduced between $\Imc_{j-1}$ and $\Imc_j$ ($j\leq i$) to satisfy an inclusion of the form $(C\sqsubseteq \exists S, r_0)$ by applying rules 9 or 10 using $(x,r')\in C^{\Imc_{j-1}}$, then there exist
\begin{itemize}
\item $(S\sqsubseteq R, \monomial_{S\sqsubseteq R})\in\Omc_k$ with $\representative{r'\times r_0\times \monomial_{S\sqsubseteq R}}=\representative{o}$, 
\item $(D\sqsubseteq A, \monomial_{D\sqsubseteq A})\in\Omc_k$
\item and either (i) $({\sf ran}(S)\sqsubseteq D, \monomial_{D})\in\Omc_k$ with $\representative{r'\times r_0\times\monomial_{D}\times\monomial_{D\sqsubseteq A}}=\representative{\onomial'}$
or (ii) $(\top\sqsubseteq D, \monomial_{D})\in\Omc_k$ with $\representative{\monomial_{D}\times\monomial_{D\sqsubseteq A}}=\representative{\onomial'}$.
\end{itemize}
Moreover, if $A=\top$ (hence $\onomial'=1$), we are in case (ii) with $D=\top$ and $\monomial_D=\monomial_{D\sqsubseteq A}=1$.

If $(\exists R.A\sqsubseteq B,\nonomial)\in\Omc_k$, it will follow (by $\CR_{9}$ in case (i) or by $\CR_{4}$ and $\CR_{10}$ in case (ii)) that $(C\sqsubseteq B,r)\in\Omc_k$ with $$r=\representative{r_0\times\monomial_{S\sqsubseteq R}\times \monomial_{D}\times\monomial_{D\sqsubseteq A}\times\nonomial},$$ so that $\representative{\onomial\times\onomial'\times\nonomial}=\representative{r\times r'}$ (note that $r$ has indeed at most $k$ variables as it is a submonomial of $\representative{\onomial\times\onomial'\times\nonomial}$).

\noindent\emph{Base case for $A=\top$. } The only possibility to obtain  $(x,y,\representative{\onomial})\in R^{\Imc_i}$ by applying a single rule that introduces $y$ to satisfy an inclusion of the form $(C\sqsubseteq \exists S, r_0)$ is that $S=R$ and $\representative{r'\times r_0}=\representative{\onomial}$. Since $(R\sqsubseteq R,1)\in\Omc_k$, we obtain the property for the case $A=\top$. 

\noindent\emph{Base case for $A\neq\top$. }There is no way to obtain both $(x,y,\representative{\onomial})\in R^\Imc_{i}$ and $(y,\representative{\onomial'})\in A^\Imc_{i}$ by applying a single rule that introduces $y$ so the base case is  $l=1$, i.e., $j=i-1$. 
If $\Imc_i\models(R(x,y),\onomial)$, $\Imc_i\models(A(y),\onomial')$, and $y\notin\individuals{\Omc}$ has been introduced between $\Imc_{i-2}$ and $\Imc_{i-1}$ to satisfy an inclusion of the form $(C\sqsubseteq \exists S, r_0)$ by applying rules 9 or 10 using $(x,r')\in C^{\Imc_{i-1}}$, then by construction of $\Imc$ it must be the case that  $S=R$ and $\representative{\onomial}=\representative{r'\times r_0}$ and either (i) there is $({\sf ran}(R)\sqsubseteq A, \monomial_{{\sf ran}(R)\sqsubseteq A})$ in $\Omc_k$ such that $\representative{r'\times r_0\times\monomial_{{\sf ran}(R)\sqsubseteq A}}=\representative{\onomial'}$ or (ii) there is $(\top\sqsubseteq A, \monomial_{\top\sqsubseteq A})$ in $\Omc_k$ such that $\representative{\monomial_{\top\sqsubseteq A}}=\representative{\onomial'}$. Hence, it is indeed the case that 
\begin{itemize}
\item  $(S\sqsubseteq R, \monomial_{S\sqsubseteq R})=(S\sqsubseteq S,1)\in\Omc_k$ with $\representative{r'\times r_0\times \monomial_{S\sqsubseteq R}}=\representative{o}$, 
\item $(D\sqsubseteq A, \monomial_{D\sqsubseteq A})=(A\sqsubseteq A,1)\in\Omc_k$
\item  either (i) $({\sf ran}(S)\sqsubseteq D, \monomial_{D})=({\sf ran}(R)\sqsubseteq A, \monomial_{{\sf ran}(R)\sqsubseteq A})\in\Omc_k$ with $\representative{r'\times r_0\times\monomial_{D}\times\monomial_{D\sqsubseteq A}}=\representative{\onomial'}$ or (ii) $(\top\sqsubseteq D, \monomial_{D})=(\top\sqsubseteq A, \monomial_{\top\sqsubseteq A})\in\Omc_k$ with $\representative{\monomial_{D}\times\monomial_{D\sqsubseteq A}}=\representative{\onomial'}$.

\end{itemize}

\noindent\emph{Inductive step. } 
Assume that the property is true for every integer up to $l$ and consider the case where $i-j=l+1$. 
Assume that $\Imc_i\models(R(x,y),\onomial)$, $\Imc_i\models(A(y),\onomial')$, and $y\notin\individuals{\Omc}$ has been introduced between $\Imc_{j-1}$ and $\Imc_{j}$ to satisfy an inclusion of the form $(C\sqsubseteq \exists S, r_0)$ by applying rules 9 or 10 using $(x,r')\in C^{\Imc_{j-1}}$, thus adding $(x,y,\representative{r'\times r_0})$ in $S^{\Imc_{j}}$. We make a case analysis on the last rule applied to add $(x,y,\representative{\onomial})$ in $R^{\Imc_{i}}$ or $(y,\representative{\onomial'})$ in $A^{\Imc_{i}}$. 
\begin{itemize}[leftmargin=*]
\item Case where $(x,y,\representative{\onomial})\in R^{\Imc_{i}}$ was added last by applying rule 1 using some $(P\sqsubseteq R,\monomial_{P\sqsubseteq R})\in\Omc_k$: $\Imc_{i-1}\models(P(x,y),\monomial_{P})$ for some $\monomial_{P}$ such that $\representative{\monomial_{P}\times\monomial_{P\sqsubseteq R}}=\representative{\onomial}$. Since $\Imc_{i-1}\models(A(y),\onomial')$, by induction hypothesis, we obtain that there exist
\begin{itemize}[leftmargin=*]
\item $(S\sqsubseteq P, \monomial_{S\sqsubseteq P})\in\Omc_k$ such that 
$\representative{r'\times r_0\times \monomial_{S\sqsubseteq P}}=\representative{\monomial_P}$, so that by $\CR_{1}$, $(S\sqsubseteq R,\representative{\monomial_{S\sqsubseteq P}\times\monomial_{P\sqsubseteq R}})\in\Omc_k$ with $\representative{r'\times r_0\times \monomial_{S\sqsubseteq P}\times\monomial_{P\sqsubseteq R}}=\representative{\onomial}$
\item $(D\sqsubseteq A, \monomial_{D\sqsubseteq A})\in\Omc_k$
\item and either (i) $({\sf ran}(S)\sqsubseteq D, \monomial_{D})\in\Omc_k$ with $\representative{r'\times r_0\times\monomial_{D}\times\monomial_{D\sqsubseteq A}}=\representative{\onomial'}$
or (ii) $(\top\sqsubseteq D, \monomial_{D})\in\Omc_k$ with $\representative{\monomial_{D}\times\monomial_{D\sqsubseteq A}}=\representative{\onomial'}$.
\end{itemize}

\item Case where  $(y,\representative{\onomial'})\in A^{\Imc_{i}}$ was added last by applying rule 2 using some $({\sf ran}(P)\sqsubseteq A,\monomial_{{\sf ran}(P)\sqsubseteq A})\in\Omc_k$:
$\Imc_{i-1}\models(P(x,y),\monomial_{P})$  for some $\monomial_{P}$ such that $\representative{\monomial_{P}\times\monomial_{{\sf ran}(P)\sqsubseteq A}}=\representative{\onomial'}$. 
By induction hypothesis, we obtain that there exists 
$(S\sqsubseteq P, \monomial_{S\sqsubseteq P})\in\Omc_k$ such that $\representative{r'\times r_0\times \monomial_{S\sqsubseteq P}}=\representative{\monomial_P}$. 
By $\CR_2$, we obtain that $({\sf ran}(S)\sqsubseteq A,\representative{\monomial_{{\sf ran}(P)\sqsubseteq A}\times \monomial_{S\sqsubseteq P}})\in\Omc_k$. 
This shows the two last items of the property, if we take $D=A$ and $\monomial_{D\sqsubseteq A}=1$: Indeed, $\representative{r'\times r_0\times\monomial_{{\sf ran}(P)\sqsubseteq A}\times \monomial_{S\sqsubseteq P}}=\representative{\onomial'}$. 

Moreover, since $\Imc_{i-1}\models(R(x,y),\onomial)$,  
by induction hypothesis, we obtain that there exists 
$(S\sqsubseteq R, \monomial_{S\sqsubseteq R})\in\Omc_k$ such that $\representative{r'\times r_0\times \monomial_{S\sqsubseteq R}}=\representative{\onomial}$, which shows the first item of the property. 

\item Case where  $(y,\representative{\onomial'})\in A^{\Imc_{i}}$ was added last by applying rule 3 using some $(E\sqsubseteq A,\monomial_{E\sqsubseteq A})\in\Omc_k$: $\Imc_{i-1}\models(R(x,y),\onomial)$ 
and $\Imc_{i-1}\models(E(y),\monomial_E)$ for some $\monomial_E$ such that $\representative{\monomial_E\times \monomial_{E\sqsubseteq A}}=\representative{\onomial'}$. 
By induction hypothesis there exist
\begin{itemize}[leftmargin=*]
\item $(S\sqsubseteq R, \monomial_{S\sqsubseteq R})\in\Omc_k$ with $\representative{r'\times r_0\times \monomial_{S\sqsubseteq R}}=\representative{o}$, 
\item $(D\sqsubseteq E, \monomial_{D\sqsubseteq E})\in\Omc_k$
\item and either (i) $({\sf ran}(S)\sqsubseteq D, \monomial_{D})\in\Omc_k$ with $\representative{r'\times r_0\times\monomial_{D}\times\monomial_{D\sqsubseteq E}}=\representative{\monomial_E}$
or (ii) $(\top\sqsubseteq D, \monomial_{D})\in\Omc_k$ with $\representative{\monomial_{D}\times\monomial_{D\sqsubseteq E}}=\representative{\monomial_E}$.
\end{itemize}
By $\CR_4$, $(D\sqsubseteq A, \representative{\monomial_{D\sqsubseteq E}\times \monomial_{E\sqsubseteq A}})\in\Omc_k$ and this gives us the desired property.

\item Case where  $(y,\representative{\onomial'})\in A^{\Imc_{i}}$ was added last by applying rule 4 using some $(\top\sqsubseteq A,\representative{\onomial'})\in\Omc_k$ or $(\top\sqcap\top\sqsubseteq A,\representative{\onomial'})\in\Omc_k$: $\Imc_{i-1}\models(R(x,y),\onomial)$ 
so by induction hypothesis there exists $(S\sqsubseteq R, \monomial_{S\sqsubseteq R})\in\Omc_k$ such that $\representative{r'\times r_0\times \monomial_{S\sqsubseteq R}}=\representative{o}$. We obtain the desired property (with case (ii)) by taking $D=\top$,  $\monomial_D=1$ and $\monomial_{D\sqsubseteq A}=\representative{\onomial'}$.

\item Case where  $(y,\representative{\onomial'})\in A^{\Imc_{i}}$ was added last by applying rule 5 using some $(A_1\sqcap A_2\sqsubseteq A,\monomial_{A_1\sqcap A_2\sqsubseteq A})\in\Omc_k$: $\Imc_{i-1}\models(R(x,y),\onomial)$, $\Imc_{i-1}\models(A_1(y),\monomial_1)$ and $\Imc_{i-1}\models(A_2(y),\monomial_2)$ for some $\monomial_1$ and $\monomial_2$ such that $\representative{\monomial_1\times\monomial_2\times \monomial_{A_1\sqcap A_2\sqsubseteq A}}=\representative{\onomial'}$. 
By induction hypothesis there exist
\begin{itemize}[leftmargin=*]
\item $(S\sqsubseteq R, \monomial_{S\sqsubseteq R})\in\Omc_k$ with $\representative{r'\times r_0\times \monomial_{S\sqsubseteq R}}=\representative{o}$, 
\item $(D_i\sqsubseteq A_i, \monomial_{D_i\sqsubseteq A_i})\in\Omc_k$ ($1\leq i\leq 2$)
\item and for each $1\leq i\leq 2$, either (i) $({\sf ran}(S)\sqsubseteq D_i, \monomial_{D_i})\in\Omc_k$ with $\representative{r'\times r_0\times\monomial_{D_i}\times\monomial_{D_i\sqsubseteq A_i}}=\representative{\monomial_i}$
or (ii) $(\top\sqsubseteq D_i, \monomial_{D_i})\in\Omc_k$ with $\representative{\monomial_{D_i}\times\monomial_{D_i\sqsubseteq A_i}}=\representative{\monomial_i}$.
\end{itemize}
If we are in case (i) for $1\leq i\leq 2$, then by $\CR_7$ $({\sf ran}(S)\sqsubseteq A, \monomial_D)\in\Omc_k$ with $\monomial_D=\representative{\monomial_{D_1}\times\monomial_{D_2}\times\monomial_{D_1\sqsubseteq A_1}\times\monomial_{D_2\sqsubseteq A_2}\times\monomial_{A_1\sqcap A_2\sqsubseteq A}}$ and we obtain the property (case (i)) by taking $D=A$ and $\monomial_{D\sqsubseteq A}=1$. 

If we are in case (ii) for $1\leq i\leq 2$, then by $\CR_3$, for each $i$ we have $(\top\sqsubseteq A_i, \representative{\monomial_{D_i}\times\monomial_{D_i\sqsubseteq A_i}})\in\Omc_k$ so by $\CR_6$, we obtain $(\top\sqsubseteq A, \monomial_D)\in\Omc_k$ with $\monomial_D=\representative{\monomial_{D_1}\times\monomial_{D_2}\times\monomial_{D_1\sqsubseteq A_1}\times\monomial_{D_2\sqsubseteq A_2}\times\monomial_{A_1\sqcap A_2\sqsubseteq A}}$ and we obtain the property (case (ii)) by taking $D=A$ and $\monomial_{D\sqsubseteq A}=1$.

If we are in case (i) for $i=1$ and (ii) for $i=2$ (or vice-versa):
\begin{itemize}[leftmargin=*]
\item Since \mbox{$(\top\sqsubseteq D_2, \monomial_{D_2})$ and $(D_2\sqsubseteq A_2, \monomial_{D_2\sqsubseteq A_2})$} are in $\Omc_k$, then $(\top\sqsubseteq A_2, \representative{\monomial_{D_2}\times\monomial_{D_2\sqsubseteq A_2}})\in\Omc_k$ by $\CR_4$.
\item By $\CR_8$, it follows that $\Omc_k$ contains \mbox{$(A_1\sqsubseteq A, \representative{\monomial_{A_1\sqcap A_2\sqsubseteq A}\times\monomial_{D_2}\times\monomial_{D_2\sqsubseteq A_2} })$.}
\item Hence by $\CR_4$, \mbox{$(D_1\sqsubseteq A, \monomial_{D_1\sqsubseteq A})\in\Omc_k$}, where \mbox{$\monomial_{D_1\sqsubseteq A}=\representative{\monomial_{D_1\sqsubseteq A_1}\times\monomial_{A_1\sqcap A_2\sqsubseteq A}\times\monomial_{D_2}\times\monomial_{D_2\sqsubseteq A_2} }$.}
\end{itemize}
We thus obtain the property (case (i)) by taking $D=D_1$: Indeed, $({\sf ran}(S)\sqsubseteq D_1, \monomial_{D_1})\in\Omc_k$ and $\representative{r'\times r_0\times\monomial_{D_1}\times\monomial_{D_1\sqsubseteq A}}=\representative{\onomial'}$.

\item Case where  $(y,\representative{\onomial'})\in A^{\Imc_{i}}$ was added last by applying rule 6 using some $(E\sqcap \top\sqsubseteq A,\monomial_{E\sqcap \top\sqsubseteq A})\in\Omc_k$: $\Imc_{i-1}\models(R(x,y),\onomial)$, and $\Imc_{i-1}\models(E(y),\monomial_E)$ for some $\monomial_E$ such that $\representative{\monomial_E\times \monomial_{E\sqcap \top\sqsubseteq A}}=\representative{\onomial'}$. 
This case is similar to the case of rule 3, but using $\CR_8$ (and $(\top\sqsubseteq\top,1)\in\Omc_k$) instead of $\CR_4$.

\item Case where $(y,\representative{\onomial'})\in A^{\Imc_{i}}$ was added last by applying rule 7 or 8 using some $(\exists P.E\sqsubseteq A,\monomial_{\exists P.E\sqsubseteq A})\in\Omc_k$ ($E$ may be equal to $\top$ or not): 
There exists $z$ such that $\Imc_{i-1}\models(P(y,z),\monomial_{P})$ and $\Imc_{i-1}\models(E(z),\monomial_{E})$ for some $\monomial_{P}$ and $\monomial_{E}$ such that $\representative{\monomial_{P}\times\monomial_{E}\times\monomial_{\exists P.E\sqsubseteq A}}=\representative{\onomial'}$. 
By construction of $\Imc$, $z$ has been introduced after $y$, hence at step $j'$ ($j<j'<i$), to satisfy an inclusion of the form $(C_z\sqsubseteq\exists S_z, r_z)$ with some $(y,r'_z)\in C_z^{\Imc_{j'-1}}$. 
By induction hypothesis, it follows that
\begin{itemize}[leftmargin=*]
\item $(S_z\sqsubseteq P, \monomial_{S_z\sqsubseteq P})\in\Omc_k$ with $\representative{r_z'\times r_z\times \monomial_{S_z\sqsubseteq P}}=\representative{\monomial_P}$, 
\item $(D_z\sqsubseteq E, \monomial_{D_z\sqsubseteq E})\in\Omc_k$
\item and either (i) $({\sf ran}(S_z)\sqsubseteq D_z, \monomial_{D_z})\in\Omc_k$ with $\representative{r_z'\times r_z\times\monomial_{D_z}\times\monomial_{D_z\sqsubseteq E}}=\representative{\monomial_E}$
or (ii) $(\top\sqsubseteq D_z, \monomial_{D_z})\in\Omc_k$ with $\representative{\monomial_{D_z}\times\monomial_{D_z\sqsubseteq E}}=\representative{\monomial_E}$.
\end{itemize}

Since $(\exists P.E\sqsubseteq A,\monomial_{\exists P.E\sqsubseteq A})\in\Omc_k$, it follows (by $\CR_{9}$ in case (i) or by $\CR_{4}$ and $\CR_{10}$ in case (ii)) that $(C_z\sqsubseteq A,\monomial_{C_z\sqsubseteq A})\in\Omc_k$ with $\monomial_{C_z\sqsubseteq A}=\representative{r_z\times\monomial_{S_z\sqsubseteq P}\times \monomial_{D_z}\times\monomial_{D_z\sqsubseteq E}\times\monomial_{\exists P.E\sqsubseteq A}}$.

Moreover, since $\Imc_{i-1}\models(R(x,y),\onomial)$ and $\Imc_{i-1}\models (C_z(y),r'_z)$,  by induction hypothesis, we obtain that there exist 
\begin{itemize}[leftmargin=*]
\item $(S\sqsubseteq R, \monomial_{S\sqsubseteq R})\in\Omc_k$ with $\representative{r'\times r_0\times \monomial_{S\sqsubseteq R}}=\representative{o}$, 
\item $(D\sqsubseteq C_z, \monomial_{D\sqsubseteq C_z})\in\Omc_k$, so that by $\CR_4$ $(D\sqsubseteq A, \monomial_{D\sqsubseteq A})\in\Omc_k$ with $\monomial_{D\sqsubseteq A}=\representative{\monomial_{D\sqsubseteq C_z} \times \monomial_{C_z\sqsubseteq A}}$
\item and either (i) $({\sf ran}(S)\sqsubseteq D, \monomial_{D})\in\Omc_k$ with $\representative{r'\times r_0\times\monomial_{D}\times\monomial_{D\sqsubseteq C_z}}=\representative{r'_z}$
or (ii) $(\top\sqsubseteq D, \monomial_{D})\in\Omc_k$ with $\representative{\monomial_{D}\times\monomial_{D\sqsubseteq C_z}}=\representative{r'_z}$.
\end{itemize}
This shows the property. Indeed, in case (i): 
\begin{align*}
&\representative{r'\times r_0\times\monomial_D\times\monomial_{D\sqsubseteq A}}\\
=&\representative{r'\times r_0\times\monomial_D\times\monomial_{D\sqsubseteq C_z} \times \monomial_{C_z\sqsubseteq A}}\\
=&\representative{r_z' \times \monomial_{C_z\sqsubseteq A}}\\
=&\representative{r'_z\times r_z\times\monomial_{S_z\sqsubseteq P}\times \monomial_{D_z}\times\monomial_{D_z\sqsubseteq E}\times\monomial_{\exists P.E\sqsubseteq A}}\\=&
\representative{\monomial_P\times\monomial_E\times\monomial_{\exists P.E\sqsubseteq A}}
\\=&\representative{\onomial'}
\end{align*}
and similarly in case (ii):
\begin{align*}\representative{\monomial_D\times\monomial_{D\sqsubseteq A}}
=&\representative{\monomial_D\times\monomial_{D\sqsubseteq C_z} \times \monomial_{C_z\sqsubseteq A}}\\
=&\representative{r_z' \times \monomial_{C_z\sqsubseteq A}}\\
=&\representative{\onomial'}
\end{align*}
\end{itemize}
This finishes the proof of the claim.
\end{proof}

\Theorembasic*
\begin{proof}
Let $A_C$ be a fresh concept name.
We show that $\Omc\models (C(a),\monomial)$ iff $\Omc\cup\{(C\sqsubseteq A_C,1)\}\models (A_C(a),\monomial)$.
If $\Omc\models (C(a),\monomial)$, and $\Imc$ is a model of $\Omc\cup\{(C\sqsubseteq A_C,1)\}$, then 
$(a^\Imc,\monomial^\Imc)\in C^\Imc$, so $(a^\Imc,\monomial^\Imc)\in A_C^\Imc$. 
Conversely, if $\Omc\not\models( C(a),\monomial)$, let $\Imc\models\Omc$ be such that $(a^\Imc,\monomial^\Imc)\notin C^\Imc$ and 
let $\Jmc$ extend $\Imc$ with $A_C^\Jmc=C^\Imc$. $\Jmc$ is a model of $\Omc\cup\{(C\sqsubseteq A_C,1)\}$ and
$\Jmc\not\models (A_C(a),\monomial)$. 
The complexity follows from Corollary~\ref{cor:complexity:provmonomial} and Proposition~\ref{prop:pspace}.
\end{proof}
 
\section*{Proofs for Section \ref{sec:combined}}

\newcommand{\roleone}{\ensuremath{R}\xspace}
\newcommand{\roletwo}{\ensuremath{R'}\xspace}
\newcommand{\termone}{\ensuremath{t_1}\xspace}
\newcommand{\termtwo}{\ensuremath{t_2}\xspace}
\newcommand{\termthree}{\ensuremath{t_3}\xspace}
\newcommand{\rttermone}{\ensuremath{t'_1}\xspace}
\newcommand{\rttermtwo}{\ensuremath{t'_2}\xspace}
\newcommand{\rttermthree}{\ensuremath{t'_3}\xspace}

Our proof strategy for dealing with provenance annotated conjunctive
queries  is based on the combined approach, introduced by \citeauthor{LTW:elcqrewriting09} \shortcite{LTW:elcqrewriting09}
for dealing with conjunctive query answering in the \EL family (without provenance). 
The combined approach incorporates consequences of the GCIs into the 
relational instance corresponding to the set of assertions of an ontology. 
In our proof we also incorporate consequences of the GCIs, which are 
now annotated with provenance information, by applying Rules~\ref{r1}-\ref{r3}. 

As in the original combined approach, we construct the canonical model $\Imc_\Omc$ of the ontology \Omc that we want to query. 
The domain $\Delta^{\Imc_\Omc}$ of $\Imc_\Omc$ contains the individual names occurring in \Omc 
and auxiliary elements of the form $\aux{m}{R}{}$ with $ m\in \monomials{\Omc}$ (note that $m=[m]$ by definition of $\monomials{\Omc}$)
and $ R\in\roles{\Omc}$. 
By definition of $\Imc_\Omc$ (in particular, by~\ref{r2}), if $\aux{m}{R}{}$ is connected to 
  $\Imc_\Omc$ then there is some $d\in \Delta^{\Imc_\Omc}$ such that 
  $(d,\aux{[\monomial]}{R}{},[\monomial])\in R^{\Imc_\Omc}$. 
Intuitively, $R$ and $m$ in $\aux{m}{R}{}$ represent the role 
name used to connect $\aux{m}{R}{}$  to $\Imc_\Omc$ (if such connection exists) 
and the provenance information of such connection. 
As we illustrate in Example~\ref{ex:provenancerewriting}, without 
 the provenance information in $\aux{m}{R}{}$, the canonical model $\Imc_\Omc$
 would not entail   annotated queries correctly (not even annotated assertions).

\begin{example}\label{ex:provenancerewriting}
Consider the following annotated ontology.
\begin{align*}
\Omc=\{&
(A(a),v_1),(A \sqsubseteq \exists R,v_2),(\exists R.B\sqsubseteq C, v_3),\\&(A(b),v_4),({\sf ran}(R)\sqsubseteq B,v_5)\}.
\end{align*}
If elements of the form $\aux{m}{R}{}$ did not have provenance information then 
$\Imc_\Omc$ would be as follows.
\[
\Delta^{\Imc_\Omc}=\{a,b,\aux{}{R}{}\}, \quad \Delta^{\Imc_\Omc}_{\sf m}=\NMrep, \text{ and } 
\]
\begin{align*}
A^\Imc=\{&(a,v_1),(b,v_4)\},\\
B^\Imc=\{&(\aux{}{R}{}, v_2\times v_4\times v_5),(\aux{}{R}{}, v_1\times v_2\times v_5)\},\\
C^\Imc=\{&(a,v_1\times v_2\times v_3\times v_5), (a,\prod^5_{i=1}v_i ),\\ 
&(b, v_2\times v_3\times v_4 \times v_5), (b,\prod^5_{i=1}v_i )\},\\ 
R^\Imc=\{&(a,\aux{}{R}{},v_1\times v_2),(b,\aux{}{R}{},v_2\times v_4)\}.
\end{align*}
For $q= (C(a),\prod^5_{i=1}v_i)$, we have that $ \Imc_\Omc\models q$ but $\Omc\not\models q$. 
Observe that adding   information about the domain and range of the connection
as \citeauthor{LTW:elcqrewriting09} \shortcite{LTW:elcqrewriting09} would 
not be a solution in our case. 
\end{example}

As already mentioned, in our rewriting we use $\varphi_1$ and $\varphi_2$, 
contructed in a similar way as 
 \citeauthor{LTW:elcqrewriting09} \shortcite{LTW:elcqrewriting09}. 
 However,   we do not use a formula corresponding to $\varphi_3$ in 
 the mentioned work. The reason is because, in our construction, 
 whenever we have some $d\in\Delta^{\Imc_\Omc}$ with  $(d,\aux{m}{R}{})$ occurring 
 in the extension of a role name $S$, it follows that $\Omc\models (R\sqsubseteq S,n)$ 
 for some 
 $n$ built using variables \NV occurring in \Omc. This different construction
 is used to establish Point~(II) in our proof of Theorem~\ref{thm:queryrewritingcombined} (below)
 in a way that is different from how $\varphi_3$ is used to prove Point~(II) of Theorem~11 by \citeauthor{LTW:elcqrewriting09} \shortcite{LTW:elcqrewriting09}.

Another difference 
 between our construction and the one by \citeauthor{LTW:elcqrewriting09} \shortcite{LTW:elcqrewriting09} 
is that  $\aux{m}{R}{}$ occurs the extension of a concept/role  name 
only if it is (possibly  undirectly) connected  to some individual name  in $\Delta^{\Imc_\Omc}$. 
So we do not need to restrict the domain of $\Imc_\Omc$
to the elements connected to some individual name  in $\Delta^{\Imc_\Omc}$, 
as it happens in the original approach (in their notation this restricted model is denoted $\Imc^r_\Omc$).

\smallskip

We now proceed with the proof of Proposition~\ref{prop:model}. 

\PropositionModel*
\begin{proof}
By definition of $\Imc_\Omc$, before the application of the rules, 
$\Imc_\Omc$
satisfies all axioms of \Omc applicable to elements of $\individuals{\Omc}$, 
except for (annotated) GCIs of the form $(A\sqsubseteq \exists R,\monomial)$.
These are satisfied in $\Imc_\Omc$ with the application of the rules~\ref{r2}, by connecting via the $R^{\Imc_\Omc}$ relation
 elements of $\individuals{\Omc}$ 
to elements of \auxs{\Omc} (with their provenance). 
We now argue that  all axioms of \Omc are satisfied in $\Imc_\Omc$, including 
the part of $\Imc_\Omc$ with elements of \auxs{\Omc}. 
Rule~\ref{r1} satisfies GCIs  $(C\sqsubseteq A,\monomial)$, 
with $C$ of the form $A_1\sqcap A_2$, $\exists R.A, {\sf ran}(R)$, or $A_3$, where 
$A_{(i)}\in\NC\cup\{\top\}$ and $R\in\NR$. 
Finally, we point out that 
RIs 
 are satisfied in $\Imc_\Omc$ by~\ref{r3}. 
\end{proof}

To show Theorem~\ref{thm:combined}, 
we use the following notions. 
Given an interpretation \Imc, let $\individuals{\Omc}^\Imc$ be $\{a^\Imc\mid a\in \individuals{\Omc}\}.$
A \emph{path} in \Imc 
is a finite sequence $d_1 \imonomial_2 R_2 d_2  \dots  \imonomial_{k} R_k d_k $, $k\geq 1$, where 
$d_1\in \individuals{\Omc}^\Imc$ and $(d_i,d_{i+1},\imonomial_{i+1})\in R^\Imc_{i+1}$ for all 
$i < k$. We use ${\sf paths}_\Omc(\Imc)$ to denote the set of all paths 
in \Imc and for all $\upath\in {\sf paths}_\Omc(\Imc)$, and  ${\sf tail}(\upath)$ 
to denote the last element $d_k$ of $\upath$. 
\begin{definition}[Unraveling]
Let \Omc be an annotated \ELHr ontology.
The \Omc-\emph{unraveling} \Jmc of an annotated interpretation \Imc 
is defined as:
\begin{itemize}
\item $\Delta^\Jmc_{\sf m}:=\Delta^\Imc_{\sf m}$,  $\Delta^\Jmc:= {\sf paths}_\Omc(\Imc)$ and $a^\Jmc:=a^\Imc$ for all $a\in\individuals{\Omc}$;
\item $A^\Jmc:= \{(\upath,\imonomial)\mid ({\sf tail}(\upath),\imonomial)\in A^\Imc\}$;
\item $R^\Jmc:= \{(d,d',\imonomial)\mid d,d'\in \individuals{\Omc}^\Imc,  (d,d',\imonomial)\in R^\Imc\}$ 
\item [] $\quad\quad\quad$ $\cup \ \{(\upath,\upath\cdot \imonomial S d',(\monomial\times\nonomial)^\Jmc)\mid
 \upath,\upath\cdot \imonomial S d'\in \Delta^\Jmc,$ 
\item [] $\quad\quad\quad$ $\monomial^\Jmc=\imonomial, 
\Omc \models (S\sqsubseteq R,\nonomial) 
\}$;
\end{itemize}
where $\cdot$ denotes concatenation. 
\end{definition}
We first show that the unraveling $\unraveling$  of $\Imc_\Omc$ 
entails exactly the same (annotated)   queries as \Omc (Theorem~\ref{thm:unravelling}). 
Then, we show that  $\Umc_\Omc$ entails a query $(\query,\polynomial)$ 
iff 
$\Imc_\Omc$  entails $(\queryrewriting,\polynomial)$ 
(Theorem~\ref{thm:queryrewritingcombined}).
Theorems~\ref{thm:unravelling} and~\ref{thm:queryrewritingcombined} together imply 
Theorem~\ref{thm:combined}.

\begin{theorem}\label{thm:unravelling}
Let \Omc be an  annotated \ELHr ontology in normal form and let 
$(\query,\polynomial)$ be an annotated query. 
Then, $\Omc\models (\query,\polynomial) \text{ iff }\  \unraveling\models (\query,\polynomial). $
\end{theorem}
\begin{proof}
The following claim is an easy consequence of the \Omc-unraveling definition. 
\begin{claim}
\label{clm:unraveling}
Let \Imc be an interpretation satisfying \Omc and let \Jmc be the 
\Omc-unraveling of \Imc. Then, for all \ELHr concepts $C$ of the form 
$A_1\sqcap A_2$, $\exists R.A,\exists R, {\sf ran}(R)$, or $A_3$, where 
$A_{(i)}\in\NC\cup\{\top\}$ and $R\in\NR$
all $\upath\in {\sf paths}_\Omc(\Imc)$, 
and all $\imonomial\in \Delta^\Imc_{\sf m}$,
$(\upath,\imonomial)\in C^\Jmc$ iff $({\sf tail}(\upath),\imonomial)\in C^\Imc$. 
\end{claim}

We start observing that, by definition of \Jmc, 
for all $A\in\NC$ and all $\imonomial\in \Delta^\Imc_{\sf m}$,
$(\upath,\imonomial)\in A^\Jmc$ iff $({\sf tail}(\upath),\imonomial)\in A^\Imc$. 
So the claim holds for $C$ of the form $A_1\sqcap A_2$ or $A$.

We now show that for all $R\in\NR$, all $A\in\NC$,
and all $\imonomial\in \Delta^\Imc_{\sf m}$,
$(\upath,\imonomial)\in (\exists R.A)^\Jmc$ iff $({\sf tail}(\upath),\imonomial)\in (\exists R.A)^\Imc$. 
In the following, we use
$\monomial,\monomial_1,\monomial_2\in\NM$ satisfying $\monomial=\monomial_1\times\monomial_2$ and 
$\monomial^\Jmc=\imonomial$, $\monomial^\Jmc_1=\imonomial_1$, $\monomial^\Jmc_2=\imonomial_2$.

\noindent
$(\Rightarrow)$ Assume $(\upath,\imonomial)\in (\exists R.A)^\Jmc$. 
By the semantics of $\exists R.A$, there is $\upath'\in\Delta^\Jmc$ 
such that $(\upath,\upath',\imonomial_1)\in R^\Jmc$ and $(\upath',\imonomial_2)\in A^\Jmc$.
If $\upath'\in \individuals{\Omc}^\Jmc$ then, by definition of \Jmc, 
$\upath\in \individuals{\Omc}^\Jmc$,  $({\sf tail}(\upath),{\sf tail}(\upath'),\imonomial_1)\in R^\Imc$ 
and $({\sf tail}(\upath'),\imonomial_2)\in A^\Imc$. 
That is, $({\sf tail}(\upath),\imonomial)\in (\exists R.A)^\Imc$. 
Otherwise, $\upath'$ is of the form $\upath\cdot \imonomial' S d'$. 
Let $\monomial'\in\NM$ be such that 
$\monomial'^\Jmc=\imonomial'$. 
By definition of \Jmc, $(\upath,\upath',\imonomial')\in S^\Jmc$
and there is $\nonomial\in\NM$ such that 
 $\Omc \models (S\sqsubseteq R,\nonomial)$ and $\monomial_1=\nonomial\times\monomial'$. 
By definition of ${\sf paths}_\Omc(\Imc)$, $({\sf tail}(\upath),{\sf tail}(\upath'),\imonomial')\in S^\Imc$.
Since \Imc satisfies \Omc,  $({\sf tail}(\upath),{\sf tail}(\upath'),\imonomial_1)\in R^\Imc$.
As $(\upath',\imonomial_2)\in A^\Jmc$, we have that $({\sf tail}(\upath'),\imonomial_2)\in A^\Imc$. 
Then, by definition of $\imonomial$ and the semantics of $\exists R.A$, we have that
$({\sf tail}(\upath),\imonomial)\in (\exists R.A)^\Imc$.  
$(\Leftarrow)$~Now assume that $({\sf tail}(\upath),\imonomial)\in (\exists R.A)^\Imc$. 
By the semantics of $\exists R.A$, there is $d\in\Delta^\Imc$ 
such that $({\sf tail}(\upath),d,\imonomial_1)\in R^\Imc$ and $(d,\imonomial_2)\in A^\Imc$.
 If $d\in \individuals{\Omc}^\Imc$ then, by definition of \Jmc, 
$\upath\in \individuals{\Omc}^\Jmc$, $d\in\Delta^\Jmc$,  $(\upath,d,\imonomial_1)\in  R^\Jmc$ and 
$(d,\imonomial_2)\in  A^\Jmc$. 
That is, $(\upath,\imonomial)\in (\exists R.A)^\Jmc$. 
Otherwise, $d\not\in \individuals{\Omc}^\Imc$ and there is $\upath'$ of the form $\upath\cdot \imonomial' S d\in{\sf paths}_\Omc(\Imc)$ 
and $\nonomial,\monomial'\in\NM$ such that 
$\Omc \models (S\sqsubseteq R,\nonomial)$, $\monomial'^\Jmc=\imonomial'$, and $\monomial_1=\nonomial\times\monomial'$. 
By definition of \Jmc, we have that $(\upath,\upath',\imonomial')\in S^\Jmc$ and  $(\upath,\upath',\imonomial_1)\in R^\Jmc$. 
As $(d,\imonomial_2)\in A^\Imc$ and $d={\sf tail}(\upath')$,  $(\upath',\imonomial_2)\in A^\Jmc$.
That is, $(\upath,\imonomial)\in (\exists R.A)^\Jmc$. 

The case  $C=\exists R$ is simpler than 
the case $C=\exists R.A$ and we omit it here.
Thus, it remains to show the case in which $C$ is of the form ${\sf ran}(R)$.

\noindent
$(\Rightarrow)$ Assume $(\upath,\imonomial)\in ({\sf ran}(R))^\Jmc$. 
By the semantics of ${\sf ran}(R)$, there is $\upath'\in\Delta^\Jmc$ 
such that $(\upath',\upath,\imonomial)\in R^\Jmc$.
If $\upath\in \individuals{\Omc}^\Jmc$ then, by definition of \Jmc, 
$\upath'\in \individuals{\Omc}^\Jmc$, and $({\sf tail}(\upath'),{\sf tail}(\upath),\imonomial)\in R^\Imc$. 
That is, $({\sf tail}(\upath),\imonomial)\in ({\sf ran}(R))^\Imc$. 
Otherwise, $\upath$ is of the form $\upath'\cdot \imonomial' S d'$. 
Let $\monomial'\in\NM$ be such that 
$\monomial'^\Jmc=\imonomial'$. 
By definition of \Jmc, $(\upath',\upath,\imonomial')\in S^\Jmc$
and there is $\nonomial\in\NM$ such that 
 $\Omc \models (S\sqsubseteq R,\nonomial)$ and $\monomial=\nonomial\times\monomial'$. 
By definition of ${\sf paths}_\Omc(\Imc)$, $({\sf tail}(\upath'),{\sf tail}(\upath),\imonomial')\in S^\Imc$.
Since \Imc satisfies \Omc,  $({\sf tail}(\upath'),{\sf tail}(\upath),\imonomial)\in R^\Imc$. 
Then, by definition of $\imonomial$ and the semantics of ${\sf ran}(R)$, we have that
$({\sf tail}(\upath),\imonomial)\in ({\sf ran}(R))^\Imc$.  
$(\Leftarrow)$ Now assume that $({\sf tail}(\upath),\imonomial)\in ({\sf ran}(R))^\Imc$. 
By the semantics of ${\sf ran}(R)$, there is $d\in\Delta^\Imc$ 
such that $(d,{\sf tail}(\upath),\imonomial)\in R^\Imc$.
 If $\upath\in \individuals{\Omc}^\Jmc$,
 then, by definition of \Jmc, 
$d\in \individuals{\Omc}^\Jmc$, and 
 $(d,\upath,\imonomial)\in  R^\Jmc$.
That is, $(\upath,\imonomial)\in ({\sf ran}(R))^\Jmc$. 
Otherwise, $\upath$ is of the form $d\cdot \imonomial' S d'\in{\sf paths}_\Omc(\Imc)$
and there are $\nonomial,\monomial'\in\NM$ such that 
$\Omc \models (S\sqsubseteq R,\nonomial)$, $\monomial'^\Jmc=\imonomial'$, and $\monomial=\nonomial\times\monomial'$. 
By definition of \Jmc, we have that $(d,\upath,\imonomial')\in S^\Jmc$ and  $(d,\upath,\imonomial)\in R^\Jmc$. 
That is, $(\upath,\imonomial)\in ({\sf ran}(R))^\Jmc$. 
This finishes the proof of the claim.

\medskip

By Claim~\ref{clm:unraveling} and Proposition~\ref{prop:model} it is straightforward 
to show that \unraveling is a model of \Omc. 
Thus,  $\Omc\models (\query,\polynomial) \text{ implies that }\  \unraveling\models (\query,\polynomial). $
It remains to show the converse direction. Assume that $\Umc_\Omc\models (\query,\polynomial)$ 
and let \Imc be a model of \Omc. 
   For each $\upath  \in \Delta^{\unraveling}$ we write 
  $\dep{\upath}$ to denote 
 the length of 
  the shortest 
 sequence $d_1,\ldots,d_k$
  such that $d_1\in\individuals{\Omc}^{\Imc_\Omc}$, 
  $(d_i,d_{i+1},\imonomial_{i+1})\in 
 \bigcup_{R\in \NR}R^{\Imc_\Omc}$ for all $i< k$ and 
 $d_k = {\sf tail}(\upath)$. 
We define a mapping $\umapping: \Delta^{\unraveling}\rightarrow \Delta^\Imc$ 
such that
\begin{enumerate}[label=(\alph*),leftmargin=*]
\item $\umapping(a) = a^\Imc$ for all $a\in\individuals{\Omc}$;
\item $(\upath,m^{\unraveling})\in C^{\unraveling}$ implies $(\umapping(\upath),m^{\Imc})\in C^{\Imc}$,  
for all $\upath\in \Delta^{\unraveling}$, all $m\in \NM$, and all
\ELHr concepts $C$ of the form $A,\exists R,\exists R.A, \exists R, {\sf ran}(R),A_1\sqcap A_2$ with $A_{(i)}\in\NC$ and $R\in\NR$; 
\item $(\upath,\upath',m^{\unraveling})\in R^{\unraveling}$ implies $(\umapping(\upath),\umapping(\upath'),m^\Imc)\in R^\Imc$, 
for all $\upath,\upath'\in \Delta^{\unraveling}$, all $m\in \NM$, and all 
$R\in\NR$. 
\end{enumerate}
For $\upath\in\Delta^{\unraveling}$, we define $\umapping(\upath)$ by induction on $\dep{\upath}$.
For the case
$\dep{\upath}=1$, 
$\umapping(\upath)$ is as in Point~(a). 
For 
$\upath = d_1 \imonomial_2 R_2 d_2  \dots  \imonomial_{k} R_k d_k\in \Delta^{\unraveling}$ with $\dep{\upath}>1$ 
we argue inductively. Suppose that Points~(a)-(c) hold for paths 
with $\dep{\upath}=k$. 
By definition of $\unraveling$,
for all $S\in\NR$ and all $\nonomial\in\NM$ such that 
$\Omc\models (R_k\sqsubseteq S,\nonomial)$, we have that
 $(d_{k-1},d_k, (\monomial\times\nonomial)^{\Imc_\Omc})\in S^{\Imc_\Omc}$,
with $\monomial^{\Imc_\Omc}=\imonomial_k$, 
and
$d_k = \aux{\imonomial_k}{R_k}{}$.
We also have that  either 
(1) $d_{k-1} = a\in \individuals{\Omc}$ or 
(2) $d_{k-1}=\aux{\imonomial_{k-1}}{R_{k-1}}{}$.

Let $\upath' = d_1 \imonomial_2 R_2 d_2  \dots  \imonomial_{k-1} R_{k-1} d_{k-1}$. 
In Case (1), the definition of $\Imc_\Omc$ implies that 
$\Omc\models (\exists R_k(a),\monomial)$. 
Since \Imc is a model of \Omc, (a) implies that 
there is some $f\in\Delta^\Imc$ with 
$(\delta(\upath'),f,\monomial^\Imc)\in R^\Imc_k$. 
We set 
$\umapping(\upath)$ 
to any such $f\in\Delta^\Imc$. 
In Case~(2), 
since $(\aux{\imonomial_{k-1}}{R_{k-1}}{},\aux{\imonomial_{k}}{R_{k}}{}, \representative{\monomial})\in R^{\Imc_\Omc}_k$, 
the semantics of $\exists R_{k}$ 
implies that
$(\aux{\imonomial_{k-1}}{R_{k-1}}{},\representative{\monomial})\in (\exists R_k)^{\Imc_\Omc}$.
By
  the Claim (of this theorem), 
we have that $(\upath',\representative{\monomial})\in (\exists R_k)^{\unraveling}$. 
By the inductive hypothesis on Point~(b), 
$(\delta(\upath'),\monomial^\Imc)\in (\exists R_k)^\Imc$.   
So there is some 
$f\in\Delta^\Imc$ with 
$(\delta(\upath'),f,\monomial^\Imc)\in R^\Imc_k$. 
We set $\umapping(\upath)$ 
to any such 
$f\in\Delta^\Imc$.

$\delta$ satisfies (a) 
by construction. We now show that $\delta$ satisfies (b) and (c). 
Recall that $\Imc_\Omc$ has been defined as the union  of  
 $\Imc^i_\Omc$, $i\geq 0$.
We denote by $\Umc^i_\Omc$ the unraveling of $\Imc^i_\Omc$.
The proof is by induction on $i$. Observe that, by definition of  $\Imc_\Omc$, 
  $i$  is a finite number, but it can be exponentially larger than the size of \Omc.
Clearly, for $i=0$, we have that:
\begin{enumerate}[label=(\roman*),leftmargin=*]
\item $(\upath,m^{\Umc^i_\Omc})\in C^{\Umc^i_\Omc}$ implies $(\umapping(\upath),m^{\Imc})\in C^{\Imc}$, all $m\in \NM$, and all
  concepts $C$ of the form $A,\exists R,\exists R.A, {\sf ran}(R), A_1\sqcap A_2$ with $A_{(i)}\in\NC$ and $R\in\NR$; 
\item $(\upath,\upath',m^{\Umc^i_\Omc})\in R^{\Umc^i_\Omc}$ implies $(\umapping(\upath),\umapping(\upath'),m^\Imc)\in R^\Imc$, 
for all $\upath,\upath'\in \Delta^{\unraveling}$, all $m\in \NM$, and all 
$R\in\NR$. 
\end{enumerate} 
Suppose that this holds for $i=k$. We now argue for $i=k+1$. We make a case distinction
based on the rule that has been applied on $\Imc^{i}_\Omc$ to define $\Imc^{i+1}_\Omc$. 
\begin{enumerate} 
\item Rule~\ref{r1}:  Point (ii) is again satisfied by the inductive hypothesis. 
We argue about Point (i).
 Assume $(C\sqsubseteq A,\monomial)\in\Omc$ and
there is  $\upath\in\Delta^{\Umc^i_\Omc}$ and $n\in\NM$  such that 
 $(\upath,[n])\in C^{\Umc^{i}_\Omc} $.
By~\ref{r1} and the definition of $\Umc^{i}_\Omc$,  $(\upath,[\nonomial\times \monomial])\in A^{\Umc^{i+1}_\Omc}$.
By the inductive hypothesis, $(\delta(\upath),[n])\in C^{\Imc} $.
Since \Imc is a model of \Omc, $(\delta(\upath),[n\times \monomial])\in A^{\Imc}$.
One can show by induction that (b) holds for all 
 \ELHr concepts.
\item Rule~\ref{r2}:  Assume $(C\sqsubseteq \exists R,\monomial)\in\Omc$
 and there is  
$\nonomial\in\NM$ and $\upath\in\Delta^{\Umc^i_\Omc}$ such that $(\upath,[\nonomial])\in C^{\Umc^{i}_\Omc} $.
By~\ref{r2} and the definition of $\Umc^{i}_\Omc$,  
  $(\upath,\upath',[\monomial\times   \nonomial])\in R^{\Umc^{i}_\Omc}$, where 
  $\upath'$ is the path $\upath\cdot [m\times\nonomial] \ R \ \aux{[m\times n]}{R}{}$. 
  By the inductive hypothesis, $(\delta(\upath),[\nonomial])\in C^{\Imc} $.
  By definition of $\delta$, $\delta(\upath')$ is chosen 
  so that  $(\delta(\upath),\delta(\upath'),\representative{\monomial\times\nonomial})\in R^\Imc$. 
  Such $\delta(\upath')$ exists since \Imc is a model of \Omc. 
\item Rule~\ref{r3}: Assume $(R\sqsubseteq S,\nonomial)\in\Omc$ and there are  
$\upath,\upath'\in\Delta^{\Umc^i_\Omc}$ and $\monomial\in\NM$ 
such that 
 $(\upath,\upath',\representative{\monomial})\in R^{\Umc^{i}_\Omc}$. 
By~\ref{r3} and the definition of $\Umc^{i}_\Omc$,  
  $(\upath,\upath',\representative{\monomial\times \nonomial})\in S^{\Umc^{i}_\Omc}$. 
  By the inductive hypothesis, 
 $(\delta(\upath),\delta(\upath'),[n])\in R^{\Imc} $. 
 Since \Imc is a model of \Omc, we have that $(\delta(\upath),\delta(\upath'),[\monomial\times \nonomial])\in S^{\Imc} $.
 Thus, Point (ii) holds. Using the inductive hypothesis, we can see that (i) also holds.
\end{enumerate}

Thus, properties (b) and (c) hold. This finishes the proof of this theorem. 
\end{proof}
Our next proof is based on Thm.~11 by~\citeauthor{LTW:elcqrewriting09} \shortcite{LTW:elcqrewriting09}.
\begin{theorem}\label{thm:queryrewritingcombined}
Let \Omc be an  annotated \ELHr ontology (in normal form) and let 
$(\query,\polynomial)$ be an annotated query. 
Then, $\unraveling\models (\query,\polynomial) \text{ iff } \ \Imc_\Omc\models (\queryrewriting,\polynomial). $
\end{theorem}
\begin{proof}
Assume $q = \exists \vec{x} .\psi(\vec{x},\vec{a},\vec{p})$. To simplify the notation,  from now on we write 
simply $\psi$ instead of $\psi(\vec{x},\vec{a},\vec{p})$, and the same 
for $\varphi_1,\varphi_2$ 
occurring in the following.
Also, we may write ${\sf term}(q)$ for the set $\vec{x}\cup\vec{a}\cup\vec{p}$.
Recall that $\queryrewriting$ is defined as $\exists \vec{x}.(\psi \wedge \varphi_1 \wedge \varphi_2)$ 
where $\varphi_1,\varphi_2$ 
are quantifier-free. 
Moreover, 
recall that 
${\sf Aux}^\unraveling = (\Delta^{\unraveling}\setminus \individuals{\Omc})\times \{1^{\unraveling}\}$. 

$(\Rightarrow)$ Assume $\unraveling\models (\query,\polynomial)$.
Let $\nu_\unraveling(q)$ be the set of all matches of $q$ in 
$\unraveling$. For each $\pi\in \nu_\unraveling(q)$, 
we define a mapping $\tau_\pi:{\sf term}(q)\rightarrow \Delta^{\Imc_\Omc}\cup \Delta^{\Imc_\Omc}_{\sf m}$
by setting $\tau_\pi(t):= {\sf tail}(\pi(t))$, for all 
$t\in {\sf term}(q)$ with $\pi(t)\in\Delta^\unraveling$, and $\tau_\pi(t):= \pi(t)$, 
for all $t\in {\sf term}(q)$ with $\pi(t)\in\Delta^\unraveling_{\sf m}$. By definition of $\tau_\pi$ and $\unraveling$, 
$\Imc_\Omc\models q$ with $\tau_\pi$ a match of 
$q$ in $\Imc_\Omc$. To show that $\tau_\pi$ is also a match of 
$\exists \vec{x}.(\varphi_1\wedge \varphi_2)$ 
in $\Imc_\Omc$, 
where $\pi\in \nu_\unraveling(q)$,
we use the following claim. 
\medskip

\noindent\textit{Claim 1.}
Let $\termtwo,\rttermtwo\in {\sf term}(q)$ with $\termtwo\sim_{\query} \rttermtwo$ 
  such that $(\pi(\termtwo),1^\unraveling)\in 
{\sf Aux}^{\unraveling}$. Then
\begin{enumerate}[label=(\alph*),leftmargin=*]
\item $\pi(\termtwo) = \pi(\rttermtwo)$; 
\item if $\roleone(\termone,\termtwo,\termthree),\roletwo(\rttermone,\rttermtwo,\rttermthree)\in q$ then $\pi(\termone) = \pi(\rttermone)$. 
\end{enumerate}

\medskip

We start with the proof of Point (a). By definition, $\sim_{\query}$
can be generated by starting with ${\sf id}_q = \{(t,t)\mid t\in {\sf term}(q)\}$
and then exhaustively applying $(\dagger)$ from Section~\ref{sec:combined}  
as a rule, together with the following rule:
\begin{itemize}[label=$(\ddagger)$,leftmargin=*]
\item if $t \sim_{\query} s$ and $s \sim_{\query} t'$
 then $t \sim_{\query} t'$.
\end{itemize} 
We prove Point (a) by induction on the number of rule applications. 
The base case is straightforward. We now make the following case distinction.
\begin{itemize}
\item Rule $(\dagger)$. Let $\roleone(\termone,\termtwo,\termthree),\roletwo(\rttermone,\rttermtwo,\rttermthree)\in q$
and $\termtwo\sim_{\query} \rttermtwo$. Then, $(\dagger)$ adds $(\termone,\rttermone)$ to $\sim_{\query}$. 
Assume that $(\pi(\termone),1^{\unraveling})\in{\sf Aux}^{\unraveling}$. 
 Since $(\pi(\termone),\pi(\termtwo),\pi(\termthree))\in \roleone^\unraveling$,  by construction 
of $\unraveling$,
we get that $(\pi(\termtwo),1^\unraveling)\in {\sf Aux}^\unraveling$. By 
the inductive hypothesis, $\pi(\termtwo)=\pi(\rttermtwo)$. By construction of $\unraveling$ and 
since $(\pi(\termone),\pi(\termtwo),\pi(\termthree))\in \roleone^\unraveling$,
$(\pi(\rttermone),\pi(\termtwo),\pi(\rttermthree))\in {\roletwo}^{\unraveling}$, and $(\pi(\termtwo),1^\unraveling)\in {\sf Aux}^\unraveling$, 
we have that $\pi(\termone)=\pi(\rttermone)$. 
\item Rule $(\ddagger)$. Let $t\sim_{\query} s$ and $s\sim_{\query} t'$. 
Then, $(\ddagger)$ adds $(t,t')$ to $\sim_{\query}$. Assume that $(\pi(t), 
1^\unraveling)\in {\sf Aux}^\unraveling$. By the inductive hypothesis, we have that $\pi(t) = \pi(s)$.
Thus $(\pi(s), 1^\unraveling)\in {\sf Aux}^\unraveling$. Again by 
the inductive hypothesis, $\pi(s)=\pi(t')$ and we obtain $\pi(t)=\pi(t')$. 
\end{itemize}

We come to Point (b). Let $\roleone(\termone,\termtwo,\termthree),\roletwo(\rttermone,\rttermtwo,\rttermthree)\in q$
and $\termtwo\sim_{\query} \rttermtwo$ and assume that $(\pi(\termtwo),1^{\unraveling})\in 
{\sf Aux}^{\unraveling}$. By Point~(a), $\pi(\termtwo)=\pi(\rttermtwo)$. Hence, by 
construction of $\unraveling$ and since $(\pi(\termtwo),1^{\unraveling})\in 
{\sf Aux}^{\unraveling}$, $\pi(\termone)=\pi(\rttermone)$. 
This finishes the proof of Claim~1.

\medskip

We first show that $\tau_\pi$ is  a match of 
$\exists \vec{x}.\varphi_1$ in $\Imc_\Omc$, 
where $\pi\in \nu_\unraveling(q)$.
We want to show that $(\tau_\pi(x),1^{\Imc_\Omc})\not\in {\sf Aux}^{\Imc_\Omc}$, 
for all $x\in 
{\sf Cyc}$. 
By definition of $\tau_\pi$ and construction of $\unraveling$, 
we have that $\tau_\pi(t) = \pi(t)$ for all 
$\pi(t)\in \individuals{\Omc}^{\unraveling} = \individuals{\Omc}^{\Imc_\Omc}$. 
Suppose to the contrary that, for some  $x\in {\sf Cyc}$, 
  we have that $(\tau_\pi(x),1^{\Imc_\Omc})\in {\sf Aux}^{\Imc_\Omc}$. 
  Then there are
  $$R_0(t^0_1,t^0_2,t^0_3),\ldots, R_m(t^m_1,t^m_2,t^m_3), \ldots, 
R_n(t^n_1,t^n_2,t^n_3)\in q,$$ with $n,m\geq 0$,  $x\sim_\query t^j_1$ for some 
$j\leq n$, $t^i_2\sim_\query t^{i+1}_1$ for all $i < n$, 
and $t^n_2\sim_\query t^m_1$. Since $(\tau_\pi(x),1^{\Imc_\Omc})\in 
{\sf Aux}^{\Imc_\Omc}$ and Point (a) of  Claim~1 holds, 
$(\pi(t^j_1),1^{\unraveling})\in 
{\sf Aux}^{\unraveling}$.  Since
$R_j(t^j_1,t^j_2,t^j_3)\in q$, the construction of unravelings
yields that $\pi(t^j_2)=\pi(t^j_1)\cdot \imonomial R_j  d$, for 
some $d\in \Delta^{\Imc_\Omc}$ and 
$\imonomial\in \Delta^{\Imc_\Omc}_{\sf m}$. In particular, 
$\pi(t^j_2)$ is auxiliary. By Point (a) of  Claim~1, 
$\pi(t^j_2)= \pi(t^{j+1}_1)$.
We can continue repeating this argument, setting $t^i_1=t^{i\ {\sf mod}\ n+1}_1$ 
and $t^i_2=t^{i\ {\sf mod}\ n+1}_2$ for all $i > n$.  
In each step, the length of the path $\pi(t^{j+\ell}_1)$ increases. 
This contradicts the fact that $\pi(t^{n+j}_1)=\pi(t^j_1)$ 
(since actually $t^{n+j}_1 = t^j_1$). We have thus
shown that 
 $\tau_\pi$ is  a match of $\exists \vec{x}.\varphi_1$ in $\Imc_\Omc$.

We now show that $\tau_\pi$ is  a match of 
$\exists \vec{x}.\varphi_2$ in $\Imc_\Omc$, 
where $\pi\in \nu_\unraveling(q)$. That is, for all 
$(\{t^1,\ldots,t^k\}, \chi)\in {\sf Fork}_{=}$, 
$(\tau_\pi(t_\chi),1^{\Imc_\Omc})\in {\sf Aux}^{\Imc_\Omc}$ implies 
$\tau_\pi(t^1) =\dots = \tau_\pi(t^k)$. 
Thus, let $(\{t^1_1,\ldots,t^k_1\}, \chi)\in {\sf Fork}_{=}$ and assume that
$(\tau_\pi(t_\chi), 1^{\Imc_\Omc})\in {\sf Aux}^{\Imc_\Omc}$.
Then $(\pi(t_\chi),1^{\unraveling})\in {\sf Aux}^{\unraveling}$ and 
there are terms $t^1_2,\ldots,t^k_2\in \chi$
and role names $R_1,\dots, R_k$ such that $R_i(t^i_1,t^i_2,t^i_3)\in q$
and by 
Point~(b) of  Claim~1, $\pi(t^1_1) = \dots = \pi(t^k_1)$, and thus 
$\tau_\pi(t^1_1) = \dots = \tau_\pi(t^k_1)$. 
This argument holds for all $\pi\in \nu_\unraveling(q)$.
Each monomial in $p$ is associated with $\pi\in\nu_{\unraveling}(q)$.
Since we have shown that $\tau_\pi$ is a match in $\Imc_\Omc$, 
by construction of $\unraveling$, we have that
$p\subseteq \p{\Imc_\Omc}{q}$. 
That is,  $\Imc_\Omc\models (\queryrewriting,\polynomial)$. 
Observe that elements in the extension of ${\sf Aux}^{\Imc_\Omc}$ 
have provenance $1^{\Imc_\Omc}$.

($\Leftarrow$)
Assume $\Imc_\Omc\models (\queryrewriting,\polynomial)$.
Let $\nu_{\Imc_\Omc}(\queryrewriting)$ be the set of all matches of $\queryrewriting$ in 
$\Imc_\Omc$. Similar to the proof strategy above, for ($\Rightarrow$),
we show that 
 each $\pi\in \nu_{\Imc_\Omc}(\queryrewriting)$ can be associated 
 with a match $\tau_\pi$ of $\query$ in $\unraveling$. 
Before constructing $\tau_\pi$, 
we introduce some  notation and prove Claim~2 below.
The \emph{degree} $d(\chi)$ of an equivalence class $\chi$ is the length
 $n\geq 0$ of a longest sequence (if it exists)
 $R_0(t^0_1,t^0_2,t^0_3),\ldots,R_n(t^n_1,t^n_2,t^n_3)\in q$
 such that $t^0_1\in \chi$ and $t^i_2\sim_{\query}t^{i+1}_1$ for all 
 $i<n$. If no longest sequence exists, we set $d(\chi)=\infty$. 
For $t\in {\sf term}(q)$, $\equivalentclass{t}$ denotes the equivalence 
class of $t$ w.r.t. $\sim_\query$.

 \medskip

\noindent\textit{Claim 2.} The following holds.
\begin{enumerate}[label=(\alph*),leftmargin=*]
\item If $(\pi(t),1^{\Imc_\Omc})\in {\sf Aux}^{\Imc_\Omc}$, then 
$d(\equivalentclass{t})< \infty$.
\item If $\termtwo\sim_{\query} \rttermtwo$ and $(\pi(\termtwo),1^{\Imc_\Omc})\in {\sf Aux}^{\Imc_\Omc}$, 
then 
(i) $\pi(\termtwo)=\pi(\rttermtwo)$; 
(ii) if $\roleone(\termone,\termtwo,\termthree),\roletwo(\rttermone,\rttermtwo,\rttermthree)\in \query$ then 
$\pi(\termone) = \pi(\rttermone)$. 
\end{enumerate}

 \medskip

We start with (a). Assume to the contrary of what has to be shown that 
there is $t^0_1$ with $(\pi(t^0_1),1^{\Imc_\Omc})\in {\sf Aux}^{\Imc_\Omc}$ 
and an infinite sequence
$$R_0(t^0_1,t^0_2,t^0_3),R_1(t^1_1,t^1_2,t^1_3), \ldots$$
with $t^i_2\sim_{\query}t^{i+1}_1$ for all $i\geq 0$. By definition
$\pi$ and $\Imc_\Omc$, $(\pi(t^0_1),1^{\Imc_\Omc})\in {\sf Aux}^{\Imc_\Omc}$ 
implies that $t^0_1$ is a variable in $\vec{x}$ (not an individual name). 
As $q$ is finite, there are $m,n$ with $0\leq m \leq n$ such that $t^n_2=t^m_2$. 
It follows that $t^0_1\in {\sf Cyc}$. Hence, $\varphi_1$ contains the 
conjunct $\neg {\sf Aux}(t^0_1,1)$ and we have derived a contradiction 
with $(\pi(t^0_1),1^{\Imc_\Omc})\in {\sf Aux}^{\Imc_\Omc}$.

Now we argue about Point (b).  Using (a), Point (i) of (b) can be proved 
by induction on $n:=d(\equivalentclass{\termtwo})=d(\equivalentclass{\rttermtwo})$. For the induction start, let $\termtwo\sim_{\query} \rttermtwo$
with $(\pi(\termtwo),1^{\Imc_\Omc})\in {\sf Aux}^{\Imc_\Omc}$ and $d(\equivalentclass{\termtwo})=0$. 
By definition of $\sim_{\query}$, we have that $\equivalentclass{\termtwo}=\{\termtwo\}$ and thus $\termtwo = \rttermtwo$. 
Therefore, $\pi(\termtwo)=\pi(\rttermtwo)$ holds. For the induction step, define:
\begin{align*}
&\sim^{(0)}_{\query}:= \ \{(t,t)\mid t\in {\sf term}(q)\}\\
&\sim^{(i+1)}_{\query}:= \ \sim^{(i)}_{\query}\cup\\
&   \{(s,t)\mid \text{ there is } s' \text{ with } s\sim^{(i)}_{\query} s' 
\text{ and } s' \sim^{(i)}_{\query} t\} \ \cup \{(\termone,\rttermone)\mid\\
&     \text{ there are } \roleone(\termone,\termtwo,\termthree), \roletwo(\rttermone,\rttermtwo,\rttermthree) \in \query 
\text{ with } \termtwo\sim^{(i-1)}_{\query} \rttermtwo\}
\end{align*}
for all $i\geq 0$. We can see that $\sim_{\query} = \bigcup_{i\geq 0} \sim^{(i)}_{\query}$. 
We show by induction on $i$ that if $s\sim^{(i)}_{\query}t$, $d([s])=n$, 
and $(\pi(s),1^{\Imc_\Omc})\in {\sf Aux}^{\Imc_\Omc}$, then 
$\pi(s) = \pi(t)$.  The induction start is trivial since $s\sim^{(0)}_{\query}t$
implies $s=t$. 
For the induction step, we distinguish the following two cases.
\begin{itemize}
\item There is $s'$ with $s\sim^{(i)}_{\query} s'$ and 
$s'\sim^{(i)}_{\query} t$.
By the inductive hypothesis on $i$, $\pi(s) = \pi(s')$ and thus 
$(\pi(s'),1^{\Imc_\Omc})\in {\sf Aux}^{\Imc_\Omc}$. Since 
$s\sim^{(i)}_{\query} s'$, we have $\equivalentclass{s} = \equivalentclass{s'}$, and thus
$d(\equivalentclass{s'})=n$. We can thus apply the inductive hypothesis on $i$ once more 
to derive $\pi(s') = \pi(t)$, thus $\pi(s) = \pi(t)$. 
\item There are $\roleone(\termone,\termtwo,\termthree),\roletwo(\rttermone,\rttermtwo,\rttermthree)\in q$ such that $\termtwo \sim^{i-1}_{\query} \rttermtwo$.
By definition of $\Imc_\Omc$, if there is $S(\termone,\termtwo,\termthree)\in q$ and 
$(\pi(\termone),1^{\Imc_\Omc})\in {\sf Aux}^{\Imc_\Omc}$ then 
$(\pi(\termtwo),1^{\Imc_\Omc})\in {\sf Aux}^{\Imc_\Omc}$. By definition of 
degree, $d(\equivalentclass{\termtwo})< d(\equivalentclass{\termone})$. We can thus apply the inductive hypothesis on 
$d(\equivalentclass{\termtwo})$ to obtain $\pi(\termtwo) = \pi(t_{\equivalentclass{\termtwo}})$. Hence, 
$(\pi(t_{\equivalentclass{\termtwo}}),1^{\Imc_\Omc})\in {\sf Aux}^{\Imc_\Omc}$. 
Thus, from  $\varphi_2$ of $\queryrewriting$, we obtain 
$\pi(\termone) = \pi(\rttermone)$.
\end{itemize}

For Point (ii) assume $(\pi(\termtwo),1^{\Imc_\Omc})\in {\sf Aux}^{\Imc_\Omc}$, 
$\roleone(\termone,\termtwo,\termthree),\roletwo(\rttermone,\rttermtwo,\rttermthree)\in q$, and $\termtwo\sim_{\query}\rttermtwo$. By Point~(i), 
$\pi(\termtwo) = \pi(t_{\equivalentclass{\termtwo}})$.  Hence, by the conjunct $\varphi_2$ of $\queryrewriting$, 
$\pi(\termone) = \pi(\rttermone)$.  This finishes the proof of Claim~2.

\medskip

Let $\sim_{\pi}$ be the transitive closure of
\begin{align*}
&  \{(t,t)\mid t\in {\sf term}(q)\}\cup \{(s,t)\in {\sf term}(q)^2\mid s\sim_{\query} t, \\
&  \quad\quad\quad (\pi(s),1^{\Imc_\Omc}),(\pi(t),1^{\Imc_\Omc})\in {\sf Aux}^{\Imc_\Omc}\}\ \cup\\
& \{(\termone,\rttermone) 
\mid  \text{ there are }
 \roleone(\termone,\termtwo,\termthree), \roletwo(\rttermone,\rttermtwo,\rttermthree) \in q\\
& \quad\quad\quad  
 \text{ such that }
 (\pi(\termtwo),1^{\Imc_\Omc})\in {\sf Aux}^{\Imc_\Omc} \text{ and } \termtwo\sim_{\query}\rttermtwo\}.
\end{align*}

\noindent
By Claim~2, we have
\begin{itemize}
\item[$(\dagger)$] $\pi(s) = \pi(t)$ whenever $s \sim_{\pi} t$.
\end{itemize}

One can see that $\sim_{\pi}$ is an equivalence relation because 
it is, by Claim~2, the transitive closure of a symmetric relation. 

Now let the query $q'$ be obtained from $q$ by identifying 
all terms $t,t'\in{\sf term}(q)$ such that $t\sim_{\pi} t'$.
More precisely, choose from each $\sim_{\pi}$-equivalence class 
$\chi$ a fixed term $t_{\chi}\in \chi$ and replace each occurrence 
of an element of $\chi$ in $q$ by $t_\chi$. By 
$(\dagger)$, $\pi$ is a match of $q'$ in $\Imc_\Omc$. 
Next we show the following.
\begin{enumerate}[label=(\Roman*),leftmargin=*]
\item If $x$ is a variable from $\vec{x}$  in $q'$ with $(\pi(x),1^{\Imc_\Omc})\in {\sf Aux}^{\Imc_\Omc}$, 
then there is at most one $t\in {\sf term}(q')$ such that $R(t,x,t')\in q'$, 
for some $R\in \NR$.
\item Assume $x$ is a variable from $\vec{x}$ in $q'$ with $(\pi(x),1^{\Imc_\Omc})\in {\sf Aux}^{\Imc_\Omc}$ 
and there is $t\in {\sf term}(q')$ such that 
$\Gamma = \{R\mid R(t,x,t')\in q'\}\neq \emptyset$. 
Then  there is $S\in\NR$ and $n\in\monomials{\Omc}$ such that, for all $R\in \Gamma$, there is $m\in\monomials{\Omc}$ with 
$\Omc\models (S\sqsubseteq R,m)$, 
$(\pi(t),\pi(x),[n])\in S^{\Imc_\Omc}$, and $\pi(t')=(m\times n)^{\Imc_\Omc}$. 
\item If $q'\supseteq \{R_0(t_0,t_1,s_0),\ldots, R_{n-1}(t_{n-1},t_n,s_{n-1})\}$ with $t_0 = t_n$, then 
$(\pi(t_i), 1^{\Imc_\Omc})\not\in {\sf Aux}^{\Imc_\Omc}$, for all $i \leq n$. 
\end{enumerate}
First for (I). Let $(\pi(x),1^{\Imc_\Omc})\in {\sf Aux}^{\Imc_\Omc}$, 
and let $\roleone(\termone,x,\termthree),\roletwo(\rttermone,x,\rttermthree)\in q'$.  Then there are 
$\roleone(s_1,s_2,s_3),\roletwo(s'_1,s'_2,s'_3)\in q$ such that $s_1\sim_{\pi} \termone$,
 $s'_1\sim_{\pi} \rttermone$, and $s_2\sim_{\pi}x\sim_{\pi}s'_2$. By ($\dagger$), 
 $\pi(s_2)=\pi(x)$, and thus $(\pi(s_2),1^{\Imc_\Omc})\in {\sf Aux}^{\Imc_\Omc}$. 
 By definition of $\sim_\pi$, $s_2 \sim_\pi s'_2$ implies 
 $s_2 \sim_{\query} s'_2$. Summing up, we thus have $\termone\sim_{\pi} \rttermone$. 
 Since both $\termone$ and $\rttermone$ occur in $q'$, we have that $\termone=\rttermone$. 
Point (II) follows from the definition of Rules~\ref{r2} and \ref{r3}.

For (III), let $$q'\supseteq \{R_0(t_0,t_1,t''_0),\ldots, R_{n-1}(t_{n-1},t_n,t''_{n-1})\}$$ with 
$t_0 = t_n$. Then there are 
$$\{R_0(s_0,s'_0,t'_0),\ldots, R_{n-1}(s_{n-1},s'_{n-1},t'_{n-1})\}\subseteq q$$ 
with $s_i \sim_{\pi} t_i$ and $s'_i \sim_{\pi} t_{i+1 \ {\sf mod} \ n}$
for all $i < n$. It follows that $s'_i \sim_{\pi} s_{i+1 \ {\sf mod} \ n}$. 
Now assume  to the contrary that $(\pi(t_i),1^{\Imc_\Omc})\in {\sf Aux}^{\Imc_\Omc}$
for some $i<n$. Since $s_i\sim_{\pi} t_i$, $(\dagger)$ yields $\pi(s_i)=\pi(t_i)$. 
Thus $(\pi(s_i),1^{\Imc_\Omc})\in {\sf Aux}^{\Imc_\Omc}$, which implies 
$s_i$ is a variable from $\vec{x}$ in $q$. Together with  $\sim_{\pi}\subseteq \sim_{\query}$,
this means that $s_i\in {\sf Cyc}$. Thus, $\neg {\sf Aux}(s_i,1)$
is a conjunct of $\varphi_1$ and $(\pi(s_i),1^{\Imc_\Omc})\not\in {\sf Aux}^{\Imc_\Omc}$, 
which is a contradiction. This finishes the proof of (I)-(III). 

\medskip

We inductively define a mapping $\tau_\pi: {\sf term}(q')\rightarrow \Delta^{\unraveling}\cup \Delta^{\unraveling}_{\sf m}$
such that ${\sf tail}(\tau_\pi(t))=\pi(t)$, for all $t\in{\sf term}(q')$ with ${\sf tail}(\tau_\pi(t))\in\Delta^{\Imc_\Omc}$,
and 
$\tau_\pi(t):= \pi(t)$, 
for all $t\in {\sf term}(q)$ with $\pi(t)\in\Delta^{\Imc_\Omc}_{\sf m}$.
For the induction start, we distinguish the following two cases.
\begin{itemize}
\item For all $t\in {\sf term}(q')$ with $(\pi(t),1^{\Imc_\Omc})\not\in {\sf Aux}^{\Imc_\Omc}$, 
set $\tau_\pi(t):= \pi(t)$. This defines $\tau_\pi(t)$ for all 
$t\in {\sf term}(q')\cap \NI$.
\item For all $x$ from the variables $\vec{x}$ in $q'$ with $(\pi(x),1^{\Imc_\Omc})\in {\sf Aux}^{\Imc_\Omc}$
 and such that there is no atom $R(t,x,t')\in q$, do the following. 
 By definition of $\unraveling$ and because each $d\in\Delta^{\Imc_\Omc}$ 
 is reachable from an element of $\individuals{\Omc}^{\Imc_\Omc}$, there is a sequence 
 $d_1,\dots,d_k\in \Delta^{\Imc_\Omc}$ and a sequence $R_1,\ldots,R_{k-1}$ 
 of role names such that $d_1\in \individuals{\Omc}^{\Imc_\Omc}$, 
 $d_k = \pi(x)$, and $(d_i,d_{i+1},\imonomial_{i+1})\in R^{\Imc_\Omc}_{i+1}$ 
 for all $i < k$. Set $\tau_\pi(x):= d_1 \imonomial_2 R_2 d_2 \dots d_{k-1} \imonomial_{k} R_{k} d_k 
 \in \Delta^{\unraveling}$.   
\end{itemize}
For the induction step we proceed as follows. 
Assume $\tau_\pi(x)$ is undefined and there exists $R(t,x,t')\in q'$ with $\tau_\pi(t)$ 
defined. Then (II) yields 
that there is $S\in\NR$ and $n\in\monomials{\Omc}$ such that, for all $R\in \{R\mid R(t,x,t')\in q'\}$,
there is $m\in\monomials{\Omc}$ with $\Omc\models (S\sqsubseteq R,m)$, 
$(\pi(t),\pi(x),[n])\in S^{\Imc_\Omc}$,  and $\pi(t')=(m\times n)^{\Imc_\Omc}$.
Set $\tau_\pi(x):=\tau_\pi(t)\cdot [n] S \pi(x)$. 
Since ${\sf tail}(\tau_\pi(t))$ 
and $(\pi(t),\pi(x),\pi(t'))\in R^{\Imc_\Omc}$, we have $\tau_\pi(x)\in \Delta^{\unraveling}$. 
By (I), the mapping $\tau_\pi$ is well-defined, i.e., the term $t$ 
in the induction step is unique. By (III), $\tau_\pi$ is total, that is, 
$\tau_\pi(t)$ is defined for all $t\in {\sf term}(q')$. To see this, 
suppose that $\tau_\pi(t)$ is undefined. Since $\tau_\pi(t)$ is not 
defined in the induction start, we have 
$(\pi(x),1^{\Imc_\Omc})\in {\sf Aux}^{\Imc_\Omc}$ and there is 
an atom $R(s,t,t')\in q$. Since $\tau_\pi(t)$ is not 
defined in the induction step, $\tau_\pi(t)$ is undefined. One can keep repeating 
this argument. Since $q'$ is finite,   there is 
a sequence $q'\supseteq \{R_1(s_1,s_2,t_2),\ldots,R_{k-1}(s_{k-1},s_k,t_k)\}$
with $s_1 = s_k$ and  $(\pi(s_i),1^{\Imc_\Omc})\in {\sf Aux}^{\Imc_\Omc}$
for all $i \leq k$, contradicting (III).

The constructed $\tau_\pi$ is a match for $q'$ in $\unraveling$. It is 
immediate that for all $A(t,t')\in q'$ we have that $(\tau_\pi(t),\tau_\pi(t'))\in A^{\unraveling}$ 
since ${\sf tail}(\tau_\pi(t))=\pi(t)$ and 
$(\upath, m^{\unraveling})\in A^{\unraveling}$ iff $({\sf tail}(\upath), m^{\Imc_\Omc})\in A^{\Imc_\Omc}$
for all $\upath\in \Delta^{\unraveling}$ and all $m\in\NM$. 
Now let $R(t,t',s)\in q'$. If 
$(\pi(t), 1^{\Imc_\Omc}),(\pi(t'),1^{\Imc_\Omc}) \not\in {\sf Aux}^{\Imc_\Omc}$, 
then $\tau_\pi(t)=\pi(t)$,  $\tau_\pi(t')=\pi(t')$, and 
$(\pi(t),\pi(t'),\pi(s))\in R^{\unraveling}$. If $(\pi(t'),1^{\Imc_\Omc})\in {\sf Aux}^{\Imc_\Omc}$ 
then the construction of $\tau_\pi$ implies that $\tau_\pi(t') = \tau_\pi(t) \cdot \pi(s) R \pi(t')$. 
By definition of $\unraveling$, 
 $(\tau_\pi(t),\tau_\pi(t'),\tau_\pi(s))\in R^{\Imc_\Omc}$. 
The case that $(\pi(t),1^{\Imc_\Omc})\in {\sf Aux}^{\Imc_\Omc}$
and $(\pi(t'),1^{\Imc_\Omc})\not\in {\sf Aux}^{\Imc_\Omc}$
cannot occur by definition of $\Imc_\Omc$. 

Finally, we extend $\tau_\pi$ to a mapping from 
${\sf term}(q)$ to $\Delta^{\unraveling}$ by 
setting $\tau_\pi(t):=\tau_\pi(t')$ 
if $t\in {\sf term}(q)\setminus {\sf term}(q')$ and $t\sim_{\pi} t'$. 
One can verify that $\tau_\pi$ is a match for $q$ in $\unraveling$. 
This argument holds for all $\pi\in \nu_{\Imc_\Omc}(\queryrewriting)$.
Each monomial in $p$ is associated with $\pi\in\nu_{\Imc_\Omc}(\queryrewriting)$.
Since we have shown that $\tau_\pi$ is a match for $q$ in $\unraveling$, 
 we have that
$p\subseteq \p{\unraveling}{q}$. 
That is,  $\unraveling\models (q,\polynomial)$. 
\end{proof}

As mentioned, Theorems~\ref{thm:unravelling} and~\ref{thm:queryrewritingcombined}  imply 
Theorem~\ref{thm:combined}.  

\Theoremcombined*

\end{document}